\documentclass[11pt,a4paper]{article}

\usepackage{graphicx,hyperref,amsmath,natbib,bm,url}
\usepackage[utf8]{inputenc} % include umlaute

\usepackage{color}  
\definecolor{darkblue}{rgb}{0.15,0,0.37}
\definecolor{darkred}{rgb}{0.35,0,0.08}
\definecolor{mygrey}{rgb}{0.85,0.85,0.85} % make grey for the table
\RequirePackage{natbib}

\usepackage{tfrupee}
\usepackage{microtype,todonotes}
\usepackage[australian]{babel} % change date to to European format
\usepackage[a4paper,text={14.5cm, 25.2cm},centering]{geometry} % Change page size
\hypersetup{pdfpagemode=UseNone} % un-comment this if you don't want to show bookmarks open when opening the pdf

\usepackage{mathpazo} % change font to something close to Palatino
\usepackage{caption}
\usepackage{subcaption}
\usepackage{longtable, booktabs, tabularx} % for nice tables
\usepackage{caption,fixltx2e}  % for nice tables
\usepackage[flushleft]{threeparttable}  % for nice tables

\usepackage{multicol}
\usepackage{etoolbox}
\usepackage{relsize}
\usepackage[onehalfspacing]{setspace} % Line spacing

\usepackage[marginal]{footmisc} % footnotes not indented
\usepackage{pdflscape} % for landscape figures

\usepackage[capposition=top]{floatrow} % For notes below figure

\usepackage{color, colortbl}
\usepackage{multirow}
\newcommand\blfootnote[1]{%
  \begingroup
  \renewcommand\thefootnote{}\footnote{#1}%
  \addtocounter{footnote}{-1}%
  \endgroup
}

\usepackage{fancyhdr}
\usepackage{lastpage}
\usepackage{psfrag}

\usepackage{sepfootnotes}

\usepackage{hyperref}
\hypersetup{
	colorlinks=true,
	hypertexnames=false,
	hyperfootnotes=false,
	colorlinks,		
	linkcolor={red},
	citecolor={blue},
	urlcolor={black},
	pdfstartview={FitV},
	unicode,
	breaklinks=true
}

\newcommand\independent{\protect\mathpalette{\protect\independenT}{\perp}}
\def\independenT#1#2{\mathrel{\rlap{$#1#2$}\mkern2mu{#1#2}}}

\makeatletter
\newcommand*{\indep}{%
	\mathbin{%
		\mathpalette{\@indep}{}%
	}%
}
\newcommand*{\nindep}{%
	\mathbin{%                   % The final symbol is a binary math operator
		\mathpalette{\@indep}{\not}% \mathpalette helps for the adaptation
	}%
}
\newcommand*{\@indep}[2]{%
	\sbox0{$#1\perp\m@th$}%        box 0 contains \perp symbol
	\sbox2{$#1=$}%                 box 2 for the height of =
	\sbox4{$#1\vcenter{}$}%        box 4 for the height of the math axis
	\rlap{\copy0}%                 first \perp
	\dimen@=\dimexpr\ht2-\ht4-.2pt\relax
	\kern\dimen@
	{#2}%
	\kern\dimen@
	\copy0 %                       second \perp
} 

\numberwithin{equation}{section}

\newtheorem{lem}{Lemma}
\newtheorem{ass}{AS}
\newtheorem{asscf}{ACF}
\newtheorem{assrc}{RC}
\newtheorem{remark}{Remark}

\newtheorem{innercustomthm}{Lemma}
\newenvironment{customlem}[1]
{\renewcommand\theinnercustomthm{#1}\innercustomthm}
{\endinnercustomthm}

\newtheorem{prop}{Proposition}
\newtheorem{aprop}{Proposition}
\newtheorem{profl}{Proof of Lemma}

\newtheorem{proof}{Proof}
\newtheorem{prof}{Proof of Proposition}

\newtheorem{theo}{Theorem} 
\newtheorem{proft}{Proof of Theorem}
\newtheorem{atheo}{Theorem} 

\newcommand{\pr}{\mathrm{Pr}}
\newcommand{\E}{\mathrm{E}}
\newcommand{\N}{\mathrm{N}}

\newcommand{\rank}{\mathrm{rank}}
\newcommand{\var}{\mathrm{Var}}
\newcommand{\cov}{\mathrm{Cov}}

\DeclareMathAlphabet{\mathbbold}{U}{bbold}{m}{n}

\include{commands}

\begin{document}

%\vspace*{0.5cm}

%\thispagestyle{empty} % page number not shown on first page

\begin{titlepage}
	\title{A Control Function Approach to Estimate Panel Data Binary Response Model\thanks{I would like to thank the anonymous referee and the associate editor whose comments and suggestions have helped to substantially improve the paper. Thanks are due to seminar participants at the The Bank of Estonia, the $ 10^{th} $ Nordic Econometric Meeting (Stockholm), the Institute of Mathematics and Statistics (University of Tartu), and the Inaugural Baltic Economic Conference (Vilnius) for the same. I would especially like to thank Soham Sahoo for helping me with the data. All remaining errors are mine.}}
	\author{Amaresh K Tiwari\thanks{University of Tartu, \newline School of Economics and Business Administration, \newline
			Narva maantee 18, 51009 Tartu, Estonia \newline Phone: +3727376374, \newline Email: amaresh.kr.tiwari@gmail.com \& amaresh.tiwari@.ut.ee } }
	\date{}
	\maketitle
	\begin{abstract}
		\noindent 	We propose a new control function (CF) method to estimate a binary response model in a triangular system with multiple unobserved heterogeneities The CFs are the expected values of the heterogeneity terms in the reduced form equations conditional on the histories of the endogenous and  the exogenous variables.  The method requires weaker restrictions compared to CF methods with similar imposed structures. If the support of endogenous regressors is large, average partial effects are point-identified even when instruments are discrete. Bounds are provided when the support assumption is violated. An application and Monte Carlo experiments compare several alternative methods with ours.\\
		\vspace{0in}
		\vspace{0in}
		
		\bigskip
	\end{abstract}
	\setcounter{page}{0}
	\thispagestyle{empty}
\end{titlepage}

\iffalse
\begin{center}
{\LARGE A Control Function Approach to Estimate Panel Data Binary Response Model }\\[0.4cm]
 {\large Amaresh K  Tiwari\thanks{University of Tartu, \newline School of Economics and Business Administration, \newline
 		Narva maantee 18, 51009 Tartu, Estonia \newline Phone: +3727376374, \newline Email: amaresh.kr.tiwari@gmail.com \& amaresh.tiwari@.ut.ee }}\\[0.2cm]
%  {\large \href{mailto:firstname.surname@uni-bonn.de}{firstname.surname@uni-bonn.de}}\\[0.2cm]
% {\large \today}\\
\end{center}

\blfootnote{\textit{\textbullet \phantom{a} We wish to thank X1, X2 and X3.}}

\vspace*{-0.1cm}
\begin{center}
\begin{minipage}{0.9\textwidth}
\textbf{Abstract.}	We propose a new control function (CF) method to estimate a binary response model in a triangular system with multiple unobserved heterogeneities The CFs are the expected values of the heterogeneity terms in the reduced form equations conditional on the histories of the endogenous and  the exogenous variables.  The method requires weaker restrictions compared to CF methods with similar imposed structures. If the support of endogenous regressors is large, average partial effects are point-identified even when instruments are discrete. Bounds are provided when the support assumption is violated. An application and Monte Carlo experiments compare several alternative methods with ours.\\[0.3cm]
\begin{footnotesize}
\end{footnotesize}
\end{minipage}
\end{center}

\fi

	\def\onepc{$^{\ast\ast\ast}$} \def\fivepc{$^{\ast\ast}$}
\def\tenpc{$^{\ast}$}
\def\legend{\multicolumn{3}{l}{\footnotesize{Significance levels
			:\hspace{1em} $\ast$ : 10\% \hspace{1em} $\ast\ast$ : 5\%
			\hspace{1em} $\ast\ast\ast$ : 1\% \normalsize}}}

\defcitealias{papke:2008}{PW}
\defcitealias{wooldridge:2003}{Problem 12.11}
\defcitealias{baland:2000}{BR}
\defcitealias{basu:2010}{BDD}
\defcitealias{bhalotra:2003}{BHy}
\defcitealias{hoderlein:2012}{HW}
\defcitealias{blundell:2004}{BP}
\defcitealias{Cockburn:2007}{CD}
\defcitealias{altonji:2005}{AM}
\defcitealias{bester1:2009}{BH}
\defcitealias{semykina:2018}{SW}
\defcitealias{imbens:2009}{IN}
\defcitealias{val:2011}{FV}
\defcitealias{james:1961}{James-Stein}

	\section{Introduction}

\citet{chamberlain:2010} and \citet{arellano:2011} point out that when panel data outcomes are discrete, serious identification issues arise when covariates are correlated with unobserved heterogeneity. \citeauthor{chamberlain:2010} shows that for binary choice model
with fixed $ T $, quantities of interest such as Average Partial Effect (APE) may not be point identified, or may not possess a $\surd N$ consistent estimator. Notwithstanding this underidentification result, various methods have been proposed to estimate the structural measures of interest. 

\citet{arellano:2011} provide an overview, and categorize, of some of the methods developed to estimate the quantities of interest. These include the fixed effect (FE) approach that treat heterogeneity or individual effects as parameters to be estimated, where several approaches have been proposed to correct for bias due the incidental parameter problem. \citet{wooldridge:2019}, points out that the FE approach, though promising, suffers from a number of shortcomings. First, the number of time periods needed for the bias adjustments to work well is often greater than is available in many applications. Secondly, the recent bias adjustments methods require the assumptions of stationarity and weak dependence; in some cases, the strong assumption of serial independence (conditional on the heterogeneity) is maintained. However, in empirical work dealing with linear models, it has been found that idiosyncratic errors exhibit serial dependence. Also, ``the requirement of stationarity is strong and has substantive restrictions as it rules out staples in empirical work such as including separate year effects, which can be estimated very precisely given a large cross section."

There is another class of models that acknowledges the fact that many nonlinear panel data models are not point identified at fixed $T$ and consequently discuss set identification (bound analysis) for certain quantiles of interest such as the marginal effects. These papers show that the bounds become tighter as the number of time periods, $T$, increases. However, the methods in most of these papers are still limited to discrete covariates. Moreover, these papers and papers utilizing FE approach assume that conditional on unobserved heterogeneity all covariates are exogenous or predetermined; this, as argued in \citet{hoderlein:2012} (henceforth HW), may not always hold true.

In this paper, we relax the assumption of conditional exogeneity to allow for endogenous covariates that are continuous, and develop a control function method to identify and estimate structural measures such as the Average Structural Function (ASF) and the APE while accounting for endogeneity and heterogeneity in a triangular system.

Some of the papers that have developed control function method to study binary or fractional response outcomes are \citet{rivers:1988}, \citet{blundell:2004} (BP), \citet{papke:2008} (PW), \cite{rothe:2009} and \citet{semykina:2018} (SW).  A partial list of papers that have studied nonparametric control function estimation of nonseparable, including binary response, models are \citet{altonji:2005} (AM), \citet{Florens:2008}, \citet{imbens:2009} (IN), and \citetalias{hoderlein:2012}, where the focus is on estimating heterogeneous effect of endogenous treatment. 

To our knowledge, the papers that allow certain unobserved heterogeneities to be correlated with the exogenous variables while developing a control function method are: \citetalias{papke:2008}, \citet{val:2011} (FV), \citetalias{hoderlein:2012}, \cite{kim:2017}, and \citetalias{semykina:2018}.  While \citetalias{papke:2008} and \citetalias{semykina:2018} use the framework of correlated random effects to account for correlation between individual specific effects and the exogenous variables, \citetalias{val:2011} consider fixed effect estimation in both the stages, where the control variable is based on estimates of the fixed effects in the reduced form equation. \citetalias{val:2011}'s method though more general, requires large $ T $ to correct the bias due to incidental parameters problem.	\citetalias{hoderlein:2012} develop a generalized version of differencing -- which differences out the fixed effects -- to identify local average responses in a nonseparable and binary response models. However, identification in \citetalias{hoderlein:2012} requires that individuals do not experience a change in covariates over time; this requirement that the support of all regressors overlap over time could be hard to satisfy -- e.g., it rules out time trends and time dummies. \citeauthor{kim:2017} exploit restrictions in the conditional moment of unobserved heterogeneity given instruments to develop ``generalized control function." We allow for the correlation between the unobserved individual effects and the instruments in a manner similar to \citetalias{papke:2008}'s and \citetalias{semykina:2018}'s. 

Typically, in a simultaneous triangular system, unobserved heterogeneity in the reduced form equations  is assumed to be scalar, where the identifying assumption is that conditional on these scalar time-varying heterogeneity/errors or its CDF, which are identified, all covariates are independent of the heterogeneity in the structural equation. However, we know that economic models suggest heterogeneity in tastes, technologies, abilities, etc. that are unobserved. Also, some of these unobserved heterogeneity might as well be multidimensional. \cite{kasy:2011} shows that for the existing control function methods, identification fails when the reduced form equations have multiple unobserved heterogeneities. 

The exceptions to our knowledge are \citetalias{papke:2008}, \citetalias{val:2011}, \citetalias{hoderlein:2012} and  \citetalias{semykina:2018}, who consider panel data where multiple heterogeneities constitute of time invariant random effects and idiosyncratic errors.  While the imposed structures in \citetalias{papke:2008} and \citetalias{semykina:2018} are similar to ours, they make the traditional control function assumption, and so their control function is scalar, whereas our control function is vector valued, whose dimension depend on the dimension of unobserved heterogeneity and the number of endogenous variables. \citetalias{hoderlein:2012}'s specification of the triangular system does not nest ours and \citetalias{val:2011}'s fixed effects method requires long panels. 

We propose that the expected values of the heterogeneity terms of the reduced form equations conditional on the history of endogenous variables, $X_{i}\equiv(\pmb{x}_{i1},\ldots, \pmb{x}_{iT})$, and the same of the exogenous variables, $Z_{i}\equiv(\pmb{z}_{i1},\ldots, \pmb{z}_{iT})$ be used as control functions. The proposed control functions are identified when the distributions of the heterogeneity terms are specified. We argue that (1) for triangular systems with set-ups similar to ours, these control functions imply a weaker restriction than the commonly made control function assumptions, and (2) the traditionally used control functions may not provide consistent estimates in a panel data setting such as ours.

Our method, while being simple, makes a number of contributions to the literature. First, we allow for multiple heterogeneities, albeit with restrictions, in the triangular system, where most papers, adopting the control function approach to handle endogeneity, do not. Secondly, when the support of the endogenous variables is large, ASF or the APEs are point-identified even when the instruments have a small support.  We exploit panel data with repeated observations of the same unit for the purpose of point-identification when support requirement is met. Sharp bounds on the ASF and the APEs are provided when the support assumption is not satisfied. Thirdly, the method accounts for multiple endogenous variables, all of which are determined simultaneously, whereas most papers on control function consider a single endogenous variable. Finally, our model retains the attractive features of \citetalias{papke:2008}'s, where no assumptions are made on the serial dependence among the outcome variable. 

Using data on India, we estimate the causal effects of household income and wealth on the incidence of child labor. We find a strong effect of correcting for endogeneity, and show that the standard parametric models give a misleading picture of the causal effect of income and wealth on child labor. 

The rest of the paper is organized as follows. In section 2 we introduce the model and discuss identification and estimation of structural measures of interests for a discrete response model in a triangular system with random effects. In section 3 we discuss the results of the Monte Carlo experiments, where we compare our estimator with some of the existing methods for panel data binary response model with imposed structures similar to ours. Section 4 contains the application of the proposed estimator to study income and wealth effects on the incidence of child labor. And finally in section 5 we conclude. The following have been put in the appendix: proofs of the lemmas, propositions, and theorems (Appendix A), generalized estimating equation (GEE) estimation of probit conditional mean function (Appendix B), extension of the random effect model in the main text to allow for random coefficients (Appendix C), large sample properties of the estimator (Appendix D), other technical details (Appendix E).

\section{Model Specification and Identification and Estimation of Structural Measures}
Consider the following binary choice model in a triangular set-up:
\begin{align}\label{eq:1}
y_{it}=1\{y^{*}_{it} = (\pmb{w}^{\prime}_{it}, \pmb{x}_{it}^{\prime})\pmb{\varphi} + \theta_{i} + \zeta_{it}>0\},
\end{align}
where  $1\{.\}$ is an indicator function that takes value 1 if the argument in the parenthesis holds true and 0 otherwise. In (\ref{eq:1}),  $\theta_{i}$ is the unobserved time invariant individual effect and $\zeta_{it}$ is the idiosyncratic error component. The variables, $\pmb{x}_{it}$, are endogenous in the sense that  $\zeta_{it} \nindep \pmb{x}_{it} |\theta_{i}$; whereas most papers studying panel data binary choice model assume that  $\zeta_{it} \indep \pmb{x}_{it} |\theta_{i}$. We assume that each of the endogenous variables are continuous and have a large support. The dimension of $\pmb{x}_{it}$ is $d_{x}$ and the dimension of the exogenous variables, $ \pmb{w}_{it} $, is $ d_{w} $.

The reduced form in the triangular system, which is estimated in the first stage, is a system of $d_{x}$ linear equations, 
\begin{align}\label{eq:2}
\pmb{x}_{it} = \pi\pmb{z}_{it} +
\pmb{\alpha}_{i} + \pmb{\epsilon}_{it}.
\end{align}
In (\ref{eq:2}),  $\pi$ has a row dimension of $d_{x}$, $\pmb{\alpha}_{i}\equiv (\alpha_{i1}, \ldots, \alpha_{id_{x}})^{\prime}$ is the $ (d_{x}\times1) $ vector of unobserved time invariant individual effects, $\pmb{\epsilon}_{it}\equiv (\epsilon_{it1}, \ldots, \epsilon_{itd_{x}})^{\prime}$ is the $ (d_{x}\times1) $ vector of idiosyncratic error terms, and $\pmb{z}_{it}\equiv (\pmb{w}_{it}^{\prime},\tilde{\pmb{z}}_{it}^{\prime})^{\prime}$ is of dimension $ d_{z} $. The dimension of the vector of instruments, $\tilde{\pmb{z}}_{it}$, is greater than or equal to the dimension of $\pmb{x}_{it}$. Such exclusion restriction, where $\tilde{\pmb{z}}_{it}$ appears in the reduced form but not in the structural, are justified on economic grounds. 

Since the exogenous variables, $ \pmb{w}_{it} $, have no bearing on the identification results obtained in the paper, to ease notations we suppress it in the binary response model in the rest of the paper. All assumptions and results are to be understood as conditional on $\pmb{w}_{it}$. Secondly, in the rest of the paper, except when needed, we will drop the individual subscript, $i$.

While we refer (\ref{eq:2}) as reduced form equation, it is possible that the triangular system in (\ref{eq:1}) and (\ref{eq:2}) is in fact fully simultaneous \citep[see][for examples]{blundell:2004}.  However, even if a simultaneous system is not triangular, the triangular representation, such as the above, can be easily derived if the simultaneous equations involving $ y^{*}_{t} $ and $ \pmb{x}_{t} $ are linear and the errors are additively separable. Also, the triangular model can be generalized to allow for random coefficients instead of fixed coefficients. For the sake of exposition, we limit the analysis to fixed coefficients with random effects; a straightforward extension of the method to allow for random coefficients is discussed in Appendix C.

We first define some notations. Let $X\equiv(\pmb{x}_{1},\ldots, \pmb{x}_{T})$, a $(d_{x}\times T)$ matrix, denote the history of the endogenous variables, $ \pmb{x} $; let  $Z\equiv (\pmb{z}_{1},\ldots, \pmb{z}_{T})$ of dimension $(d_{z}\times T)$,  denote the history of the exogenous variables, $ \pmb{z} $; similarly,  $\pmb{\zeta}\equiv(\zeta_{1},\ldots, \zeta_{T})^{\prime}$ is a vector containing the realizations of idiosyncratic shocks in the structural  equation, and $\pmb{\epsilon}\equiv(\pmb{\epsilon}_{1},\ldots, \pmb{\epsilon}_{T})$ is a $(d_{x}\times T)$ matrix containing the realizations of idiosyncratic shocks in the reduced form equations. 

The first assumptions toward identifying the structural measures of interest such as ASF and APE are:
\begin{ass}\label{as:1}  $ \pmb{\zeta}, \pmb{\epsilon} \indep Z, \theta, \pmb{\alpha} $.
\end{ass}

\begin{ass}\label{as:2}
	
	(a) $ \theta, \pmb{\zeta}|X, Z, \pmb{\alpha} \sim \theta, \pmb{\zeta}|\pmb{\epsilon}, Z, \pmb{\alpha} \sim \theta, \pmb{\zeta}| \pmb{\epsilon}, \pmb{\alpha}$  where   $ \pmb{\epsilon}= X-\E (X|Z, \pmb{\alpha}) $,

	(b) $\theta, \zeta_{t} \indep \pmb{\epsilon}_{-t}|\pmb{\alpha}, \pmb{\epsilon}_{t}$. 
\end{ass}

\begin{ass}\label{as:3}
	\begin{align}\nonumber
	\pmb{\alpha}|Z \sim
	\N\begin{bmatrix} \E(\pmb{\alpha}|Z), \Lambda_{\alpha\alpha}\end{bmatrix} \text{  and  } \pmb{\epsilon}_{t} \sim
	\N\begin{bmatrix} 0,  \Sigma_{\epsilon\epsilon} \end{bmatrix}, 
	\end{align}
	where $\E(\pmb{\alpha}|Z)=\bar{\pi}\bar{\pmb{z}} $ could be either \citeauthor{chamberlain:1984}'s or \citeauthor{mundlak:1978}'s specification for correlated random effects.
\end{ass}

Assumptions AS \ref{as:1} and AS \ref{as:2}, which serve to account for unobserved confounders such as the unobserved heterogeneities that are fixed at least in short panels and to eliminate the confounding influences of observed and unobserved confounders, are weaker than the identifying assumptions for the traditional control function method such as in \citetalias{blundell:2004} and \citet{rothe:2009}. In the traditional control function method, (a) $ Z $ is assumed independent of all heterogeneity terms, $ \theta, \zeta_{t}, \pmb{\alpha}, \pmb{\epsilon}$, and (b) it is assumed that $\theta+\zeta_{t} \indep X | \pmb{\alpha}+ \pmb{\epsilon}_{t} =\pmb{\upsilon}_{t} $; such an assumption also implies that heterogeneity in each of the $ d_{x} $ reduced form equations is scalar, whereas one would like to allow for additional heterogeneities such as individual effects and/or random coefficients. We allow $ Z $ to be correlated with $ \theta $ and $ \pmb{\alpha} $, and in part (a) of AS \ref{as:2}, assume that conditional on the history of reduced form error terms, $\pmb{\epsilon}$ and $\pmb{\alpha}$, $Z$, and thereby $ X $, is independent of the structural error terms $\theta$ and $\zeta_{t}$. The assumption in AS \ref{as:2} (b), where only contemporaneous errors are correlated, has been made to ease exposition, and can be dropped.

To see why Assumption AS \ref{as:2} (a) is justified or realistic in empirical settings, consider the empirical example in \citetalias{blundell:2004}, in which they estimate the causal effect of ``other" household income on work participation decision by men without college education. The triangular model in \citetalias{blundell:2004} augmented with individual effects is given as:   
\begin{align} \label{eq:bp1}
y_{t}=1\{y^{*}_{t} = x_{t}\varphi + z_{1t}\varphi_{z} + \theta + \zeta_{t}>0\}\\ \label{eq:bp2}
x_{t} =  z_{1t}\pi_{1} + z_{21t}\pi_{21} +  z_{22t}\pi_{22} + \alpha + \epsilon_{t},
\end{align}
where $ y^{*}_{t} $ is the number of hours worked in a week by the man in the house; $ x_{t} $, which is weekly ``other" household income and which includes the income of the spouse, is endogenous; $z_{1t}$ is a set of strictly exogenous variables that includes various observable social demographic variables; $ z_{21t} $ is a set of strictly exogenous variables that includes household characteristics, for example, the education level of the spouse; and the instrument, $ z_{22t} $, is the weekly welfare benefit entitlement variable, which is excluded from the structural equation (\ref{eq:bp1}). This entitlement variable measures the transfer income the family would receive if neither spouse was working.  

The above triangular representation can be obtained by augmenting with individual effects the simultaneous equation model considered in \citet{blundell:1994}: 
\begin{align} \label{eq:bp3}
&y^{*}_{t} = x_{t}\varphi + z_{1t}\varphi_{z} + \theta + \zeta_{t}\\ \label{eq:bp4}
&x_{t} = y^{*}_{t}\beta_{y} + z_{21t}\beta_{z1} +  z_{22t}\beta_{z2} + \gamma + \xi_{t}. 
\end{align}
Since the structural equation, (\ref{eq:bp3}) is derived using a wage equation (see also \citetalias{blundell:2004}), the individual effect, $ \theta = f(\mu, \omega) $, could be a composite of unobserved ``taste" for work, $ \mu $,  and unobserved ability/productivity, $ \omega $; whereas $ \gamma $ in  equation (\ref{eq:bp4}) could represent household's or spouse's unobserved productivity \citep[see][]{blundell:2007}. Substituting $ x_{t}\varphi + z_{1t}\varphi_{z} + \theta + \zeta_{t} $ for $ y^{*}_{t} $ in equation (\ref{eq:bp4}), we get the reduced form in equation (\ref{eq:bp2}), where\begin{align}\nonumber
\pi_{1} = \frac{\varphi_{z}\beta_{y}}{1-\varphi\beta_{y}},  \pi_{21} = \frac{\beta_{z1}}{1-\varphi\beta_{y}},  \pi_{22} = \frac{\beta_{z2}}{1-\varphi\beta_{y}}  ,  \alpha = \frac{\theta\beta_{y} + \gamma }{1-\varphi\beta_{y}}  , \text{ and }  \epsilon_{t} = \frac{\zeta_{t}\beta_{y} + \xi_{t} }{1-\varphi\beta_{y}}.   	
\end{align}

Let $\pmb{z}_{t}\equiv \{z_{1t}, z_{21t},  z_{22t}\} $. First, given what the unobserved heterogeneities, $ \theta $, $ \gamma $ and $\alpha = g(\theta, \gamma) $, are, it is quite likely that $ \pmb{z}_{t} $, which includes the education level of the couple and the welfare benefits they receive, is correlated with them. Second, if, as in \citetalias{hoderlein:2012}, $ \zeta_{t} $ in equation (\ref{eq:bp3}) represents new private information revealed to the household, which affects both $y^{*}_{t}$ and $x_{t}$, then (a) $\xi_{t} = f(\zeta_{t})$ in equation (\ref{eq:bp4}) and (b) even after conditioning on individual effects,  $x_{t}$ and $\zeta_{t}$ would be dependent.     

For the example above and in general, given $ Z $, the only source of dependence between $ X $ and $(\theta, \zeta_{t} )$  is through the relationship between $(\pmb{\alpha},\pmb{\epsilon}_{1}, \ldots, \pmb{\epsilon}_{T})$ and $(\theta, \zeta_{t} )$. Therefore, given $ Z $, conditioning on $(\pmb{\alpha},\pmb{\epsilon}_{1}, \ldots, \pmb{\epsilon}_{T})$ eliminates this source of dependency. Since $\pmb{\epsilon}_{t}$ and $\zeta_{t}$ are by Assumption AS \ref{as:1} independent of $ Z $, it can be shown that Assumption AS \ref{as:2} (a) boils down to $ \theta \indep Z| \pmb{\alpha}$. In other words, what we are assuming is that no information about $ Z $ is contained
in $ \theta $ over and above that contained in $ \pmb{\alpha} $. This assumption, in effect, is similar to the assumption of strict exogeneity in panel data models: once the endogeneity of $ \pmb{x}_{t} $ has been addressed by conditioning on $(\pmb{\alpha},\pmb{\epsilon}_{1}, \ldots, \pmb{\epsilon}_{T})$, given time-invariant heterogeneity, $ \pmb{\alpha} $, $ Z $ has no extraneous influence on $y^{*}_{t}$.

Now, as the unobserved conditioning variables, $\pmb{\alpha}$ and $\pmb{\epsilon}_{t}$, cannot be identified separately, for identifying structural measures our method then requires that we be able to recover the conditional distribution of $\pmb{\alpha}$ given $X$ and $Z$ so that the control functions, which are based on $\E(\pmb{\alpha}|X, Z)$, can be estimated. However, we do not know of any semi or nonparametric estimator, and it is outside the scope of this paper to develop one, where the distribution or the expectation of individual effects or random coefficients conditional on $X$ and $Z$ are estimated for a system of regressions. Although, with parametric specification of the error components as in Assumption AS \ref{as:3}, this conditional distribution is obtained readily.

\citet{biorn:2004} proposed a step-wise maximum likelihood method for estimating the systems of regression equations, where the distributions of error components are specified as normal. Given Assumption AS \ref{as:3}, where the conditional distribution of $ \pmb{\alpha} $ given $Z$ and the marginal distribution of $\pmb{\epsilon}_{t}$ are both normal, the tail, $\pmb{a}=\pmb{\alpha} -\E(\pmb{\alpha}|Z)= \pmb{\alpha}-\bar{\pi}\bar{\pmb{z}}$, is distributed normally with conditional mean zero and variance $\Lambda_{\alpha\alpha}$. We can therefore write the reduced form in (\ref{eq:2}) as 
\begin{align}\label{eq:3}
&\pmb{x}_{t} = \pi\pmb{z}_{t} +\bar{\pi}\bar{\pmb{z}}+
\pmb{a} + \pmb{\epsilon}_{t},
\end{align}
which can be estimated using the method in \citet{biorn:2004}

Some recent papers that have employed \citeauthor{mundlak:1978}'s specification for correlated random effects are listed in \citetalias{semykina:2018}. Though we have assumed the error term to be normally distributed, as we discuss here, violation of assumed normality of reduced form errors is unlikely to have a bearing on the estimates. First, the coefficients in \citet{biorn:2004} are estimated by the method of generalized least squares (GLS), which does not require normality of the errors, $ \pmb{\alpha} $ and $ \pmb{\epsilon}_{t} $. Second, \citeauthor{biorn:2004} shows that the ML estimates of the covariance matrices, $ \Sigma_{\epsilon\epsilon}$ and $\Lambda_{\alpha\alpha} $, for a moderately large $ N $ are approximately same as those that are obtained when the distributions of $ \pmb{\alpha} $ and $ \pmb{\epsilon}_{t} $ are unknown. Third, estimating the reduced form equation augmented with $ \bar{\pmb{z}}=T^{-1}\sum_{t=1}^{T}\pmb{z}_{t} $ -- to account for the correlation between $ \pmb{\alpha} $ and  $ Z $ --  by GLS yields fixed-effects (FE) estimates of $ \pi $ for time varying $ \pmb{z}_{t} $  \citep[see][]{wooldridge:2019}.

%While it is desirable that the distribution of the errors be flexible with mild distributional assumptions	
For scalar $ x $, \citet{baltagi:2010} allow for heteroscedastic $a$ and serial correlation among $\epsilon_{t}$.  Thus, when $d_{x}=1$, the assumptions that $a$ is completely independent of $Z$ and that $\epsilon_{t} $'s are i.i.d. can be weakened to allow for non-spherical error components. However, since we want to account for the endogeneity of multiple endogenous regressors, we will stick to Assumption AS \ref{as:3}, and estimate the first-stage parameters, $\Theta_{1}\equiv \{\pi, \bar{\pi}, \Sigma_{\epsilon\epsilon},\Lambda_{\alpha\alpha}\}$, of the reduced form equations (\ref{eq:3}) using \citeauthor{biorn:2004}'s step-wise likelihood method, which is briefly described in Appendix E.

%.

\subsection{Identification of Structural Coefficients}

Now, by Assumptions AS \ref{as:1} and AS \ref{as:2} the dependence of $(\theta,\pmb{\zeta})$ on $X$, $Z$, and $ \pmb{\alpha} $ is characterized by $ \pmb{\alpha} $ and $ \pmb{\epsilon}_{t} $.  If $ \pmb{\alpha} $ and $ \pmb{\epsilon}_{t} $ could be identified, we could augment the structural equation with $ \pmb{\alpha} $ and $ \pmb{\epsilon}_{t} $ and estimate the coefficients, $ \pmb{\varphi} $. Since $ \pmb{\alpha} $ and $ \pmb{\epsilon}_{t} $ are not identified separately, the traditional control function approach assumes that the composite error, $\pmb{\upsilon}_{t} = \pmb{\alpha} + \pmb{\epsilon}_{t} $, which are estimated as the residuals of the reduced form equations, is independent of $ Z $ and that conditional on $\pmb{\upsilon}_{t}$,  $X$ is independent of $ \theta+ \zeta_{t} $.  Such an assumption, as we discuss in detail, could quite likely be violated.

In Theorem \ref{lm:2} we show that by estimating the modified structural equation (\ref{eq:7}), which is augmented with the additional control variables, $\hat{\pmb{\alpha}}(X, Z) \equiv \E(\pmb{\alpha}|X, Z)$ and $\hat{\pmb{\epsilon}}_{t} (X, Z) \equiv \E(\pmb{\epsilon}_{t}|X, Z)$, the structural coefficients $ \pmb{\varphi} $ can be estimated consistently. The modified structural equation (\ref{eq:7}) is derived based on Lemma \ref{lm:0}, and the control variables are identified in Lemma \ref{lm:1}.   

In Lemma \ref{lm:0}, we show that:
\begin{lem}\label{lm:0}
	If (i) Assumptions AS \ref{as:1} and AS \ref{as:2} hold and (ii) $ \E(\theta|\pmb{\alpha}) $ and $ \E(\zeta_{t}| \pmb{\epsilon}_{t}) $ are linear in $ \pmb{\alpha} $ and $  \pmb{\epsilon}_{t} $ respectively so that $ \E(\theta|\pmb{\alpha}) + \E(\zeta_{t}| \pmb{\epsilon}_{t})
	= \pmb{\varphi}_{\alpha}\pmb{\alpha} + \pmb{\varphi}_{\epsilon}\pmb{\epsilon}_{t} $, then  $ \E(\theta + \zeta_{t}|X, Z) $ depends on $ (X, Z)$ only through  $\hat{\pmb{\alpha}}(X, Z) $ and $\hat{\pmb{\epsilon}}_{t} (X, Z) $. 
\end{lem}
\begin{profl}
	Now,
	\begin{align}\nonumber
	\E(\theta + \zeta_{t}|X, Z)&=\E(\E(\theta + \zeta_{t}|X, Z, \pmb{\alpha})|X, Z) =\E(\E(\theta +\zeta_{t}|\pmb{\alpha}, \pmb{\epsilon}_{t}) |X, Z)\\\label{eq:19}&=\E(\E(\theta|\pmb{\alpha}) + \E( \zeta_{t}| \pmb{\epsilon}_{t}) |X, Z)
	= \pmb{\varphi}_{\alpha}\E(\pmb{\alpha}|X, Z) + \pmb{\varphi}_{\epsilon}\E(\pmb{\epsilon}_{t}|X, Z),
	\end{align}
	where the first equality is due to the law of iterated expectations, the second is due to Assumptions AS \ref{as:2}, the third due to AS \ref{as:1}, and the fourth due assumption (ii) in the Lemma.     
	
\end{profl}

In Lemma \ref{lm:1} we show that:\footnote{In Lemma \ref{alm:1}, part (b), in the appendix we derive the conditional distribution of $\pmb{\alpha}$ given $X$ and $Z$ for the estimator in \citet{baltagi:2010}, where $\alpha$ and $\epsilon_{t} $ are both scalar, $\alpha$ is heteroscedastic, and the distribution of $\epsilon_{t} $ is non-spherical.} 
\begin{lem}\label{lm:1}
	Let $\pmb{x}_{t}= \pi\pmb{z}_{t} +\bar{\pi}\bar{\pmb{z}} +\pmb{a} + \pmb{\epsilon}_{t}$, $ t \in \{1,\ldots, T\} $, where $ \bar{\pmb{z}} = \frac{1}{T}\sum_{t=1}^{T}\pmb{z}_{t} $, and let AS \ref{as:3} hold, then $ \pmb{\alpha}= \bar{\pi}\bar{\pmb{z}} +\pmb{a} $, given $ X $ and $ Z $, is distributed with conditional mean  
	\begin{align}\nonumber
	\E( \pmb{\alpha}|X, Z) \equiv \hat{\pmb{\alpha}}(X, Z, \Theta_{1} )=\bar{\pi}\bar{\pmb{z}}+ \E( \pmb{a}|X, Z) = \bar{\pi}\bar{\pmb{z}}+ \Omega\Sigma_{\epsilon\epsilon}^{-1}\sum_{t=1}^{T}(\pmb{x}_{t}-\pi\pmb{z}_{t}  -\bar{\pi}\bar{\pmb{z}}),
	\end{align}
	where $\Omega =[T\Sigma_{\epsilon\epsilon}^{-1}+\Lambda_{\alpha\alpha}^{-1}]^{-1}$ is the conditional variance of $ \pmb{\alpha} $ given $ X $ and $ Z $;  $ \Lambda_{\alpha\alpha}$ and $\Sigma_{\epsilon\epsilon}$ being the covariance matrices of $ \pmb{a} $ and $ \pmb{\epsilon}_{t} $ respectively. 
	
\end{lem}
\begin{profl}
	
	\emph{Given in Appendix A. }
	
\end{profl}
Conditional mean of $\pmb{\epsilon}_{t} $ given $ X $ and $ Z $ is then given by
\begin{align}\nonumber
\hat{\pmb{\epsilon}}_{t}(X, Z, \Theta_{1}) = \pmb{x}_{t} - \pi\pmb{z}_{t} - \E(\pmb{\alpha}|X, Z) = \pmb{x}_{t}-\pi\pmb{z}_{t} - \hat{\pmb{\alpha}}(X, Z, \Theta_{1}) =\pmb{\upsilon}_{t} - \hat{\pmb{\alpha}}(X, Z, \Theta_{1} ). 
\end{align}

\begin{remark}\label{rm:0}
	The expected posterior estimates, $\hat{\pmb{\alpha}}$,\footnote{For notational convenience, we use  $  \hat{\pmb{\alpha}}(X, Z, \Theta_{1}) $, $\hat{\pmb{\alpha}}(X, Z)$ and $ \hat{\pmb{\alpha}} $ interchangeably; the same for $\hat{\pmb{\epsilon}}_{t}(X, Z, \Theta_{1}) $, $\hat{\pmb{\epsilon}}_{t}(X, Z)$ and $\hat{\pmb{\epsilon}}_{t}$.} of $\pmb{\alpha}$ in Lemma \ref{lm:1}, however, is the empirical Bayes or the \citetalias{james:1961}'s shrinkage estimator of $\pmb{\alpha}$ \citep[see][]{efron:2010}. The empirical Bayes estimation has gained certain popularity in economics. In education economics, it is employed as a procedure to calculate teacher value added and often as a way to make imprecise estimates more reliable \citep[see][and the references therein]{guarino:2015}. 
	
	Now, we can write  $\hat{\pmb{\alpha}}$ in Lemma \ref{lm:1} as
	\begin{align}\nonumber
	\hat{\pmb{\alpha}}& =\bar{\pi}\bar{\pmb{z}}+ \E( \pmb{a}|X, Z) =\bar{\pi}\bar{\pmb{z}}+ \underbrace{[\Sigma_{\epsilon\epsilon}^{-1}+\frac{1}{T}\Lambda_{\alpha\alpha}^{-1}]^{-1}\Sigma_{\epsilon\epsilon}^{-1}}_{\text{shrinkage factor}}\frac{1}{T}\sum_{t=1}^{T}(\pmb{x}_{t}-\pi\pmb{z}_{t}  -\bar{\pi}\bar{\pmb{z}}).\end{align}
	With the reduced form equation being specified as in equation (\ref{eq:3}) and given Assumption AS \ref{as:3}, it can be verified that, for a given $\Theta_{1}$, the MLE of $ \pmb{a} $ is $ \frac{1}{T}\sum_{t=1}^{T}(\pmb{x}_{t}-\pi\pmb{z}_{t}  -\bar{\pi}\bar{\pmb{z}}) $.  Since for small $ T $ the MLE of $ \pmb{a} $ is less reliable, the \text{shrinkage factor} of the empirical Bayes estimator shrinks the MLE of $\pmb{a}$ towards its mean, 0; thus shrinking the ML estimate of $\pmb{\alpha}$ towards its prior mean, $ \bar{\pi}\bar{\pmb{z}} $.  Given consistent estimates of reduced-form parameters, $\Theta_{1}$, the empirical Bayes estimate, $\hat{\pmb{\alpha}}$, of $\pmb{\alpha}$ is the minimum mean squared error predictor of $\pmb{\alpha} $ under normality, and therefore a justified estimator of $\pmb{\alpha}$.\footnote{When there is a single endogenous regressor, so that the reduced form has a single equation, then one can employ the estimation method in \citet{gu:2017}, who, for longitudinal data, have developed a non-parametric estimation method to estimate the empirical Bayes estimates of the individual effects, $\alpha $, and the distribution of $\alpha $. Since the posterior mean $\hat{\alpha} $ is estimated non-parametrically, the large sample properties of the structural coefficients will have to be worked out anew.}
	
	For large $T$, since the estimates of $\pi $ are the FE estimates of $\pi$, it can be shown that $ \hat{\pmb{\alpha}}$ consistently estimates the fixed effects, $ \pmb{\alpha} $. With the FE estimates of $ \pmb{\alpha} $ given by $\hat{\pmb{\alpha}}_{FE} = \frac{1}{T}\sum_{t=1}^{T}(\pmb{x}_{t}-\pi\pmb{z}_{t})$, we can write $ \hat{\pmb{\alpha}} $ as \begin{align}\nonumber
	\hat{\pmb{\alpha}}  = \bar{\pi}\bar{\pmb{z}}+ [\Sigma_{\epsilon\epsilon}^{-1}+\frac{1}{T}\Lambda_{\alpha\alpha}^{-1}]^{-1}\Sigma_{\epsilon\epsilon}^{-1}(\hat{\pmb{\alpha}}_{FE} - \bar{\pi}\bar{\pmb{z}}). 
	\end{align} Assuming $ N $ is large to have consistently estimated the reduced form parameters, since $  \hat{\pmb{\alpha}}_{FE} $ converges in probability to $ \pmb{\alpha}$ and $[\Sigma_{\epsilon\epsilon}^{-1}+\frac{1}{T}\Lambda_{\alpha\alpha}^{-1}]^{-1}\Sigma_{\epsilon\epsilon}^{-1} $ to an identity matrix  as $ T\rightarrow \infty $, by continuous mapping theorem it can be shown that $ \hat{\pmb{\alpha}} \xrightarrow{p}  \pmb{\alpha} $, and consequently $\hat{\pmb{\epsilon}}_{t} \xrightarrow{p} \pmb{\epsilon}_{t}$.
	% Therefore for large $N$ and $T $, $(\hat{\pmb{\alpha}},\hat{\pmb{\epsilon}}_{t})\approx (\pmb{\alpha},\pmb{\epsilon}_{t})$.
	
\end{remark}

Given Lemma \ref{lm:0}, we have $ \E(\textit{y}^{*}_{t}| X, Z) = \pmb{x}^{\prime}_{t}\pmb{\varphi}+ \E(\theta+\zeta_{t}| X,Z) = \pmb{x}^{\prime}_{t}\pmb{\varphi} + \pmb{\varphi}_{\alpha}\hat{\pmb{\alpha}}+\pmb{\varphi}_{\epsilon}\hat{\pmb{\epsilon}}_{t}  $. 
We can then write equation (\ref{eq:1}) written as
\begin{align}\label{eq:7}
&y_{t} = 1\{\mathbb{X}^{\prime}_{t}\Theta_{2}+ \eta_{t}>0\}, 
\end{align}
where $ \Theta_{2} \equiv (\pmb{\varphi}^{\prime}, \pmb{\varphi}_{\alpha}^{\prime}, \pmb{\varphi}_{\epsilon}^{\prime})^{\prime} $, $\mathbb{X} \equiv  (\pmb{x}^{\prime}_{t}, \hat{\pmb{\alpha}}^{\prime}, \hat{\pmb{\epsilon}}^{\prime}_{t})^{\prime}$, and  $\eta_{t} = \theta+\zeta_{t}-\E(\theta+\zeta_{t}|X, Z)$. The two vectors, $ \pmb{\varphi}_{\alpha} $ and $ \pmb{\varphi}_{\epsilon} $, when estimated give us a test of exogeneity of $\pmb{x}_{t}$. Though $ \eta_{t} $ by construction is mean independent of $\mathbb{X}_{t}$,\footnote{Now, $ \E(\eta_{t} |\mathbb{X}_{t}) = \E( \E(\eta_{t} |X, Z, \mathbb{X}_{t})| \mathbb{X}_{t} ) = \E( \E(\eta_{t} |X, Z)| \mathbb{X}_{t} ) = \E( 0| \mathbb{X}_{t} )=0$. The first equality is due to law of iterated expectation, the second is because $ \mathbb{X}_{t} $ is a function of $ X,Z $, and the third follows from the definition of $\eta_{t}$.} for estimation of binary response model in (\ref{eq:7}), the strong condition of complete independence is required \citep{manski:1988}, or   
\begin{ass}\label{as:4}
	$ \eta_{t} \indep \mathbb{X}_{t} $.
\end{ass}

Since $ \hat{\pmb{\epsilon}}_{t} $ and $ \hat{\pmb{\alpha}}$, both, are of dimension $ d_{x} $, the dimension of $ \mathbb{X}_{t} $ is $ 3d_{x} $.\footnote{Note that we have suppressed $ \pmb{w}_{t} $ in $ \mathbb{X}_{t} $, where $ \pmb{w}_{t} $ is of dimension $d_{w}$. So, in fact, the dimension of  $ \mathbb{X}_{t} $ is $ 3d_{x} +d_{w}$. Suppressing $ \pmb{w}_{t} $ in $ \mathbb{X}_{t} $, however, results in no loss of generality.} The identification conditions for $ \Theta_{2} $ in (\ref{eq:7}) to be identified when $ \eta_{t} $ is assumed to follow a known distribution are: (a)  $\eta_{t} $ be independent of $ \mathbb{X}_{t} $ and (b) $ \rank(\E(\mathbb{X}_{t}\mathbb{X}_{t}^{\prime})) = 3d_{x} $. In Theorem \ref{lm:2} we show that condition (b) is satisfied.
\begin{theo}\label{lm:2}
	If  (i) $ \rank(\E(\pmb{x}_{t}\pmb{x}_{t}^{\prime}))=d_{x}$; (ii) $ \rank(\Pi)=d_{x} $, where $ \Pi = \begin{pmatrix}
	\pi & \bar{\pi}
	\end{pmatrix} $; (iii) $ \rank(\E((\pmb{z}^{\prime}_{t},\bar{\pmb{z}}^{\prime} )^{\prime}(\pmb{z}^{\prime}_{t},\bar{\pmb{z}}^{\prime} )))=k$  where $ k = \dim((\pmb{z}^{\prime}_{t},\bar{\pmb{z}}^{\prime} )^{\prime}) $; and (iv) if AS \ref{as:3} holds so that the covariance matrices of $\pmb{\epsilon}_{t}$ and $\pmb{\alpha}$ are of full rank,  then $\rank(\E(\mathbb{X}_{t}\mathbb{X}_{t}^{\prime})) = 3d_{x} $. 
\end{theo}
\begin{proft}
	\emph{Given in Appendix A.}
\end{proft}
Condition (ii) is the rank condition in \citetalias{imbens:2009}, and it underscores the necessity of exclusion restriction for identification. Conditions (i) to (iii) in Theorem \ref{lm:2} are standard conditions for identification of $ \pmb{\varphi} $ in the traditional control function methods, where the control function is the composite error, $\pmb{\upsilon}_{t} = \pmb{\alpha} + \pmb{\epsilon}_{t}= \pmb{x}_{t} - \pi\pmb{z}_{t}$.  Our conditioning variables, however, are $\hat{\pmb{\epsilon}}_{t}$ and $ \hat{\pmb{\alpha}}$, which are also functions of $ \Lambda_{\alpha\alpha} $ and $ \Sigma_{\epsilon\epsilon}$. Positive definiteness of $ \Lambda_{\alpha\alpha} $ and $ \Sigma_{\epsilon\epsilon} $ in condition (iv) helps establish the statement of the Theorem to be true. 

Appendix B discusses how one can use the method of generalized estimating equation (GEE), which can account for heteroscedasticity and serial dependence in the response outcome, to estimate $ \Theta_{2} $. Following Theorem \ref{lm:2}, since the components of $\pmb{x}_{t}$ are continuous, with scale and location normalization, $ \Theta_{2} $ can be estimated by semiparametric methods without specifying the distribution of $ \eta_{t} $ \citep[see][for a review of identification results for semiparametric binary choice models]{horowitz:2009}. 
%We would like to note that the identification results derived here are easily extended to identify measures of interest for models such as sample selection and tobit models.

Now, we have demonstrated that, with $\hat{\pmb{\alpha}}$ and $\hat{\pmb{\epsilon}}_{t}$ identified in Lemma \ref{lm:1},  AS \ref{as:1}, AS \ref{as:2}, condition (ii) of Lemma \ref{lm:0}, and AS \ref{as:4} can help us identify $\pmb{\varphi} $. Given $\hat{\pmb{\alpha}}$ and $\hat{\pmb{\epsilon}}_{t}$, the same, however, can be achieved through the following assumption: 	
\begin{asscf}\label{ascf:1}	
	(a)	$ \pmb{\zeta},\theta | X, Z, \hat{\pmb{\alpha}} \sim \pmb{\zeta},\theta | V, Z , \hat{\pmb{\alpha}}\sim \pmb{\zeta},\theta | V, \hat{\pmb{\alpha}} $,
	where $ V \equiv (\pmb{\upsilon}_{1}, \ldots, \pmb{\upsilon}_{T}) =  X-\pi Z $ and  $ \hat{\pmb{\alpha}} = \E(\pmb{\alpha}| X,Z)  $. 
	
	(b) $\zeta_{t},\theta \indep V_{-t} | \pmb{\upsilon}_{t}, \hat{\pmb{\alpha}} $. 
\end{asscf} 	
In part (a), the assumption is that the dependence of $(\theta,\pmb{\zeta})$ on $X$ and $Z$ is completely characterized by $V$ and $\hat{\pmb{\alpha}}$. Since $\hat{\pmb{\epsilon}}_{t}= \pmb{\upsilon}_{t} - \hat{\pmb{\alpha}} $, there is one-to-one mapping between $(\hat{\pmb{\epsilon}}_{t}, \hat{\pmb{\alpha}}) $ and $(\pmb{\upsilon}_{t} , \hat{\pmb{\alpha}}) $, and therefore the conditioning $ \sigma $-algebra, $\sigma(\hat{\pmb{\epsilon}}_{t}, \hat{\pmb{\alpha}}) $, is same as the $ \sigma $-algebra, $\sigma(\pmb{\upsilon}_{t}, \hat{\pmb{\alpha}}) $. By parts (a) and (b) therefore $ X $ is independent of $ (\zeta_{t},\theta) $ given $ (\hat{\pmb{\epsilon}}_{t}, \hat{\pmb{\alpha}}) $.\footnote{Part (b), however, as in AS \ref{as:2} (b), is to facilitate comparison with the traditional control function method and can be dropped.} In other words, in Assumption ACF \ref{ascf:1} we are proposing $\hat{\pmb{\epsilon}}_{t} $ and  $\hat{\pmb{\alpha}} $ as control functions for panel data. If we further assume that $\E(\zeta_{t}+\theta| \hat{\pmb{\epsilon}}_{t}, \hat{\pmb{\alpha}})$ is linear in $ \hat{\pmb{\epsilon}}_{t}$ and $ \hat{\pmb{\alpha}} $, and let $\eta_{t} = \zeta_{t}+\theta - \E(\zeta_{t}+\theta| \hat{\pmb{\epsilon}}_{t}, \hat{\pmb{\alpha}})$, we obtain equation (\ref{eq:7}). Given AS \ref{as:4}, as shown in Theorem \ref{lm:2}, we can then estimate the structural coefficients, $ \pmb{\varphi} $.\footnote{It would be worthwhile to extend the methodology to estimate the parameters of the dynamic binary choice model, $ y_{t}=1\{y^{*}_{t} = \gamma y_{t-1} + \pmb{x}_{t}^{\prime}\pmb{\varphi} + \theta + \zeta_{t}>0\} $. Any control function (CF) method for the dynamic model, in addition to endogeneity of $ \pmb{x}_{t} $, would have to account for the fact that in each period $ t $, the history, $ Y^{t-1} \equiv \{ y_{1}, \ldots, y_{t-1} \} $, up to $ t $ of the predetermined variable, $ y_{t} $, is correlated with the unobserved heterogeneity, $ \theta $ \citep[see][]{arellano:2003}. Recent papers on dynamic binary choice such as by \citet{honore:2000}, \citet{kitazawa:2021} and  \citet{khan:2020}, where the conditioning is on $\theta$,  have found novel ways such that either  $ \theta $ is eliminated or it plays no role in the identification of $ (\gamma , \pmb{\varphi}) $. If the common assumption in these papers that conditional on $ \theta $, the covariates, $ \pmb{x}_{t} $, are exogenous fails, then constructing CFs for the dynamic model could offer a solution to the problem of endogeneity. Besides, there would be a possibility for estimating partial effects (see subsection 2.2) as most dynamic models, because of the presence of the unobserved, $ \theta $, are unable to.}

We now compare the proposed control functions with the traditional control functions, and argue for the appropriateness of the proposed control functions in the context of panel data. Now, we have pointed out that conditioning on the proposed control functions, $\hat{\pmb{\epsilon}}_{t} $ and $ \hat{\pmb{\alpha}} $, is equivalent to conditioning on the traditional control function, $ \pmb{\upsilon}_{t}(\pmb{x}_{t},\pmb{z}_{t}) =\pmb{x}_{t} -\pi\pmb{z}_{t}$, \textit{and additionally} on individual specific information as summarized by $ \hat{\pmb{\alpha}}(X,Z) $. That is, in assuming that $ (\zeta_{t},\theta) \indep X|\hat{\pmb{\epsilon}}_{t}, \hat{\pmb{\alpha}} $, we are saying that no information about $ X $ is contained in $ (\zeta_{t},\theta) $ over and above that contained in $ (\hat{\pmb{\epsilon}}_{t}, \hat{\pmb{\alpha}}) $ or equivalently in $ (\pmb{\upsilon}_{t}, \hat{\pmb{\alpha}}) $. This, as we discuss in Remark \ref{rm:1} and Remark \ref{rm:2}, may not hold true if only $\pmb{\upsilon}_{t}$ is assumed to be the control function.

\begin{remark}\label{rm:1}
	When $ Z \nindep (\theta , \pmb{\alpha}) $, then the requirement of the traditional control function method that $ Z$ be independent of $ (\pmb{\upsilon}_{t},  \zeta_{t}+\theta) $ is violated and $\zeta_{t}+\theta \indep Z | \pmb{\upsilon}_{t} $ does not hold generally. In which case, Assumption ACF \ref{ascf:1} seems plausible as $\zeta_{t}+\theta  $ is mean independent of  $ (X, Z) $ given $(\hat{\pmb{\alpha}},  \hat{\pmb{\epsilon}}_{t})$.  This assumption is related to the dependence assumptions in \citetalias{altonji:2005}, \citet{bester1:2009} (BH) and \citetalias{hoderlein:2012}, where the distribution of unobserved effects depends on the
	observed variables only through certain function of the observed variables. These functions, as \citetalias{bester1:2009} argue, may be viewed as sufficient statistic.  \citetalias{altonji:2005} assume that $ (\zeta_{t},\theta)$ is independent of $ X $ given certain summary statistics such as the mean, $ T^{-1}\sum_{t=1}^{T} \pmb{x}_{t} $, or index functions of summary statistics, while in \citetalias{bester1:2009} these functions of observed variables are assumed to be unrestricted index functions. In our case, the control function, $(\hat{\pmb{\epsilon}}_{t}, \hat{\pmb{\alpha}})$, is motivated by the result that under certain restrictions, the mean of $\theta +\zeta_{t}$ given the histories, $ (X, Z) $, of the endogenous and the exogenous variables, depends on $ (X, Z )$ only through $ \hat{\pmb{\epsilon}}_{t} $ and $\hat{\pmb{\alpha}} $. Moreover, as  $(\hat{\pmb{\alpha}},\hat{\pmb{\epsilon}}_{t})$ consistently estimates  $ (\pmb{\alpha},\pmb{\epsilon}_{t})$ when $ T $ is large (see Remark \ref{rm:0}), it implies that for large $ T $, Assumptions ACF \ref{ascf:1} and AS \ref{as:2} are asymptotically equivalent.  
\end{remark}

\begin{remark}\label{rm:2} 
	When $ Z \indep (\theta, \pmb{\alpha}, \zeta_{t}, \pmb{\epsilon}_{t})$, as in the traditional control function approach, then $ (\zeta_{t},\theta) \indep Z|V $ holds true. Since $ V $ is invertible in $ X $ when $ Z $ is given, we have $ (\zeta_{t},\theta) | X, Z \sim (\zeta_{t},\theta) | V, Z \sim (\zeta_{t},\theta) | V $.  In panel data, therefore, when $ Z $ is independent of the unobserved heterogeneities,  $ V $ should be employed as the control function.
	
	Given that the unobserved heterogeneities, $(\theta, \pmb{\alpha})$, which represent unobserved, time-invariant attributes, such as preferences, technologies,
	or abilities, influence the choice of $ \pmb{x}_{t} $ in each time period, it is not only with $ \pmb{x}_{t} $ that the errors, $\theta+\zeta_{t}$, are correlated, but generally with the entire history, $ X $, of the endogenous variable. Moreover, since $ \pmb{x}_{t}$ is endogenous, due to potential feedback from $y_{t}$ to $\pmb{x}_{s}$ for $ s > t $, it is likely that the optimal choice of $\pmb{x}_{t}$ depends on $\zeta $ from the past, or more generally $\zeta $ from other time periods, which is likely to make $ \zeta_{t} $ and $ \pmb{x} $ from other time periods dependent.\footnote{For example, in the study of child labor in section 4, the endogenous variables, household income, amount of land owned by a household, and index of household ownership of productive farm assets, in each period will also depend on unobserved household characteristics such as parents' abilities, quality of land, or possibly other omitted variables fixed at the household level. Moreover, apart from contemporaneous shocks, $ (\zeta_{t}, \pmb{\epsilon}_{t})$, that affect the current choices of both $ \pmb{x}_{t} $ and $ y_{t} $,  there may be feedback from lagged values of $ \zeta $ or $ y $ to $ \pmb{x}_{t} $. In the study of the child labor, for example, current choices of labor supply can impact the future choices of endogenous variables mentioned above. It is thus possible that $(\theta, \zeta_{t}) $ and $ X $, as in the considered example of child labor, could be dependent.} If only $ \pmb{\upsilon}_{t} $ is employed as the control function, as it has been traditionally, then $ X $ may not be conditionally independent of $ (\zeta_{t},\theta) $.  Therefore, there will exist some partial correlation between  $ y_{t} $ and $ \pmb{x} $ from the other time periods if the dependency between the structural errors, $(\theta, \zeta_{t}) $, and the history,  $ X $, is not accounted for. Employing only $ \pmb{\upsilon}_{t} $ as the control function places a strong restriction on the dependence between $(\theta, \zeta_{t}) $ and $X$, which, as shown in Proposition \ref{prop:1}, is unlikely to hold.\footnote{See also section 3, Figure \ref{fig:fig2}, where through  numerical experiments we demonstrate the claims made in this Remark and Proposition \ref{prop:1}.} 		
\end{remark}

\begin{prop}\label{prop:1}
	Let  $ Z \indep (\theta, \pmb{\alpha}, \zeta_{t}, \pmb{\epsilon}_{t})$.  When $ (\theta, \zeta_{t}) $ and $ (\pmb{\alpha}, \pmb{\epsilon}_{t})$  are correlated, then $ \theta + \zeta_{t}\indep X | V$ whereas $\theta + \zeta_{t}\nindep X | \pmb{\upsilon}_{t}$, where  
	$ \pmb{\upsilon}_{t} =  \pmb{\alpha}+ \pmb{\epsilon}_{t} = \pmb{x}_{t}+ \pi\pmb{z}_{t} $ and $ V \equiv (\pmb{\upsilon}_{1}, \ldots, \pmb{\upsilon}_{T})$. 
\end{prop}

\begin{prof}
	
	\emph{Given in Appendix A. }
	
\end{prof}

Here, we would like to point that $ (\hat{\pmb{\alpha}}, \hat{\pmb{\epsilon}}_{t}) $  can be employed in semiparametric methods in \citetalias{blundell:2004} or in \citet{rothe:2009} to estimate $ \pmb{\varphi}$ and measures like the ASF semiparametrically. \citetalias{blundell:2004} extend the the matching estimator for the single-index models with exogenous variables to allow for control functions for handling endogeneity, whereas  \citeauthor{rothe:2009} develops a semiparametric maximum likelihood (SML) method for binary response model to account for endogeneity using control functions. These semiparametric methods do not require one to specify the conditional distribution of $\theta+\zeta_{t}$ given $\hat{\pmb{\epsilon}}_{t}$ and $\hat{\pmb{\alpha}}$. The methods, however, do require that, $ \pmb{z}=\tilde{\pmb{z}} $, contains an instrument that is continuous. If all instruments are discrete, the ``rank condition" in \citetalias{blundell:2004} and condition (ii) of Theorem 1 in \citeauthor{rothe:2009},\footnote{Since conditioning on $\hat{\pmb{\epsilon}}_{t} $ and $ \hat{\pmb{\alpha}} $  is equivalent to conditioning on $\pmb{\upsilon}_{t} =  \pmb{x}_{t} - \pi\pmb{z}_{t} $ and  $ \hat{\pmb{\alpha}} $, and if identification requires that conditional on the control variables, $\hat{\pmb{\epsilon}}_{t}$ and $ \hat{\pmb{\alpha}} $, -- hence, $\pmb{\upsilon}_{t}$ and $ \hat{\pmb{\alpha}} $ -- the vector $ \pmb{x}_{t} $ contains at least one, $ x^{1}_{t} $, continuously distributed component with non-zero coefficient, then it would be necessary that $ \pmb{z}_{t} $ contains a continuously distributed regressor.} necessary for identification, are violated. We do not pursue semiparametric estimation of binary choice models with the control functions developed in this paper any further. Semiparametric estimation and the large sample properties of the estimates are left for future research.

\citet{wooldridge:2015} in providing an overview of control function (CF) methods writes, ``in evaluating the scope of an estimation method, it is important to understand how it works in familiar settings, including cases when it is not necessarily needed." While noting that the standard CF method for cross sectional data when applied to linear models gives coefficients for the endogenous variables that are equal to the standard two-stage least squares (2SLS), \citeauthor{wooldridge:2015} points out certain advantages of the CF approach, such as providing a robust, regression based Hausman test of exogeneity, compared to the 2SLS approach. Before ending this subsection, we therefore show that using  $ \hat{\pmb{\alpha}} $ and $\hat{\pmb{\epsilon}}_{t} $ as additional covariates in linear panel data models is equivalent to estimating the models by a certain 2SLS method.  

Now, in linear panel data models,
\begin{align}\label{leq:0}
y_{t}  =  \pmb{x}_{t}^{\prime}\pmb{\varphi}  + \theta + \zeta_{t}, 
\end{align}
when instruments are correlated with time invariant heterogeneity, fixed effect two-stage least squares (FE2SLS), which employs time-variant instruments that are deviations from the group-mean, $\ddot{\pmb{z}}_{t} = \pmb{z}_{t} - \bar{\pmb{z}}$, is employed. An alternative approach \citep[see][chapter 11]{wooldridge:2010}, is to write $ \theta $ in equation (\ref{leq:0})
as $ \theta = \E(\theta|Z) + \tau $, specify $ \E(\theta|Z) $ as in \citeauthor{mundlak:1978}, and estimate the model by pooled 2SLS using $ ( \ddot{\pmb{z}}_{t}, \bar{\pmb{z}}) $ as instruments. 

By Assumptions AS \ref{as:1}, AS \ref{as:2}, AS \ref{as:3} and condition (ii) of Lemma \ref{lm:0}, we have $ \E(\theta| Z ) = \E(\E(\theta|\pmb{\alpha})| Z ) = \E( \pmb{\varphi}_{\alpha}\pmb{\alpha} | Z ) = \pmb{\varphi}_{\alpha}\bar{\pi}\bar{\pmb{z}}$. Thus we can write the model in equation (\ref{leq:0}) as
\begin{align}\label{leq:1}
y_{t}  =  \pmb{x}_{t}^{\prime}\pmb{\varphi}  + \pmb{\varphi}_{\alpha}\bar{\pi}\bar{\pmb{z}} +\tau + \zeta_{t},
\end{align} 
where the heterogeneity term, $ \tau + \zeta_{t} $, and $ X $ are dependent even though $ \tau + \zeta_{t} $ is mean independent of $Z$. Thus, we can estimate of $ \pmb{\varphi} $ in (\ref{leq:1}) by using instrument variables. Also, by Assumptions ACF \ref{ascf:1} and, with a slight abuse of notations, letting $\E(\zeta_{t}+\theta| \hat{\pmb{\epsilon}}_{t}, \hat{\pmb{\alpha}}) = \pmb{\varphi}_{\alpha}\hat{\pmb{\alpha}} + \pmb{\varphi}_{\epsilon}\hat{\pmb{\epsilon}}_{t}$, the linear model in equation (\ref{leq:0}) as can be written as 
\begin{align}\label{leq:2}
y_{t} =  \pmb{x}_{t}^{\prime}\pmb{\varphi} + \pmb{\varphi}_{\alpha}\hat{\pmb{\alpha}} + \pmb{\varphi}_{\epsilon}\hat{\pmb{\epsilon}}_{t} + \eta_{t}, 
\end{align}
where $\eta_{t} = \theta+\zeta_{t} - \E(\theta+\zeta_{t}|\hat{\pmb{\epsilon}}_{t}, \hat{\pmb{\alpha}})$. 
\begin{theo}\label{lm:3}
	Let the estimate of $\pmb{\varphi} $ in (\ref{leq:1}) by pooled 2SLS using $ ( \ddot{\pmb{z}}_{t} , \bar{\pmb{z}}) $ as additional instruments be denoted by  $ \hat{\pmb{\varphi}}_{IV} $ and let the estimate of the same obtained from estimating (\ref{leq:2}) by pooling the data be denoted by $ \hat{\pmb{\varphi}}_{CF} $, then $ \hat{\pmb{\varphi}}_{CF} =  \hat{\pmb{\varphi}}_{IV} $.  
\end{theo}
\begin{proft}
	\emph{Given in Appendix A. }
\end{proft}
Estimating (\ref{leq:1}) by pooled 2SLS using $ ( \ddot{\pmb{z}}_{t} , \bar{\pmb{z}}) $ as additional instruments is similar to estimating by the error component two-stage least squares (EC2SLS) method proposed in \cite{baltagi:1981}, with the difference that the time-varying instruments are allowed to be correlated to the individual specific unobserved heterogeneity.

%\footnote{In particular, strictly exogenous time-constant instruments can also be employed to help identify the structural coefficients through their variations between individuals. For }

\subsubsection{Identification of Average Structural Function and Average Partial Effects}

With  $ \hat{\pmb{\alpha}} $ and $\hat{\pmb{\epsilon}}_{t} $ as control functions, we have  
\begin{align}\nonumber
\Pr(y_{t} = 1| X, \hat{\pmb{\alpha}}, \hat{\pmb{\epsilon}}_{t}) 
&=\int 1\{-(\zeta_{t}+\theta) < \pmb{x}^{\prime}_{t}\pmb{\varphi}\} dF(\zeta_{t}+ \theta | X, \hat{\pmb{\alpha}}, \hat{\pmb{\epsilon}}_{t})\\\nonumber
&=\int 1\{-(\zeta_{t}+\theta) < \pmb{x}^{\prime}_{t}\pmb{\varphi}\} dF(\zeta_{t}+ \theta | \hat{\pmb{\alpha}}, \hat{\pmb{\epsilon}}_{t})=F(\pmb{x}^{\prime}_{t}\pmb{\varphi}; \hat{\pmb{\alpha}}, \hat{\pmb{\epsilon}}_{t}),
\end{align}
where $ F(\pmb{x}^{\prime}_{t}\pmb{\varphi}; \hat{\pmb{\alpha}}, \hat{\pmb{\epsilon}}_{t}) $ is the conditional CDF of $\zeta_{t}+\theta $ given $ (\hat{\pmb{\alpha}}, \hat{\pmb{\epsilon}}_{t}) $ evaluated at $ \pmb{x}^{\prime}_{t}\pmb{\varphi} $.

%	If conditional on $ (\hat{\pmb{\alpha}}, \hat{\pmb{\epsilon}}_{t}) $, $\zeta_{t}+\theta $ is assumed to be distributed normally with variance $ \sigma^{2} $ and  conditional mean, the regression function,  $\pmb{\rho}_{\alpha}\hat{\pmb{\alpha}} +\pmb{\rho}_{\epsilon}\hat{\pmb{\epsilon}}_{t}$, then \begin{align}\nonumber F(\pmb{x}^{\prime}_{t}\pmb{\varphi}; \hat{\pmb{\alpha}}, \hat{\pmb{\epsilon}}_{t}) =\Phi\biggr(\frac{\pmb{x}^{\prime}_{t}\pmb{\varphi}+\pmb{\rho}_{\alpha}\hat{\pmb{\alpha}} +\pmb{\rho}_{\epsilon}\hat{\pmb{\epsilon}}_{t}}{\sigma} \biggr), \end{align}where $ \Phi(.) $ is the cumulative distribution function of a standard normal random variable.

Given a particular value $ \bar{\pmb{x}} $ of $ \pmb{x}_{t} $, averaging $ F(\bar{\pmb{x}}^{\prime}\pmb{\varphi}; \hat{\pmb{\alpha}}, \hat{\pmb{\epsilon}}_{t}) $ over $ (\hat{\pmb{\alpha}}, \hat{\pmb{\epsilon}}_{t}) $, we get the ASF: 
\begin{align}\nonumber
G(\bar{\pmb{x}})&=\int F(\bar{\pmb{x}}^{\prime}\pmb{\varphi}; \hat{\pmb{\alpha}}, \hat{\pmb{\epsilon}}_{t}) dF(\hat{\pmb{\alpha}},\hat{\pmb{\epsilon}}_{t}),\\ \nonumber
&=\int \biggr[\int 1\{\bar{\pmb{x}}^{\prime}\pmb{\varphi} + \theta + \zeta_{t} >0\}dF(\theta+ \zeta|\hat{\pmb{\alpha}},\hat{\pmb{\epsilon}}_{t})\biggr]dF(\hat{\pmb{\alpha}},\hat{\pmb{\epsilon}}_{t})\\\label{eq:11}
&=\E_{\theta+ \zeta} (1\{\bar{\pmb{x}}^{\prime}\pmb{\varphi} + \theta + \zeta_{t} >0\}).  
\end{align}
The APE of changing a variable, say $\bar{x}_{k}$, from $\bar{x}_{k}$ to $\bar{x}_{k} + \Delta_{k}$ can be obtained as
\begin{align}
\frac{\Delta G(\bar{\pmb{x}})}{\Delta_{k}}= \frac{G(\bar{\pmb{x}}_{-k},\bar{x}_{k}+\Delta_{k}) - G(\bar{\pmb{x}})}{\Delta_{k}}.&
\end{align}

To point-identify the ASF, $ G(\bar{\pmb{x}}) $,  it is required that $ F(\pmb{x}^{\prime}_{t}\pmb{\varphi}=\bar{\pmb{x}}^{\prime}\pmb{\varphi}; \hat{\pmb{\alpha}}= \bar{\pmb{\alpha}}, \hat{\pmb{\epsilon}}_{t} = \bar{\pmb{\epsilon}} ) $ be evaluated at all values of  $ (\bar{\pmb{\alpha}}, \bar{\pmb{\epsilon}}) $ in the support of the unconditional distribution of $ (\hat{\pmb{\alpha}}, \hat{\pmb{\epsilon}}) $. This requires that the support of the conditional distribution of $ (\hat{\pmb{\alpha}}, \hat{\pmb{\epsilon}}) $ conditional on $ \pmb{x}_{t}=\bar{\pmb{x}} $ be equal to the support of the unconditional distribution. It ensures that for any group of
individuals defined in terms of $ (\hat{\pmb{\alpha}}, \hat{\pmb{\epsilon}}_{t} ) $, at least some experience $ \pmb{x}_{t}=\bar{\pmb{x}} $. This is analogous to the overlap condition in the program
evaluation literature, where treatment is discrete.
% \citep[see notes on covariate adjustment in][]{graham:2015}.  

For many triangular systems that employ the control function, $ \pmb{\upsilon}_{t} $, or $ \{F(\upsilon_{1,t}) \ldots F(\upsilon_{d_{x},t})\}^{\prime}$ -- where $ F(\upsilon_{1,t}) $, the CDF of $ \upsilon_{1,t} $,  is equal to $ F(x_{1,t}|\pmb{z}_{it}) $, the CDF of $x_{1,t} $ given $ \pmb{z}_{t} $ --  the requirement of common support necessitates that along with the rank condition (see theorem \ref{lm:2}) the set of instruments, $ \pmb{z}_{t} $, contains a continuous instrument with large support (this is discussed in \citetalias{imbens:2009}, \citeauthor{Florens:2008}, and in  \citet{blundell:2003}). In lemma \ref{lm:4} we show that when the instruments have small support -- that is, when instruments are binary, discrete, or continuous but without large support -- the support requirement for  $ G(\bar{\pmb{x}}) $ to be point-identified by the ``partial-mean" formulation in equation (\ref{eq:11}) is satisfied if $ \pmb{x} $ have a large support.
\begin{lem}\label{lm:4}
	If the endogenous variables, $ \pmb{x} $, have large a support, then under AS \ref{as:3}, the support of the conditional distribution of $\hat{\pmb{\alpha}}(X, Z, \Theta_{1})$ and $\hat{\pmb{\epsilon}}_{t}(X, Z, \Theta_{1})$, conditional on $\pmb{x}_{t}=\bar{\pmb{x}}$, is same as the support of their marginal distribution.
\end{lem}
\begin{profl}
	\emph{Given in Appendix A.}
\end{profl}

In our approach, the control functions, $\hat{\pmb{\epsilon}}_{t}$ and $\hat{\pmb{\alpha}}$, are smooth, unbounded functions of $\pmb{x}_{t}$'s, $ t\in\{1,\ldots,T\} $. Therefore, when $\pmb{x}$ is continuous with a large support, because the $\pmb{x}_{s}$'s, $s \neq t$, are unrestricted  the ranges of $\hat{\pmb{\alpha}}$ and $\hat{\pmb{\epsilon}}_{t}= \pmb{x}_{t}- \pi\pmb{z}_{t}- \hat{\pmb{\alpha}}$ conditional on $\pmb{x}_{t}=\bar{\pmb{x}} $ do not depend on $ \pmb{x}_{t} $. Since the result does not rely on any kind of restriction on $\pmb{z}_{t}$'s support, our method circumvents the need to have a continuous instrument with large support to point-identify the ASF and/or the APEs when some of the $\pmb{x}$'s have large supports.

Before proceeding further, we note that these results could be useful for computing the quantile structural function (QSF), as in \citetalias{imbens:2009}, for the kind of triangular set-ups considered in \citet{blundell:2003}, where the structural equation is nonseparable in errors, but additively separable in the reduced form. When the nonseparable structural and reduced form equations are strictly increasing in their respective scalar errors, \citet{dhault:2015} and \citet{torgovitsky:2015} show that the structural function, $y_{t}= g(\pmb{x}_{t},\varepsilon_{t}) $, is point-identified with discrete instruments. The key observation in both the papers is that establishing point identification is tantamount to solving a functional fixed point problem. While the papers differ in their approach to establishing sufficient conditions under which the fixed point problem admits a unique solution, they are able to avoid the partial-mean formulation in equation (\ref{eq:11}), or in equation (6) of \citetalias{imbens:2009}'s, to identify the QSF, which in their case is same as the structural function. Since $ g(\pmb{x}_{t},\varepsilon_{t}) $ is required to be strictly monotonic in $ \varepsilon_{t} $, these methods, however, are not suitable for identification in discrete choice models.

When the support condition in lemma \ref{lm:4} is not satisfied, one can establish bounds on the ASF and the APE's. Let $ \mathcal{A} $ be the unconditional support of $ (\hat{\pmb{\alpha}},\hat{\pmb{\epsilon}}_{t}) $ and $ \mathcal{A}(\bar{\pmb{x}}) \equiv \{ \hat{\pmb{\alpha}}, \hat{\pmb{\epsilon}}_{t} : f(\hat{\pmb{\alpha}},\hat{\pmb{\epsilon}}_{t} | \bar{\pmb{x}} )>0\} $ be the support of $ (\hat{\pmb{\alpha}},\hat{\pmb{\epsilon}}_{t}) $ conditional on $\bar{\pmb{x}}$. When the support of $\pmb{x}$ is bounded and the instruments have a small support then $ \mathcal{A} \neq \mathcal{A}(\bar{\pmb{x}}) $. Now, let 	\begin{align}\label{eq:17}
\tilde{G}(\bar{\pmb{x}})=\int_{\mathcal{A}(\bar{\pmb{x}})} F(\bar{\pmb{x}}^{\prime}\pmb{\varphi}; \hat{\pmb{\alpha}}, \hat{\pmb{\epsilon}}_{t}) dF(\hat{\pmb{\alpha}},\hat{\pmb{\epsilon}}_{t}) 
\end{align} 
be the identified object and let $ P(\bar{\pmb{x}})=\int_{\mathcal{A} \cap \mathcal{A}(\bar{\pmb{x}})^{c}} dF(\hat{\pmb{\alpha}},\hat{\pmb{\epsilon}}_{t})$. Since	\begin{align}\nonumber
G(\pmb{x}_{t})&  = \tilde{G}(\bar{\pmb{x}}) + \int_{\mathcal{A} \cap \mathcal{A}(\bar{\pmb{x}})^{c}} F(\bar{\pmb{x}}^{\prime}\pmb{\varphi}; \hat{\pmb{\alpha}}, \hat{\pmb{\epsilon}}_{t}) dF(\hat{\pmb{\alpha}},\hat{\pmb{\epsilon}}_{t}) 
\end{align} 
and since $ 0 \leq F(\bar{\pmb{x}}^{\prime}\pmb{\varphi}; \hat{\pmb{\alpha}}, \hat{\pmb{\epsilon}}_{t}) \leq 1 $, the above equation implies that   $G(\bar{\pmb{x}})\in [ \tilde{G}(\bar{\pmb{x}}), \tilde{G}(\bar{\pmb{x}}) + P(\bar{\pmb{x}}) ]$; that is,  $ G(\bar{\pmb{x}}) $ is set-identified. The bounds on $ G(\bar{\pmb{x}}) $ are sharp since there are no restrictions on $\E(y_{t}|\pmb{x}_{t}=\bar{\pmb{x}}, \hat{\pmb{\alpha}}, \hat{\pmb{\epsilon}}_{t} )=F(\bar{\pmb{x}}^{\prime}\pmb{\varphi}; \hat{\pmb{\alpha}}, \hat{\pmb{\epsilon}}_{t}) $ imposed by the data.

It follows then that the APEs are also set-identified when support requirement in lemma \ref{lm:4} is not met. To derive bounds for the APE of changing $ x_{k} $ from $ \bar{x}_{k} $ to $ \bar{x}_{k}+\Delta_{k} $,  let us first denote $ (\bar{\pmb{x}}^{\prime}_{-k}, \bar{x}_{k}+\Delta_{k})^{\prime} $ by $ \bar{\pmb{x}}_{\Delta k} $. Since the ASF, $ G(\bar{\pmb{x}}_{\Delta k}) $, at $\bar{\pmb{x}}_{\Delta k} $  is partially identified, where $ G(\bar{\pmb{x}}_{\Delta k})\in [ \tilde{G}(\bar{\pmb{x}}_{\Delta k}), \tilde{G}(\bar{\pmb{x}}_{\Delta k}) + P(\bar{\pmb{x}}_{\Delta k}) ]  $, the APE of $ x_{k} $ at $ \bar{\pmb{x}} $, $ \Delta G(\bar{\pmb{x}})/\Delta_{k} $, lies in the interval,
\begin{align}\label{eq:16}
\biggr[ \frac{\tilde{G}(\bar{\pmb{x}}_{\Delta k}) - \tilde{G}(\bar{\pmb{x}}) - P(\bar{\pmb{x}})}{\Delta_{k}}, \frac{\tilde{G}(\bar{\pmb{x}}_{\Delta k}) + P(\bar{\pmb{x}}_{\Delta k}) - \tilde{G}(\bar{\pmb{x}})}{\Delta_{k}}  \biggr], 
\end{align}
where the sharpness of the bounds on APE derives from that of the bounds on ASF.

Once we have the consistent estimates, $\hat{\Theta}_{2}$, of $\Theta_{2}$, to estimate the bounds on APE of a variable, $ x_{k} $,  we first, as in \citetalias{imbens:2009}, estimate the support of the estimates, $(\hat{\hat{\pmb{\alpha}}}_{i},\hat{\hat{\pmb{\epsilon}}}_{it})$,  given $ \bar{\pmb{x}} $ as \begin{align}\nonumber
\hat{\mathcal{A}}(\bar{\pmb{x}}) = \{ \hat{\hat{\pmb{\alpha}}}_{i}, \hat{\hat{\pmb{\epsilon}}}_{it}: \hat{f}(\hat{\hat{\pmb{\alpha}}}_{i},\hat{\hat{\pmb{\epsilon}}}_{it} | \bar{\pmb{x}} )\geq\delta(\bar{p}), (\hat{\hat{\pmb{\alpha}}}_{i}, \hat{\hat{\pmb{\epsilon}}}_{it}) \in \hat{\mathcal{A}} \}, 
\end{align} where  $ \hat{f}(\hat{\hat{\pmb{\alpha}}}_{i},\hat{\hat{\pmb{\epsilon}}}_{it} | \bar{\pmb{x}} ) $, which is the estimate of the conditional density of $(\hat{\pmb{\alpha}}_{i},\hat{\pmb{\epsilon}}_{it})$  given $ \bar{\pmb{x}} $, is obtained by employing the method of estimating the conditional density function in \citet{hall:2004}\footnote{R's `np' package developed by \citet{hayfield:2008}  implements the method. The package's `npcdens' function computes kernel conditional density estimates of $ p $ variables conditional on $ q $ variables. }.  $\hat{\mathcal{A}}$ is an estimator of the support, $\mathcal{A}$,  containing all $(\hat{\hat{\pmb{\alpha}}}_{i},\hat{\hat{\pmb{\epsilon}}}_{it})$. In the above,  $\delta(\bar{p})$ is the trimming parameter and, as discussed in \citet{cadre:2013}, is obtained as a solution to the following equation:
\begin{align}\nonumber
\int_{\{\hat{f}\geq\delta\}} \hat{f}(\hat{\hat{\pmb{\alpha}}}_{i},\hat{\hat{\pmb{\epsilon}}}_{it} | \bar{\pmb{x}} ) =\bar{p},  \text{where $ \bar{p} $ is a fixed probability level.}
\end{align}
When $ \bar{p} $ is close to 1, the upper level set, $ \{ \hat{\hat{\pmb{\alpha}}}_{i}, \hat{\hat{\pmb{\epsilon}}}_{it}: \hat{f}(\hat{\hat{\pmb{\alpha}}}_{i},\hat{\hat{\pmb{\epsilon}}}_{it} | \bar{\pmb{x}} )\geq\delta(\bar{p}) \} $, is close to the support of the conditional distribution. 

To estimate $ \delta(\bar{p}) $, let 
\begin{align}\nonumber
\hat{H}(\gamma) = \frac{1}{NT}\sum_{i,t}1\{\hat{f}(\hat{\hat{\pmb{\alpha}}}_{i},\hat{\hat{\pmb{\epsilon}}}_{it} | \bar{\pmb{x}} )\leq \gamma \} \text{ be the estimate of } H(\gamma) = \pr(f(\hat{\pmb{\alpha}}_{i},\hat{\pmb{\epsilon}}_{it} | \bar{\pmb{x}} )\leq \gamma), 
\end{align}
and let the estimate of $ (1-\bar{p}) $-quantile of the law of $ f(\hat{\pmb{\alpha}}_{i},\hat{\pmb{\epsilon}}_{it} | \bar{\pmb{x}} ) $ be $ \gamma(\bar{p}) = \inf\{\gamma\in\mathbb{R}: \hat{H}(\gamma)\geq 1- \bar{p}\} $. $ \gamma(\bar{p}) $ can be  computed by considering the order statistic induced by the sample: $ \hat{f}(\hat{\hat{\pmb{\alpha}}}_{1},\hat{\hat{\pmb{\epsilon}}}_{1,1} | \bar{\pmb{x}} ), \ldots, \hat{f}(\hat{\hat{\pmb{\alpha}}}_{N},\hat{\hat{\pmb{\epsilon}}}_{N,T} | \bar{\pmb{x}} ) $. \citeauthor{cadre:2013} note that whenever $ H(\gamma) $ is continuous at $ \gamma(\bar{p}) $, then $ \delta(\bar{p})=\gamma(\bar{p}) $. Following \citetalias{imbens:2009}, for the application in section 4, we set $\bar{p} =0.975 $.

Given the estimate $ \hat{\mathcal{A}}(\bar{\pmb{x}}) $, we can estimate $ \tilde{G}(.) $ and $ P(.)$ in (\ref{eq:17}) at $ \bar{\pmb{x}} $ as
\begin{align}\nonumber
&\hat{\tilde{G}}(\bar{\pmb{x}})=\frac{1}{NT}\sum_{i,t}\Phi(\bar{\pmb{x}}^{\prime}\hat{\pmb{\varphi}}  +\hat{\pmb{\varphi}}_{\alpha}\hat{\hat{\pmb{\alpha}}}_{i}
+\hat{\pmb{\varphi}}_{\epsilon}\hat{\hat{\pmb{\epsilon}}}_{it})1[(\hat{\hat{\pmb{\alpha}}}_{i},
\hat{\hat{\pmb{\epsilon}}}_{it}) \in \hat{\mathcal{A}}(\bar{\pmb{x}})] \text{ and } \\ \label{eq:20}	 &\hat{P}(\bar{\pmb{x}})=\frac{1}{NT}\sum_{i,t}1[(\hat{\hat{\pmb{\alpha}}}_{i},
\hat{\hat{\pmb{\epsilon}}}_{it}) \notin \hat{\mathcal{A}}(\bar{\pmb{x}})] \text{ respectively. }
\end{align}
Now that we can estimate $ \tilde{G}(.) $ and $ P(.)$ at any $\pmb{x}$, the bounds on APE in (\ref{eq:16}), too, can be computed.

In Appendix D we derive the asymptotic covariance matrix of the second-stage coefficient estimates when the first stage estimation involves estimating a system of regression using the method in \cite{biorn:2004}. Given the covariance matrix of the second-stage coefficient estimates, we also derive the confidence intervals (CIs) proposed in \citet{imbens:2004} for the partially identified APEs. %These CIs asymptotically cover the true value of the APEs that lies within the bounds with fixed probability. 

However, first, because the expressions needed to compute the covariance matrices might be computationally involved, and secondly, because new expressions for the covariance matrix of the second-stage coefficient estimates will have to be derived when a different estimator for the first stage reduced form is employed,  we suggest that bootstrapping procedure be employed to approximate the variance of the estimated coefficient. To obtain bootstrap standard errors for control function methods, both parts of the estimation are included for every bootstrap sample \citep[see][]{wooldridge:2015}, where resampling, as in \citetalias{papke:2008}, can be done at the level of cross-sectional unit.

\section{Monte Carlo Experiments}

In this section we discuss the results of the Monte Carlo (MC) experiments, which we conduct to analyze the finite sample behaviour of our model and compare the estimates of APEs from ours and alternative estimators to the true measures of the APEs. Since we want to compare the performance of our estimator to the performances of alternative estimators with setups similar to ours, such as that of \citetalias{papke:2008}'s, which has a single endogenous regressor, we first conduct the simulation exercise with one endogenous variable, $ x $. And since our method allows for multiple endogenous regressors, we also experiment with two endogenous regressors, $ (x_{1}, x_{2}) $.     

In the first simulation exercise, we consider the following data generating process (DGP):
\begin{align}\label{eq:21}
& y_{it} = 1\{\varphi x_{it}  + \theta_{i} + \zeta_{it}>0\} \text{ and 0 otherwise, where }\\\label{eq:22}
& x_{it}=  \pi z_{it} + \alpha_{i} + \epsilon_{it}, i = 1,\ldots, n, t = 1, \ldots, 5,
\end{align}
and where $ z_{it} $ is the instrument. We assume that  $ \varphi=-1  $ and that $ \pi = 1.5 $.  We allow the individual specific effects $ \alpha_{i} $ and $ \theta_{i} $ to be correlated with the vector of instruments, $Z_{i} = ( z_{i1}, \ldots, z_{i5} )^{\prime} $. The $ z_{it} $'s are i.i.d and marginally distributed as $ \N[0,\sigma^{2}_{z} ] $, where  $ \sigma_{z} = 5 $.  The variables, $ Z_{i} $, $ \alpha_{i}, $ and $ \theta_{i} $, are drawn from the following distribution: $ (Z^{\prime}_{i}, \alpha_{i}, \theta_{i})^{\prime} \sim \N\begin{bmatrix} 0,  \Sigma_{z\alpha\theta} \end{bmatrix}$, where $ \sigma_{\alpha} =3 $,  $ \sigma_{\theta} = 4 $,  $ \rho_{z\alpha} = 0.4 $,  $ \rho_{z\theta} = 0.2 $, and $ \rho_{\alpha\theta} = 0.5 $. The above choice of correlation coefficients ensures that, conditional on  $ \alpha_{i} $, the conditional correlation between $z_{it}$ and $ \theta_{i} $, $ \rho_{z\theta|\alpha} = \rho_{z\theta} - \rho_{z\alpha}\rho_{\alpha\theta} = 0 $, which, in this case, also implies that conditional on  $ \alpha_{i} $,  $ \theta_{i} \indep  Z_{i} | \alpha_{i} $.

In accordance with assumption AS \ref{as:1}, we assume that $ (\pmb{\zeta}_{i}, \pmb{\epsilon}_{i}) \indep  (Z^{\prime}_{i}, \alpha_{i}, \theta_{i})^{\prime} $, and draw $ (\zeta_{it}, \epsilon_{it}) $ from $ \N\begin{bmatrix} 0,  \Sigma_{\zeta\epsilon} \end{bmatrix}$, where the elements $ \sigma^{2}_{\zeta} $, $ \sigma^{2}_{\epsilon} $, and $ \rho_{\zeta \epsilon} $ of $ \Sigma_{\zeta\epsilon} $ are assumed as $ \sigma^{2}_{\zeta} = \sigma^{2}_{\epsilon} =1$ and $ \rho_{\zeta \epsilon} =0.75 $. The DGP assumptions, $ \theta_{i} \indep  Z_{i} | \alpha_{i} $ and $ (Z^{\prime}_{i}, \alpha_{i}, \theta_{i})^{\prime} \indep (\pmb{\zeta}_{i}, \pmb{\epsilon}_{i})$,  together satisfy AS \ref{as:2} (a) \footnote{Given the DGP assumptions, it can be verified that $ \pmb{\zeta}_{i} \indep  Z_{i} | \alpha_{i},  \pmb{\epsilon}_{i}, \theta_{i} $ and $ \theta_{i} \indep  Z_{i} | \alpha_{i},  \pmb{\epsilon}_{i} $; these two then imply that $ \theta_{i},\pmb{\zeta}_{i} \indep  Z_{i} | \alpha_{i}, \pmb{\epsilon}_{i} $ or equivalently $ \theta_{i}, \pmb{\zeta}_{i} \indep  X_{i} | \alpha_{i},  \pmb{\epsilon}_{i}$.}, and since $ (\zeta_{it}, \epsilon_{it})$ are i.i.d., they also satisfy AS \ref{as:2} (b).

From this DGP we generate $ (Z^{\prime}_{i}, \alpha_{i}, \theta_{i})^{\prime} $ and $ (\zeta_{it}, \epsilon_{it})$ of varying size, $ n $, with $ t $ fixed at $t=5$.  We then discretized $z_{it}$ to take value 1 if $z_{it}>0$  and 0 otherwise. Having generated $ (Z^{\prime}_{i}, \alpha_{i}, \theta_{i})^{\prime} $ and $ (\zeta_{it}, \epsilon_{it})$, we generate $ x_{it} $ according to (\ref{eq:22}) and then  $ y_{it} $ according to (\ref{eq:21}). 

% When generating $ (Z^{\prime}_{i}, \alpha_{i}, \theta_{i})^{\prime} $, we first generate $ n $ number of $ \alpha_{i} $'s, and then given $ \alpha_{i} $, we generate $ \theta_{i} $ and $ Z_{i}= ( z_{i1}, \ldots, z_{i5} )^{\prime} $.

We have showed that by estimating 
\begin{align}\label{eq:25}
y_{it}=1\{ \varphi x_{it} + \varphi_{\alpha}\hat{\alpha}_{i} +\varphi_{\epsilon}\hat{\epsilon}_{it}+ \eta_{it}>0\},
\end{align} 
where $ \hat{\alpha}_{i} $ and $\hat{\epsilon}_{it} $ are the control variables, as a probit model we can obtain consistent estimates of the ASF and the APE. The control variables, $ \hat{\alpha}_{i} $ and $\hat{\epsilon}_{it} $, are obtained from the estimates of the reduced form equation (\ref{eq:22}), augmented with $ \bar{\pi}\bar{z}_{i} = \bar{\pi}T^{-1}\sum_{t=1}^{T}z_{it}$, which is estimated as a random effect model by the method of MLE. The APE at $ \bar{x} $ after estimating equation (\ref{eq:25}) could be obtained by averaging  
\begin{align}\label{eq:31}
\frac{1}{\Delta x}\Phi\biggr(\frac{\varphi x_{it}   + \varphi_{\alpha}\hat{\alpha}_{i} +\varphi_{\epsilon}\hat{\epsilon}_{it}}{\sigma_{\eta}}\biggr)
\end{align} 
over $\hat{\alpha} $ and $ \hat{\epsilon} $ at $x_{it} = \bar{x} + \Delta x$ and $x_{it} = \bar{x}$ and taking the difference. In the above,  $ \sigma^{2}_{\eta} $ is the variance of $ \eta_{it} $ in equation (\ref{eq:25}) and $ \Phi $ is the standard cumulative normal density function. 

Now, while in practice the heterogeneity terms, $(\theta_{i} ,\zeta_{it}) $ and  $(\alpha_{i} ,\epsilon_{it}) $, are unobserved, in MC experiments we do know what these values are. By averaging  $1\{x_{it}\varphi + \theta_{i} + \zeta_{it} >0\}  $ over  $(\theta_{i} ,\zeta_{it})$ at $ x_{it}= \bar{x} $ to obtain $  G(\bar{x}) $ and the same at $ x_{it} = \bar{x} + \Delta x $ to obtain $ G(\bar{x} + \Delta x ) $, we could compute the true measure of APE, $\frac{\partial G(x_{it}) }{\partial x}$, at $ x_{it} =\bar{x} $ by computing $\frac{G(\bar{x}+\Delta x) - G(\bar{x})}{\Delta x} $. For the exercise, we chose  $ \bar{x} = 1 $ and $ \Delta x =0.05 $. Since we average over realizations of $(\theta_{i} ,\zeta_{it})$, there is some variability in the values of $ \frac{\partial G(\bar{x}) }{\partial x}$ over the replications; the average over the replications for every sample size is reported in the tables containing the results. For notational convenience we will denote the true APE by  $ \frac{\partial G(\bar{x}) }{\partial x} $.  Estimates of APE from any of the model considered in this section will be denoted by $ \widehat{\frac{\partial G(\bar{x}) }{\partial x}} $.  

One of the alternative estimators, which has its set-up similar to ours is the method proposed by \citetalias{papke:2008}.  To address the issue of endogeneity, \citetalias{papke:2008} also propose a two-step control function method. They first assume that $  \theta_{i} = \E(\theta_{i}| Z_{i})+\tau_{i} = \bar{\pi}_{\theta}\bar{z}_{i} + \tau_{i} $  and  $ \alpha_{i} = \E(\alpha_{i}| Z_{i})+ a_{i} =\bar{\pi}_{\alpha}\bar{z}_{i} + a_{i} $, where $ \bar{z}_{i} = T^{-1}\sum_{t=1}^{T}z_{it} $.  Given the assumptions, they write the triangular system in (\ref{eq:21}) and (\ref{eq:22}) as  
\begin{align}\label{eq:32}
& y_{it} = 1\{\varphi x_{it}  + \bar{\pi}_{\theta}\bar{z}_{i} + \tau_{i}+ \zeta_{it}>0\}\\ \label{eq:33}
& x_{it}=  \pi z_{it} + \bar{\pi}_{\alpha}\bar{z}_{i} +  \upsilon_{PWit}, 
\end{align}
where $ \upsilon_{PWit} = a_{i} + \epsilon_{it} $. They then make the control function assumption that $ \tau_{i}+ \zeta_{it} \indep x_{it} | \upsilon_{PWit} $. This allows them to estimate the APE at $x_{it}=\bar{x}$ by averaging      
\begin{align}\nonumber
\frac{1}{\Delta x}\Phi\biggr(\frac{\varphi x_{it}  + \bar{\pi}_{\theta}\bar{z}_{i} + \rho\upsilon_{PWit}}{\sigma_{PW}}\biggr)    
\end{align} 
over $ \bar{\pi}_{\theta}\bar{z}$ and $ \upsilon_{PW} $ at $x_{it} = \bar{x} + \Delta x$ and $x_{it} = \bar{x}$ and taking the difference. In the above, $\rho $ is population regression coefficient of $\tau_{i} + \zeta_{it} $ on $ \upsilon_{PWit} $, and where $\upsilon_{PWit} $ is obtained as residuals after estimating (\ref{eq:33}) in the first stage. The conditional distribution of $\tau_{i} + \zeta_{it} $ given $ \upsilon_{PWit} $ is assumed to follow a normal distribution with variance $ \sigma^{2}_{PW} $. 
If their method gives consistent estimates of APE, then it must be that the above measure is equal to $ \frac{\partial G(\bar{x}) }{\partial x}$. 

In \citeauthor{chamberlain:1984}'s correlated random effects (CRE) probit and in \citeauthor{chamberlain:1984}'s conditional logit (CL),  $ x_{it} $ is assumed to be independent of the idiosyncratic term, $ \zeta_{it} $. While in CRE probit model $\E(\theta_{i}| X_{i}) $ is specified, in the CL model the distribution of  $\theta_{i}$ is left unspecified. Assuming that $  \theta_{i} = \bar{\pi}_{\theta}\bar{x}_{i} + \tau_{i} $, where $ \bar{\pi}_{\theta}\bar{x}_{i} $ is the specification for $\E(\theta_{i}| X_{i}) $, the structural equation for the CRE probit model is given by 
\begin{align}\nonumber
y_{it} = 1\{\varphi x_{it}  + \bar{\pi}_{\theta}\bar{x}_{i} +  \tau_{i}  + \zeta_{it}>0\}, \text{ where }  \tau_{i} = \theta_{i} - \E(\theta_{i}|X_{i}). 
\end{align}
$ \tau_{i}  + \zeta_{it} $ is assumed independent of $ X_{i} $ and is distributed normally with variance $ \sigma^{2}_{CRE} $. The CRE probit model is estimated as a probit model by pooling the data. If the CRE probit model, too, gives consistent measure of APE then it has to be that
\begin{align}\nonumber
\frac{1}{\Delta x}\biggr[\int\Phi\biggr(\frac{\varphi (x_{it}+ \Delta x) + \bar{\pi}_{\theta}\bar{x}_{i}}{\sigma_{CRE}}\biggr)dF(\bar{\pi}_{\theta}\bar{x})-\int\Phi\biggr(\frac{\varphi x_{it} + \bar{\pi}_{\theta}\bar{x}_{i}}{\sigma_{CRE}}\biggr)dF(\bar{\pi}_{\theta}\bar{x})\biggr]=\frac{\partial G(x_{it}) }{\partial x},  
\end{align} 
where the LHS is the measure of APE of $ x $ at $ x_{it} $ pertaining to the CRE probit model.

The structural equation for the CL model is same as equation (\ref{eq:21}), where $ \zeta_{it} $ follows a logistic distribution. The APE of $ x $ at $ x_{it} $ for the CL model is \begin{align}\nonumber
&\frac{1}{\Delta x}\biggr[\int\Lambda(x_{it}+\Delta x, \theta_{i}) dF(\theta)- \int \Lambda(x_{it}, \theta_{i})dF(\theta)\biggr], 
\end{align}
where  $ \Lambda(x_{it}, \theta_{i}) = \Pr (y_{it} = 1|x_{it}, \theta_{i}) = \frac{\exp(\varphi x_{it}  + \theta_{i})}{1 + \exp(\varphi x_{it}  + \theta_{i})} $. Once we have estimated $ \varphi $ by estimating the CL model, we can estimate the APE by averaging $ \Lambda(x_{it} + \Delta x, \theta_{i}) $  and  $ \Lambda(x_{it}, \theta_{i}) $ over $ \theta_{i} $ and taking the difference.  

Table \ref{table:mc1} provides the results for various sample size, $ n $, with $ m = 2000 $ Monte Carlo replications. In the Table and in Figure \ref{fig:fig1} we compare the performance of our method, which we term CRECF\footnote{The acronym derives from fact that the control functions are based on correlated random effects in the reduced form equations.} method, to the alternative estimators considered above.

\begin{table}[h!]
	\caption{\textbf{Performance of the APE, $\frac{\partial G(\bar{x}=1)}{\partial x} $, for alternative estimators. }}
	\centering
	\resizebox{\columnwidth}{!}{
		
		\begin{tabular}{l|c|c|c|c|c|c|c|c|c}
			\hline\hline
			& True APE	&\multicolumn{2}{c|}{CRECF Method} 	&\multicolumn{2}{c|}{ \citeauthor{papke:2008}} & \multicolumn{2}{c|}{\citeauthor{chamberlain:1984}'s CRE Probit} & \multicolumn{2}{c}{\citeauthor{chamberlain:1984}'s Logit} \\\hline
			& Mean   	&	RMSE	&	Mean &	RMSE	&	Mean 	&	RMSE	&	Mean 	&	RMSE	&	Mean\\\hline			
			$N$= 200 & -.0931	&	 .0445	&	-.0920	&	.0724	&	-.0354	&	.0561	&	-.0578	&	.0473	&	-.1070	\\
			$N$= 500& -.0944	&	.0283	&	-.0932	&	.0654	&	-.0353	&	.0462	&	-.0578	&	.0310	&	-.1061	\\			
			$N$= 1000 & -.0935	&	.0203	&	-.0936	&	.0616	&	-.0353	&	.0408	&	-.0579	&	.0238	&	-.1057	\\		
			$N$= 2000 &  -.0934	&	.0143	&	-.0936	& .0597	&	-.0353	&	.0381	&	-.0579	&	.0189	&	-.1057	\\
			$N$= 5000 & -.0939	&	.0088	&	-.0936	&	.0592	&	-.0353	&	.0370	&	-.0579	&	.0144	&	-.1053	\\
			\hline\hline
			\multicolumn{10}{l}{RMSE is Root Mean Square Error and Mean is the mean value of $ m=2000 $ APEs.}
	\end{tabular} }
	\label{table:mc1}
\end{table}

In Figure \ref{fig:fig1} we plot the densities of $ m = 2000 $ MC estimates of $\widehat{\partial G(x_{it})}/\partial x - \partial G(x_{it})/\partial x$, at $x_{it}=1 $, obtained for the four estimation methods for different sample sizes. It can be seen from the figure that for each of the alternative estimators, the APE of $x$ is estimated with a bias, which persists as the sample size grows larger. Thus, even as the variance of $\widehat{\partial G(x_{it})}/\partial x - \partial G(x_{it})/\partial x$ for each of the alternative methods decreases, the root mean square error (RMSE) for alternative methods in Table \ref{table:mc1} decreases quite slowly.

\begin{figure}[h!]
	\psfrag{x}{\tiny{$\widehat{\partial G(\bar{x})/\partial x} - \partial G(\bar{x})/\partial x$}}
	\psfrag{CRE Control Function}{\tiny{CRECF Method}}
	\psfrag{Density}{\tiny{Density}}
	\psfrag{Condtitional Logit}{\tiny{\citeauthor{chamberlain:1984}'s Logit}}
	\psfrag{Chamberlain's CRE Probit}{\tiny{\citeauthor{chamberlain:1984}'s Probit}}
	\psfrag{Papke & Wooldridge (2008)}{\tiny{\citeauthor{papke:2008}}}
	
	\centering
	\begin{subfigure}{.5\textwidth}
		\centering
		\includegraphics[width=.925\linewidth,natwidth=353,natheight=350]{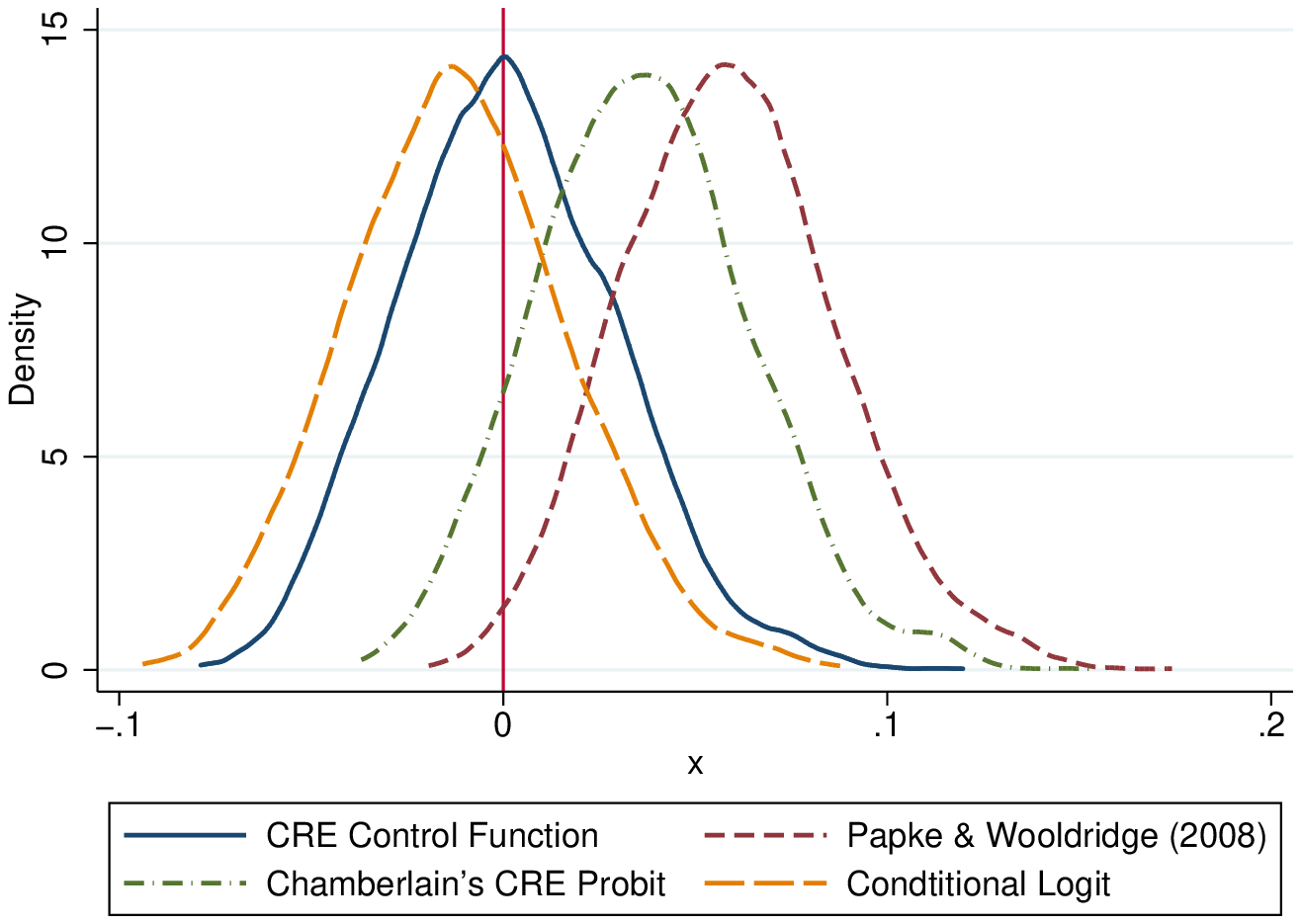}
		\caption{$N$=500}
	\end{subfigure}%
	\begin{subfigure}{.5\textwidth}
		\centering
		\includegraphics[width=.925\linewidth,natwidth=353,natheight=350]{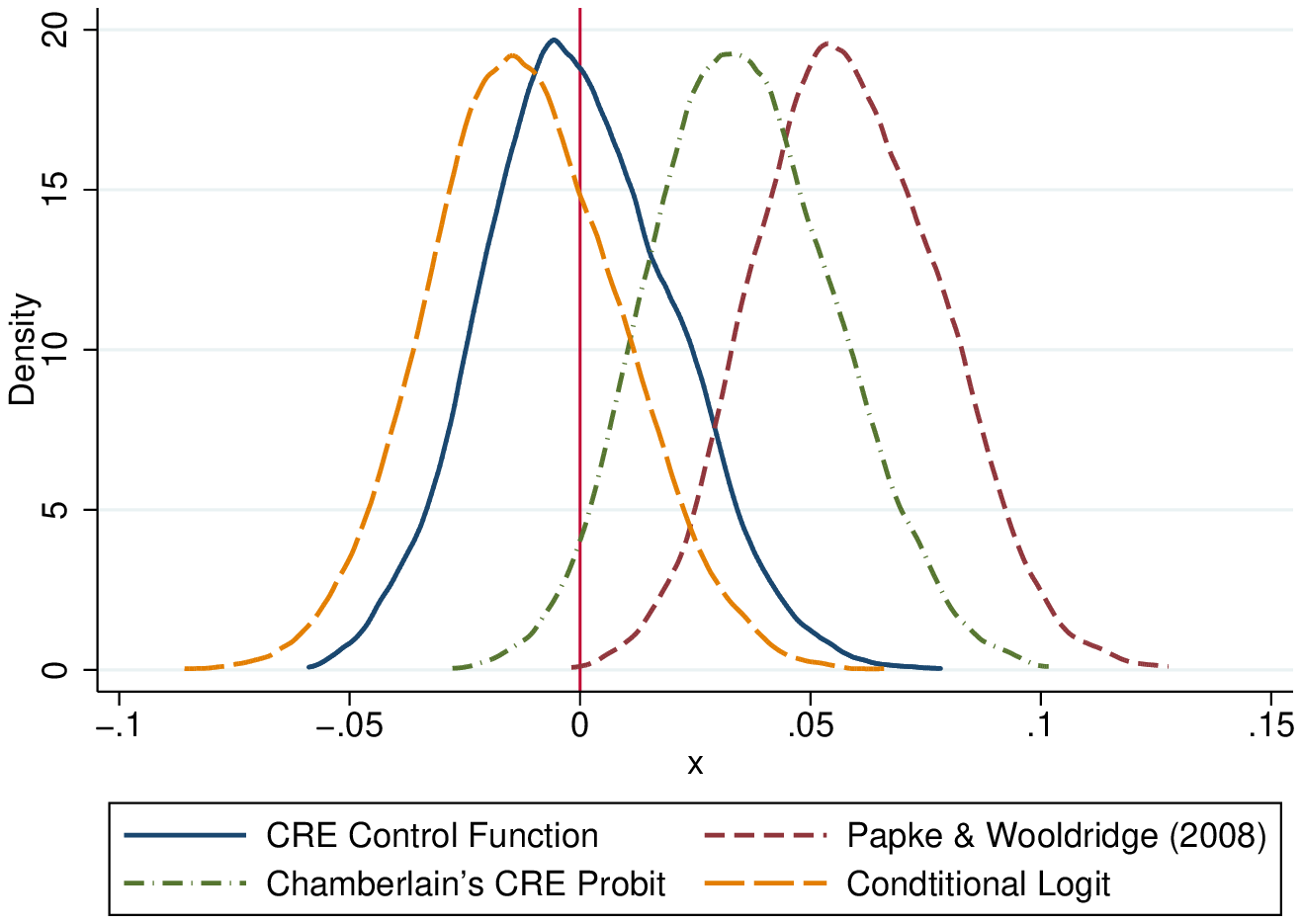}
		\caption{$N$=1000}
	\end{subfigure}\\
	\centering
	\begin{subfigure}{.5\textwidth}
		\centering
		\includegraphics[width=.925\linewidth,natwidth=354,natheight=350]{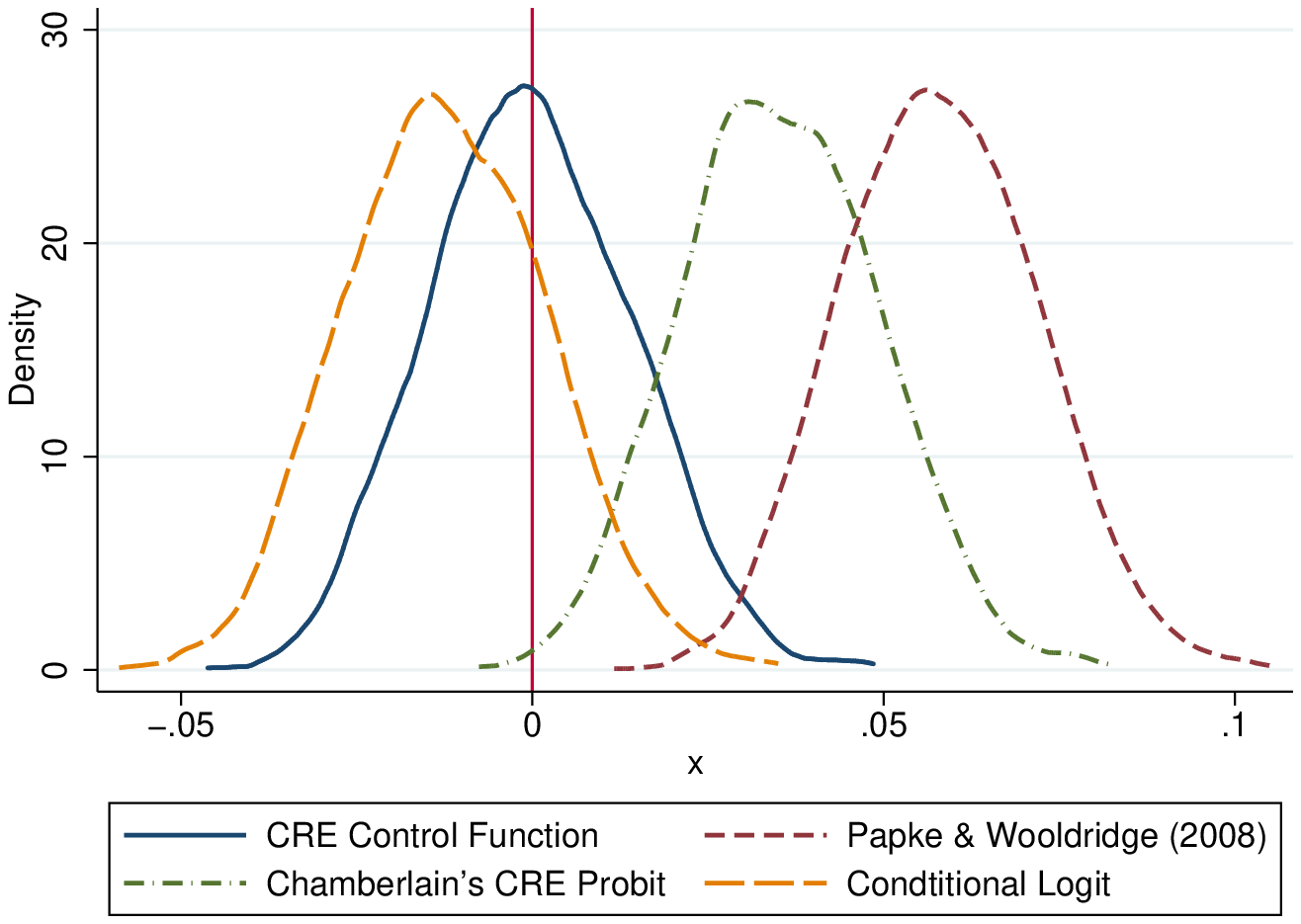}
		\caption{$N$=2000}
	\end{subfigure}%
	\begin{subfigure}{.5\textwidth}
		\centering
		\includegraphics[width=.925\linewidth,natwidth=353,natheight=350]{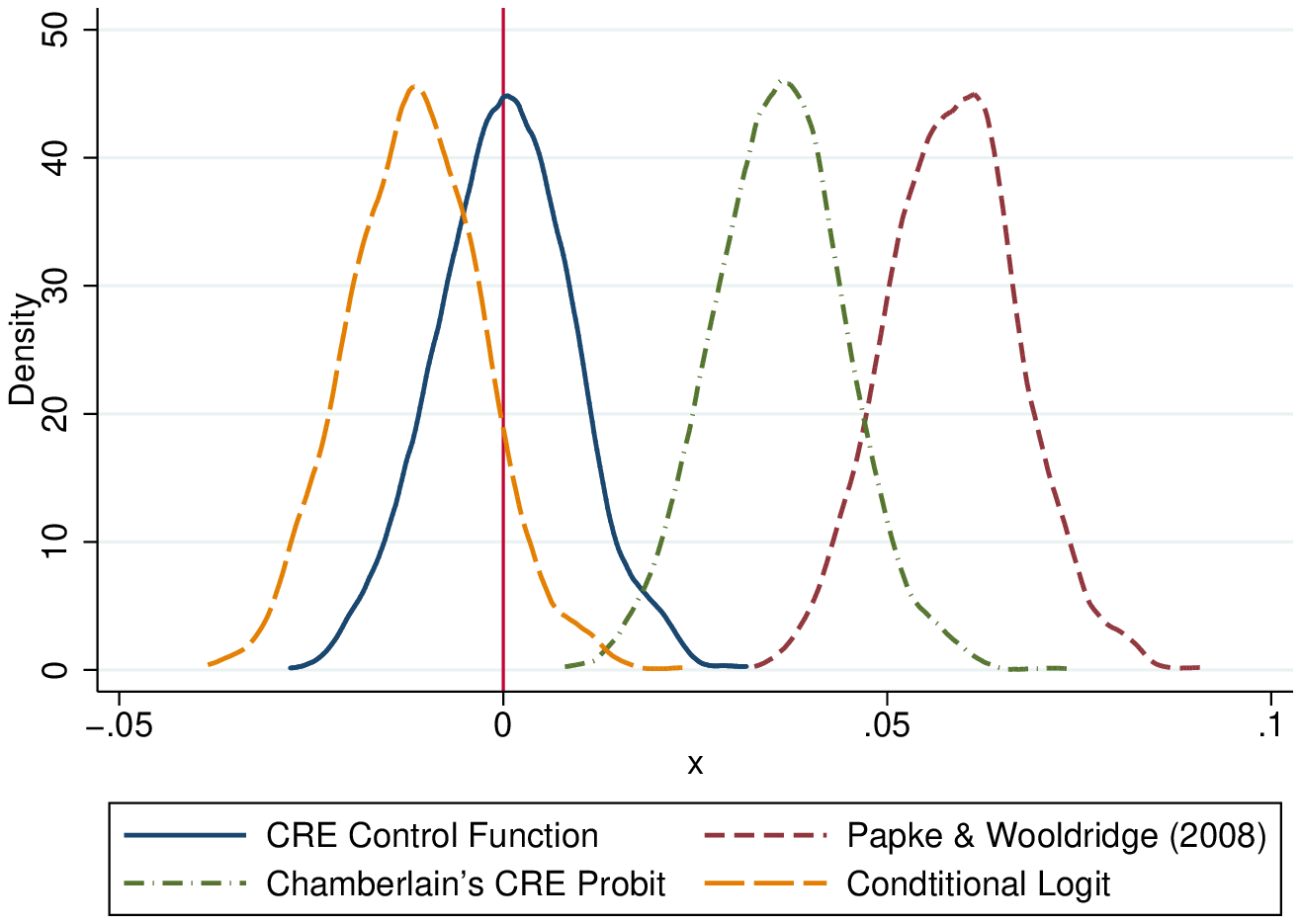}
		\caption{$N$=5000}
	\end{subfigure}
	\caption{Comparison with Alternative Estimators: Density of $\widehat{\partial G(\bar{x})/\partial x} - \partial G(\bar{x})/\partial x$ at $ \bar{x}=1 $ for different Sample Size.  }
	\label{fig:fig1}
\end{figure}

Since the CL and CRE probit models do not account for the endogeneity of $ x_{it} $ with respect to the transitory errors, $ \zeta_{it} $, the methods can give biased results. Unexpectedly, however, the method proposed by \citetalias{papke:2008}, which tries to accounts for correlation between $ Z_{i} $ and $ \theta_{i} $ and the correlation of $ x_{it} $ with both $ \theta_{i} $ and $ \zeta_{it} $, gives the least satisfactory results. This suggests that under a more general DGP, as in our MC experiments, their control function assumption that $ \tau_{i}+ \zeta_{it} \indep X_{i} | \upsilon_{PWit} $, where both $ \tau_{i}+ \zeta_{it} $ and $ \upsilon_{PWit} $ are assumed to be independent of $ Z_{i} $, is, as discussed in remark \ref{rm:2} and proposition \ref{prop:1}, likely to get violated.

To validate the claims made in remark \ref{rm:2}, we conduct some MC experiments to compare the APE when $V_{PW,i} \equiv\{\upsilon_{PW,i1}, \ldots, \upsilon_{PW,iT}\}$ is used as a control function as against when only $\upsilon_{PW,it}$ is used a control function. In Figure \ref{fig:fig2} below, we have plotted the density of difference between the estimated and the true APEs, $\widehat{\partial G(x_{it})/\partial x} - \partial G(x_{it})/\partial x$, where the estimated APEs are the APEs that are obtained by varying the control functions and the instruments in \citetalias{papke:2008}'s model. As can be seen in figure, when the instrument is continuous with a large support, employing $V_{PW,i}$ yields consistent estimates of the APE, whereas if only $\pmb{\upsilon}_{PW,it}$ is employed as control function, then we get estimates that are biased. When the same instrument is discretized to take value 1 and 0, we get biased estimates when either $V_{PW,i}$ or $\pmb{\upsilon}_{PW,it}$ is employed as control function. This is because when the instrument is discrete, even though $V_{PW,i}$ is the appropriate control function, the APE is not point but partially identified.

\begin{figure}[h!]
	\psfrag{x}{\small{$\widehat{\partial G(\bar{x})/\partial x} - \partial G(\bar{x})/\partial x$}}
	\psfrag{Continuous Instrument, CF: v_it}{\small{Cnt. Inst., CF: $\upsilon_{PW,it}$}}
	\psfrag{Binary Instrument, CF: v_it}{\small{Binary Inst., CF: $\upsilon_{PW,it}$}}
	\psfrag{Continuous Instrument, CF: V_i}{\small{Cnt. Inst., CF: $V_{PW,i}$}}	
	\psfrag{Binary Instrument, CF: V_i}{\small{Binary Inst., CF: $V_{PW,i}$}}	
	\psfrag{Density}{\small{Density}}
	\centering
	\includegraphics[width=.65\linewidth,natwidth=353,natheight=275]{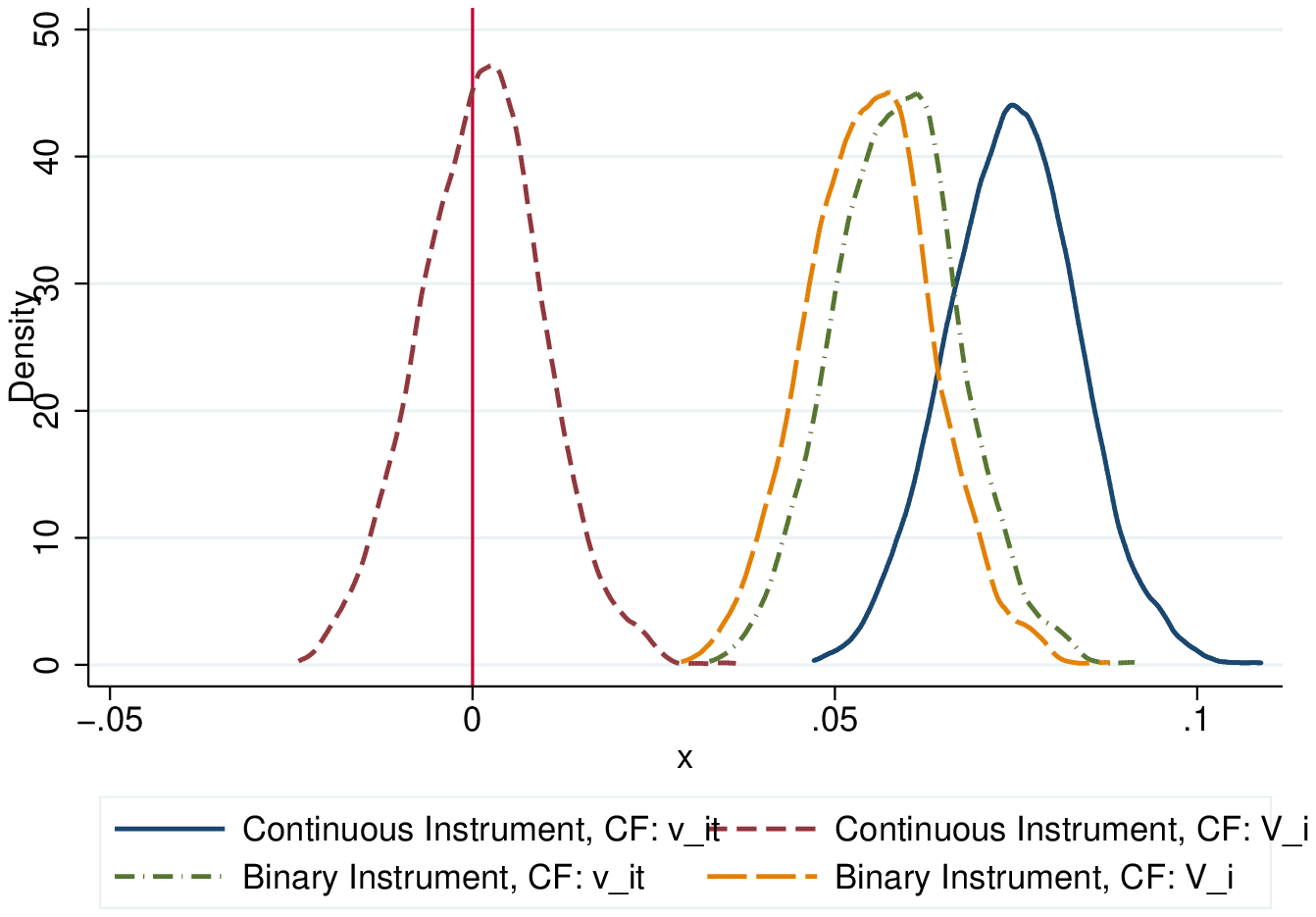}
	\caption{Density of difference between the estimated and the true APEs, $\widehat{\partial G(\bar{x})/\partial x} - \partial G(\bar{x})/\partial x$, at $ \bar{x}=1 $ where the estimated APEs, $ \widehat{\partial G(\bar{x})/\partial x} $, are obtained by varying the control functions and the instruments in \citetalias{papke:2008}'s model. Sample Size: $N$=5000}
	\label{fig:fig2}
\end{figure}

That the APEs when using the traditional control function, as in \citetalias{papke:2008}'s, may not be point identified when the instrument, $ z_{it} $, is binary even when $ x_{it} $ has a large support can be seen in Figure \ref{fig:fig3} (c), which has the plots of the level sets of the kernel estimates of the joint density of ($ x $, $ \upsilon_{PWit} $). The figure suggests that the common support requirement for point identification may be satisfied only over a small range of $ x $ values in \citetalias{papke:2008}'s model. Whereas from Figure \ref{fig:fig3} (a) and (b), we can deduce that the support of the conditional distribution of $(\hat{\epsilon}_{it}, \hat{\alpha}_{i}) $ given $x$ is almost the same for large ranges of $ x $.

\begin{figure}[h!]
	\psfrag{x}{$x$}
	\psfrag{v}{$ \upsilon_{it} $}
	\psfrag{ex_x}{$ \hat{\alpha}_{i} $}
	\psfrag{ep_x}{ $\hat{\epsilon}_{it} $}
	
	\centering
	\begin{subfigure}{.5\textwidth}
		\centering
		\includegraphics[width=.925\linewidth,natwidth=355,natheight=275]{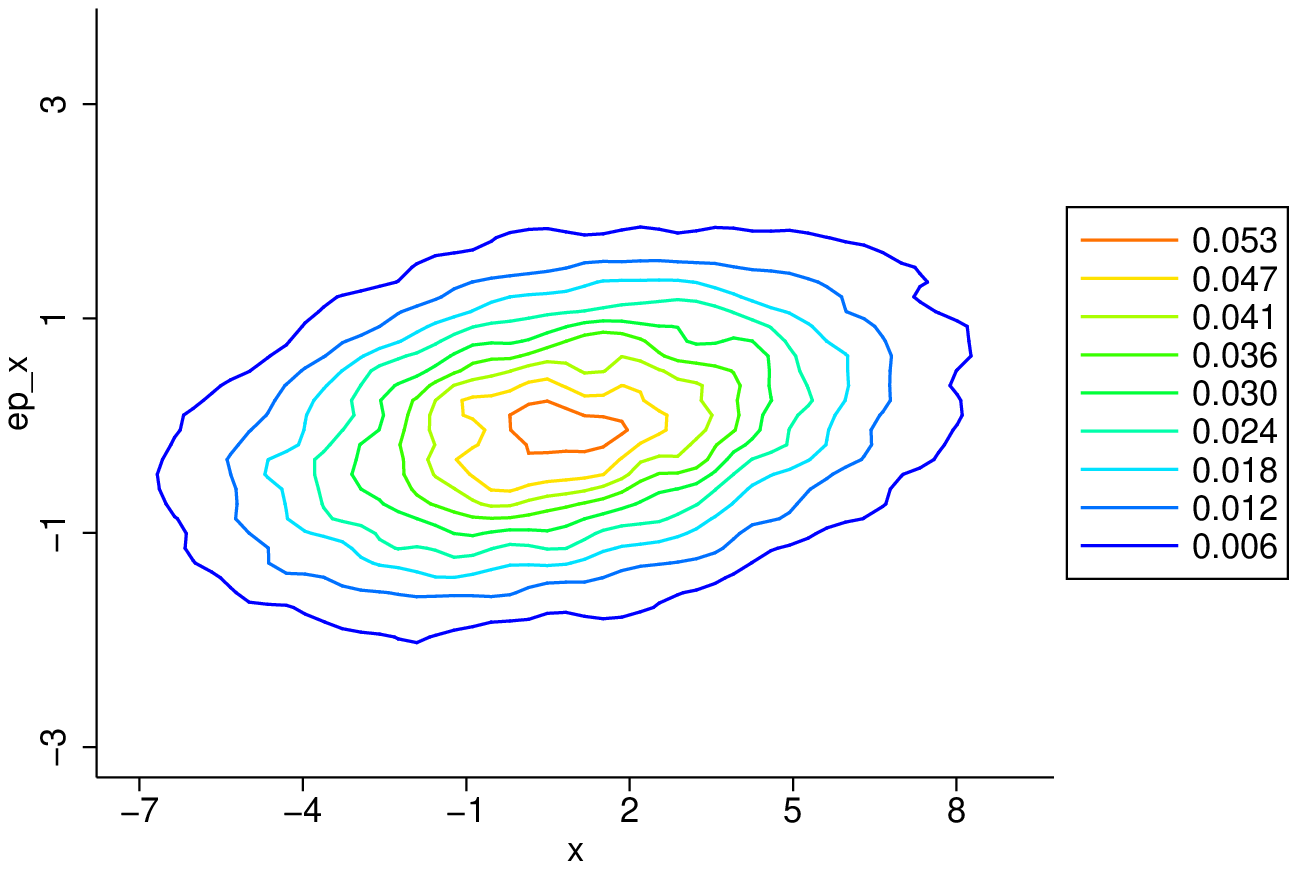}
		\caption{Joint Density of $\hat{\epsilon}_{it} $ and $ x $.}
	\end{subfigure}%
	\begin{subfigure}{.5\textwidth}
		\centering
		\includegraphics[width=.925\linewidth,natwidth=353,natheight=275]{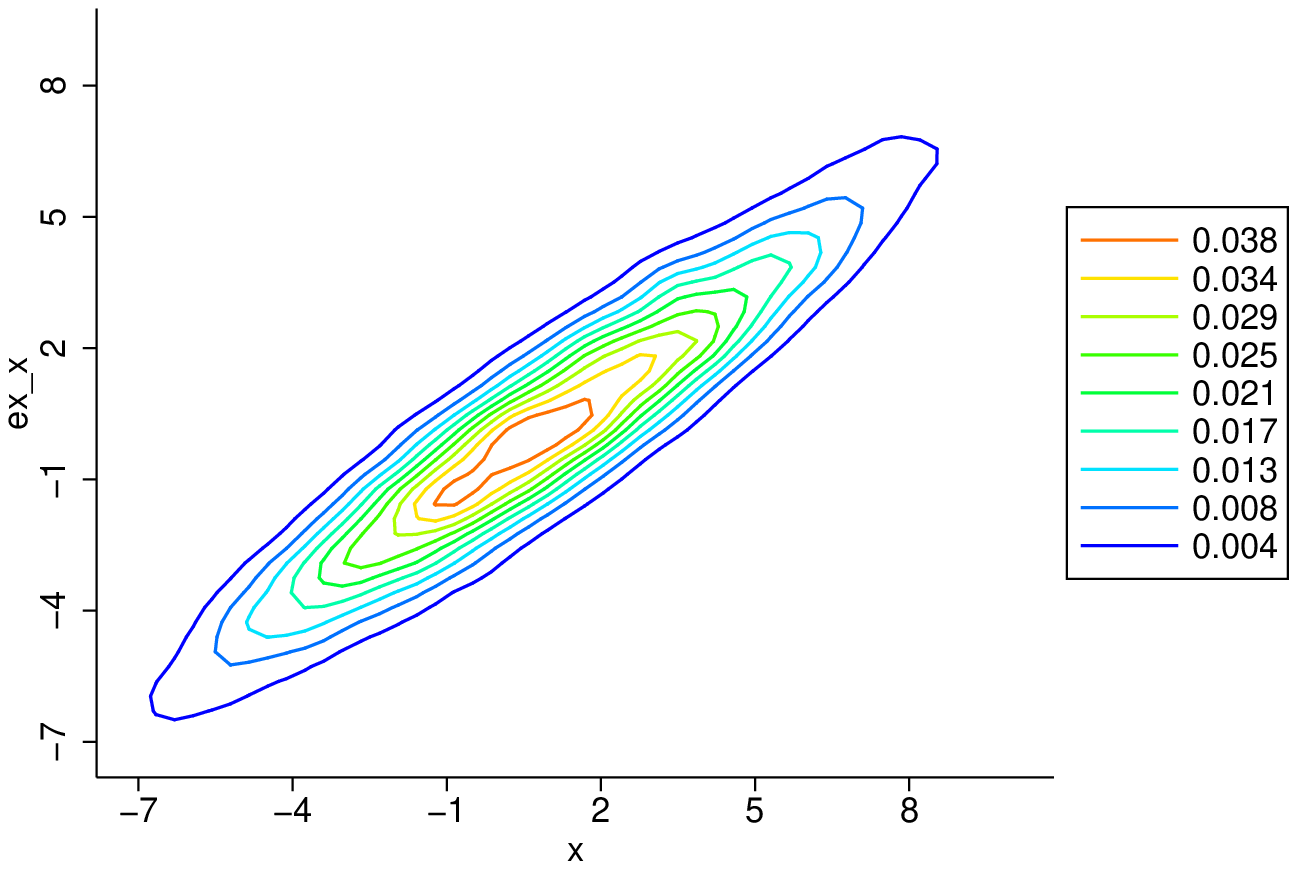}
		\caption{Joint Density of $\hat{\alpha}_{i} $ and $ x $.}
	\end{subfigure}\\
	\centering
	\begin{subfigure}{.5\textwidth}
		\centering
		\includegraphics[width=.925\linewidth,natwidth=354,natheight=300]{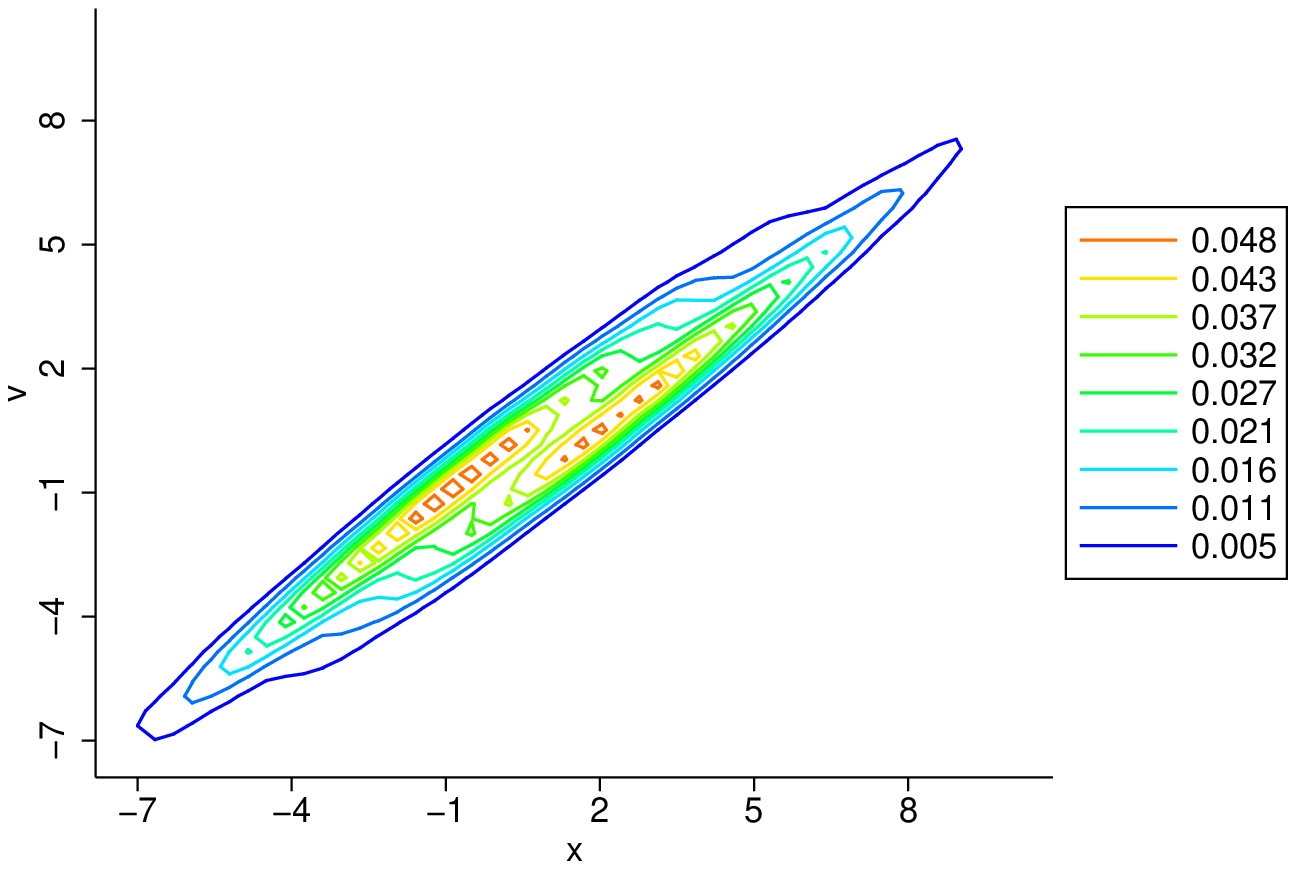}
		\caption{Joint Density of $\upsilon_{it} $ and $ x $.}
	\end{subfigure}%
	\begin{subfigure}{.5\textwidth}
		\centering
		
	\end{subfigure}
	\caption{Level curves of estimated joint density of $ x $ and various control functions.  }
	\label{fig:fig3}
\end{figure}

As our method allows for multiple endogenous regressors, we also conduct a simulation exercise with two endogenous regressors. The two instruments, $ \pmb{z}_{it} = (z_{1it}, z_{2it}) $, are i.i.d and marginally distributed as $ \N[0,\sigma^{2}_{z} ] $, where  $ \sigma_{z_{1}} = 5 $, $ \sigma_{z_{2}} = 2 $, and $ \rho_{z_{1}z_{2}} = 0.25 $.  The instruments, $ Z^{\prime}_{i} \equiv \{\pmb{z}_{i1}, \ldots, \pmb{z}_{i5} \} $ and the individual effects are correlated, and follow the joint distribution:  $ (Z^{\prime}_{i}, \alpha_{1i}, \alpha_{2i}, \theta_{i})^{\prime} \sim \N\begin{bmatrix} 0,  \Sigma_{z\alpha\theta} \end{bmatrix}$, where $ \sigma_{\alpha_{1}} =6 $, $ \sigma_{\alpha_{2}} =2 $,  $ \sigma_{\theta} = 4 $,  $ \rho_{z_{1}\alpha_{1}} = 0.2 $, $\rho_{z_{2}\alpha_{1}} = 0.3 $, $ \rho_{z_{1}\alpha_{2}} = 0.25 $, $\rho_{z_{2}\alpha_{2}} = 0.3 $, $ \rho_{\alpha_{1}\alpha_{2}} = 0.5 $,  $ \rho_{z_{1}\theta} = 0.1 $, $ \rho_{z_{2}\theta} = 0.15 $,  $ \rho_{\alpha_{1}\theta} = 0.5 $, $ \rho_{\alpha_{1}\theta} = 0.25 $. The above choice of correlation coefficients ensures that, conditional on $ \pmb{\alpha}_{i}  = (\alpha_{i2}, \alpha_{i2})  $, the conditional correlation between $\pmb{z}_{it}$ and $ \theta_{i} $ is 0. Having generated the data, we then discretize the instruments to take values 0 and 1: $ z_{1it} $ takes value 1 if it is non-negative and 0 otherwise, while $ z_{2it} $ takes value 1 if it is greater than or equal 1 and 0 otherwise. The idiosyncratic error terms  $ (\zeta_{it}, \epsilon_{1it}, \epsilon_{2it}) $ are drawn from $ \N\begin{bmatrix} 0,  \Sigma_{\zeta\epsilon} \end{bmatrix}$, where the elements of $ \Sigma_{\zeta\epsilon} $ are assumed as $ \sigma^{2}_{\zeta} = \sigma^{2}_{\epsilon_{1}} = \sigma^{2}_{\epsilon_{2}} =1$, $ \rho_{\zeta \epsilon_{1}} =0.75 $, $ \rho_{\zeta \epsilon_{2}} =0.25 $, and $ \rho_{\epsilon_{1} \epsilon_{2}} =0.5 $.

With $ Z_{i} $ and the error terms in place, we next generate $ x_{1it} $, $ x_{2it} $ and $ y_{it} $ according to:
\begin{align}\nonumber
&x_{1it} = -1 z_{1it} + .05z_{2it} + \alpha_{1i} + ep_{1it}\\\nonumber
&x_{2it} = .025z_{1it} + .75z_{2it} + \alpha_{2i} + ep_{2it} \\\nonumber 
&y_{it}  = 1\{ -1 x_{1it} + 0.5 x_{2it}   + \theta_{i} + \zeta_{it}>0\}.\\\nonumber 
\end{align} 
We compute the APEs at $\bar{x}_{1}= 0.5$, $\bar{x}_{2}=1$ and chose $ \Delta x_{1} = 0.05 $ and $ \Delta x_{2} = 0.1 $. Table \ref{table:mc2} provides the results for various sample size, $ n $, with $ m = 2000 $ Monte Carlo replications. In the Table, we compare the performance of our method to the CRE probit and the conditional logit models. Since \citetalias{papke:2008} consider only a single endogenous regressor, their model is not considered in these simulations. With lowest RMSE  for every sample size, our method outperforms the CRE probit and the conditional logit models.

\begin{table}[h!]
	\caption{\textbf{Performance of the APEs, $\frac{\partial G(\bar{x}_{1}=0.5, \bar{x}_{2}=1 )}{\partial x_{1}} $ and $\frac{\partial G(\bar{x}_{1}=0.5, \bar{x}_{2}=1 )}{\partial x_{2}} $, for alternative estimators. }}
	\centering
	\resizebox{\columnwidth}{!}{
		
		\begin{tabular}{l|c|c|c|c|c|c|c|c}
			\hline\hline
			&	& True APE	&\multicolumn{2}{c|}{CRECF Method} 	 & \multicolumn{2}{c|}{\citeauthor{chamberlain:1984}'s CRE Probit} & \multicolumn{2}{c}{\citeauthor{chamberlain:1984}'s Logit} \\\hline
			&	& Mean   	&	RMSE	&	Mean &	RMSE	&	Mean 	&	RMSE	&	Mean 	\\\hline
			$N$= 500 & $ x_{1} $ &	-0.3328	&	0.0538	&	-0.3385	&	0.0761	&	-0.3874	&	0.1114	&	-0.4294	\\	
			&  $ x_{2}$  &	0.1655	&	0.0411	&	0.1817	&	0.186	&	0.3491	&	0.2085	&	0.372	\\	\hline
			$N$= 1000 &  $ x_{1}$ &	-0.3325	&	0.0383	&	-0.3395	&	0.0659	&	-0.3869	&	0.103	&	-0.428	\\	
			& $ x_{2}$ &	0.1656	&	0.0316	&	0.1821	&	0.1845	&	0.349	&	0.2069	&	0.3715	\\	\hline
			$N$= 2000 & $ x_{1}$ &	-0.331	&	0.0282	&	-0.3402	&	0.0618	&	-0.3869	&	0.1009	&	-0.4281	\\	
			& $ x_{2}$ &	0.1662	&	0.0247	&	0.182	&	0.1828	&	0.3484	&	0.2055	&	0.3713	\\	\hline
			$N$= 5000 & $ x_{1}$ &	-0.3316	&	0.0192	&	-0.3407	&	0.0579	&	-0.387	&	0.098	&	-0.4281	\\	
			& $ x_{2}$ &	0.1657	&	0.0205	&	0.1813	&	0.1832	&	0.3487	&	0.2059	&	0.3714	\\	\hline \hline
			\multicolumn{9}{l}{RMSE is Root Mean Square Error and Mean is the mean value of $ m=2000 $ APEs.}
	\end{tabular} }
	\label{table:mc2}
\end{table}

The results therefore imply that assuming $(\hat{\pmb{\epsilon}}_{it}, \hat{\pmb{\alpha}}_{i}) $ as control function for identifying the ASF and APE, may not be restrictive, and that the developed method can yield consistent result.

%This could be because the assumptions of our model that $ \theta_{i} \indep  Z_{i} | \alpha_{i} $ and that $ (Z^{\prime}_{i}, \alpha_{i}, \theta_{i})^{\prime} \indep (\zeta_{it}, \epsilon_{it})$ -- according to which we generate the data and which implies that $ \theta_{i} +\zeta_{it} \indep  x_{it} | \alpha_{i}, \epsilon_{it}$ -- does not imply that their control function assumption, $ \tau_{i}+ \zeta_{it} \indep x_{it} | a_{i} + \epsilon_{it} $, holds true.    

To conclude, this finite sample study establishes the following: 

(1) Our method performs well with sample sizes frequently encountered in practice.

(2) It performs better than the alternative estimators with set-ups similar to ours. 

(3) Employing $ V_{PW}$, instead of $ \pmb{\upsilon}_{PWt} $, as a control function yielded consistent APEs when the instrument had a large support. This suggest that when $ Z $ is independent of the error terms and the instruments, $ \pmb{z}_{t} $, have a large support, then employing $V$ as a control function could yield consistent APEs.

\section{Implications of Ownership of Land and Farm Assets on Child Labor }
\subsection{Introduction}
Child labor is a pressing concern in all developing countries. According to International Labour Office's 2016 estimates, worldwide, 152 million children in the age group of 5 to 17 years age group are victims of child labor; 62 million of which are in the
Asia-Pacific region. Conditions of child labor can vary. Many children work in hazardous industries that take a toll on their health. Moreover, when children work, they forego education and human capital accumulation, with deleterious effect on their future earning potential. Furthermore, since there is positive externality to human capital accumulation, as argued by \citet{baland:2000} (BR), the social return to such accumulation, too, is not realized.

There is a huge literature, both empirical and theoretical, that has sought to understand the mechanism underlying child labor. What has emerged is that poverty (\citealt{basu:1998} \& \citetalias{baland:2000}), along with imperfection in labor and land market \citep[][]{bhalotra:2003,dumas:2007,basu:2010} and capital market (\citetalias{baland:2000}) to be the major causes of child labor. \citetalias{baland:2000} show that child labor increases when endowments of parents are low, and that when capital market
imperfections exist and parents cannot borrow, child labor becomes inefficiently high.

\citet{basu:2010} (BDD) point out that papers like \citet{bhalotra:2003}(BHy) and \citet{dumas:2007} show that in some developing countries the amount of work the children of a household do increases with the amount of land possessed by the household. Since land is usually strongly correlated with a household's income, this finding seems to challenge the presumption that child labor involves the poorest households. They argue that these perverse findings are a facet of labor and land market imperfections, and that in developing countries, poor households in order to escape poverty want to send their children to work but are unable to do so because they have no access to labor markets close to their home. In such a situation, if the household comes to acquire some wealth, say land, its children, if only to escape penury, will start working. However, if the household's land ownership continues to rise, then beyond a point the household will be well-off enough and it will not want to make its children work.

\citetalias{bhalotra:2003} argue that on one hand there is the negative wealth effect of large landholding on child labor, whereby large landholding generate higher income and, thereby, makes it easier for the household to forgo the income that child labor would bring. On the other there is the substitution effect, where due to labor market imperfections, owners of land who are unable to productively hire labor on their farms have an incentive to employ their children. Since the marginal product of child labor is increasing in farm size, this incentive is stronger amongst larger landowners. The value of work experience will also tend to increase in farm size  if the child stands to inherit the family farm. Furthermore, they argue that large landowners who cannot productively hire labor would want to sell their land rather than employ their children on it, but, because of land market failure, are unable to do so. Thus, land market failure reinforces labor market failure.

\citet{Cockburn:2007} (CD) in their analysis of child labor in Ethiopia find that in presence of labor market imperfections, all assets need not be child labor enhancing. They find that certain productive assets that enable an increase in the total family income may not necessarily increase child labor. They show that assets such as oxen and ploughs that are operated by adults decrease child labor. To test this hypothesis, in our empirical specification we include an index of productive farm assets.

Now, while land and labor market imperfections may exist in developing countries, the extent of imperfection may not be uniform across all countries, or regions within a country. Hence, the relationship between child labor and different kinds of assets, such as landholding
or agrarian assets, is an empirical question. The question is important because policy implications could be different under different relationships between various kinds of assets and child labor. For example, if one were to confirm the findings in \citetalias{bhalotra:2003} and \citetalias{basu:2010}, then if monetary transfers are used to increase landholding or land redistribution is done in favor of the poor, child labor may in fact increase. On the other hand, when monetary transfers are used to increase agrarian assets, then is an inverse relationship between agrarian assets and child labor holds, such transfers could reduce the incidence of child labor.

In our data, we find non-agricultural income to be much higher than agricultural income (Table \ref{table:app2}). This suggests that land is not the only source of income as in \citetalias{bhalotra:2003} and \citetalias{basu:2010}. \citetalias{basu:2010} assume that land is the only source of income and derive a regression equation where household income is left out. Since non-agricultural income constitutes a major portion of total household income, we also control for household income.

We also find that, overtime, land size distribution has become more unequal (Table \ref{table:app2}), which indicates that land market exists in the regions from where the data has been collected. Now, if land market exists, even if imperfect, then it is unlikely that land owned by households will be exogenous to a household's labor supply as in \citetalias{bhalotra:2003} and \citetalias{basu:2010}, where land is mainly inherited, but endogenously determined
along with household's, including children's, labor supply decisions. However, endogeneity could also arise due to omitted variables. To account for the endogeneity of landholding along with that of productive assets and household income, we employ the method developed in the paper.

\subsection{Data and Empirical Model}
\subsubsection{Data}
We conduct our empirical analysis at the level of the child using two waves, 2006-07 and 2009-2010, of the data from Young Lives Study (YLS), a panel study from six districts of the state of United\footnote{In 2014 the north-western portion of the then Andhra Pradesh was separated to form the new state of Telangana.} Andhra Pradesh (henceforth AP) in India. We restrict our sample to children
in the age group of 5 to 14 years in 2007 living in rural areas. Finally, excluding children for whom relevant information was missing, we were left with 2458 children. Table \ref{table:app1} and Table \ref{table:app2} have the relevant descriptive statistics.
\begin{table}[h!]
	\caption{\textbf{Work Status by Age Group}}
	\centering
	\resizebox{\columnwidth}{!}{
		\begin{tabular}{l|c|c|c|l|c|c|c}
			\hline\hline
			\multicolumn{4}{c|}{Year 2007} & \multicolumn{4}{c}{ Year: 2010}\\\hline
			Age Group & Not Working  & Working  & Total  & Age Group  & Not Working  & Working  & Total \\\hline
			5 to 7 years  & 45.25  & 5.02  & 50.27  & 8 to 10 years  &  31.25  & 19.03  & 50.27 \\
			8 to 14 years  & 22.88  & 26.85  & 49.73  & 11 to 17 years  & 14.98  & 34.75 & 49.73\\\hline
			Total  & 68.13  & 31.87  & 100.00  & Total  & 46.23  & 53.77  & 100.00 \\\hline\hline
			\multicolumn{8}{l}{The figures are in percentage. Total number of children in each period: 2458}
	\end{tabular}}
	\label{table:app1}
\end{table}

Children were asked how much time they spent in the reference period (a typical day in the last week) doing (a) wage labor, (b) non-wage labor, or (c) domestic chores\footnote{Wage labor involves activities for pay, work done for money outside of household, or work done for
	someone not a part the household. Non-wage labor includes tasks on family farm, cattle herding (household and/or community), other family business, shepherding, piecework or handicrafts done at home (not just farming), and domestic work includes tasks and chores such as fetching water, firewood, cleaning, cooking, washing, and shopping.}. If the answer was positive number of hours for any one of the activities, then the binary variable $ DWORK $ was assigned value 1, 0 otherwise. The major component of work (not reported here) is due to domestic chores. While both domestic and non-domestic work registered increase over the years, the increase in the proportion of children doing non-domestic work was higher.

\begin{table}[h!]
	\caption{\textbf{Descriptive Statistics}}
	\centering
	\resizebox{\columnwidth}{!}{
		\begin{tabular}{l|c c|c c}\hline\hline
			& \multicolumn{2}{c|}{2007} & \multicolumn{2}{c}{2010} \\ \hline
			Variable	&	Mean &	Std. Dev.	& 	Mean  &	Std. Dev.\\ \hline
			\multicolumn{5}{c}{\textbf{Child characteristics}} \\\hline
			Sex (Male=1, Female=0) 	&	0.52 &	0.50	&		0.52 &	0.50 \\
			Age (yrs.)	&	8.07 &	2.97	&		11.07 &	2.97 \\\hline
			\multicolumn{5}{c}{\textbf{Household characteristics}} \\\hline
			Parents participated in NREGS	(Yes=1 \& No=0) &	0.33	&0.47	&		0.66	&0.47 \\
			Total number of days parents worked in NREGS	&	  9.21  &  21.44   & 36.00 & 48.10\\\hline
			Land Owned (acre) &  2.32   &  3.42    & 3.86  & 43.53\\
			Farm Asset Index &  -0.13   &  0.98    & 0.22  & 1.46\\\hline
			Gini Coefficient for Land Owned & \multicolumn{2}{c|}{ 0.62}  & \multicolumn{2}{c}{0.74}   \\
			\hline
			Total Income of Household (in Thousand \rupee) &  30.91  &  34.35   & 48.88 & 60.24\\
			Annual non-agricultural income (\rupee)& 20787 & 35813  &  29013  &  62225 \\
			Annual agricultural income (\rupee) & 5060  &  23319  & 9936  &  42746  \\\hline
			Does a household own farm assets (Yes=1 \& No=0)&  0.69   & 0.46 &  0.91   &  0.29 \\
			Number of farm assets              &4.70  &  11.06 & 6.29  &  9.01  \\
			\hline
			Engineered Road to the Locality (Yes=1 \& No=0) & 0.32 & 0.47 & 0.58 & 0.49\\
			Drinkable Water in the Locality (Yes=1 \& No=0) & 0.87 & 0.34 & 0.86 & 0.34\\
			National Bank in the Locality (Yes=1 \& No=0) & 0.23 & 0.41 & 0.08 & 0.27\\
			Hospital in the Locality (Yes=1 \& No=0) & 0.37 & 0.89 & 0.38 & 0.48 \\\hline
			\multicolumn{5}{c}{\textbf{Community (Mandal) characteristics}}\\\hline
			Total NREGS amount sanctioned (in Million \rupee)& 7.25 & 8.30 & 20.19 & 19.17\\\hline
			
			\hline
			\multicolumn{5}{l}{Total number of children/observations in each period: 2458}
	\end{tabular} }
	\label{table:app2}
\end{table}

In Table \ref{table:app2} we find that while land ownership has become more unequal, the  average size of land owned increased over the years. Farming Asset Index, which too increased over the years, was constructed by Principal Component Analysis of several variables, each of which indicate the number of farming related assets\footnote{Farming assets constitute of agriculture tools, carts, pesticide pumps, ploughs, water pumps, threshers, tractors, and other farm equipments.} of each kind that the household owns.

\subsubsection{Empirical Model}

Let $ y_{it} $ be the binary variable that takes value 1 if the parents of the child $ i $ decide that the child works and 0 otherwise. The decision is modelled as in equation (\ref{eq:1}), where $  y^{\ast}_{it} $ is amount of time devoted to work by child $i$ in period $t$. The set of endogenous variables, $\pmb{x}_{it}$, include income ($INCOME_{it}$) of the household to which the child $ i $ belongs, size of the land holdings ($LAND_{it}$), and the index of productive farm assets ($ASSET_{it}$).

To address the issues of endogeneity and heterogeneity, we employ the two-step control function methodology developed in the paper, where the control functions,  \begin{align}\nonumber
\hat{\pmb{\alpha}}_{i}^{\prime} =
(\hat{\alpha}_{INCOME, i}, \hat{\alpha}_{LAND, i} , \hat{\alpha}_{ASSET, i})
\textrm{ and  }  \hat{\pmb{\epsilon}}_{it}^{\prime}=
(\hat{\epsilon}_{INCOME, it} , \hat{\epsilon}_{LAND, it}, \hat{\epsilon}_{ASSET, it}),
\end{align}
are obtained from the estimates of the first stage reduced form equations (\ref{eq:3}). After augmenting the structural equation (\ref{eq:1}) with the control functions, we get the modified structural equation (\ref{eq:7}), which we estimate as a probit model\footnote{
	For many children, as we know, the optimal choice of $ y^{\ast}_{it} $ is the corner solution, $ y^{\ast}_{it} = 0 $. For corner solution outcomes, we are interested in features of the distribution such as $ \int \Pr(y^{\ast}_{it} > 0 | \mathcal{X}_{it}, \hat{\pmb{\alpha}}_{i}, \hat{\pmb{\epsilon}}_{it})dF(\hat{\pmb{\alpha}},\hat{\pmb{\epsilon}}) $ and $ \int \E(y^{\ast}_{it} | \mathcal{X}_{it}, \hat{\pmb{\alpha}}_{i}, \hat{\pmb{\epsilon}}_{it} )dF(\hat{\pmb{\alpha}},\hat{\pmb{\epsilon}}) $, where $ \mathcal{X}_{it} = (\pmb{x}_{it}^{\prime}, \pmb{w}^{\prime}_{it})^{\prime} $ and		
	\begin{align}\nonumber
	\E(y^{\ast}_{it} | \mathcal{X}_{it}, \hat{\pmb{\alpha}}_{i}, \hat{\pmb{\epsilon}}_{it} ) &=  \Pr(y^{\ast}_{it} =0 | \mathcal{X}_{it}, \hat{\pmb{\alpha}}_{i}, \hat{\pmb{\epsilon}}_{it}). 0 +  \Pr(y^{\ast}_{it} > 0 | \mathcal{X}_{it}, \hat{\pmb{\alpha}}_{i}, \hat{\pmb{\epsilon}}_{it}). \E(y^{\ast}_{it} |  \mathcal{X}_{it}, \hat{\pmb{\alpha}}_{i}, \hat{\pmb{\epsilon}}_{it}, y^{\ast}_{it} > 0)\\\nonumber
	&= \Pr(y^{\ast}_{it} > 0 | \mathcal{X}_{it}, \hat{\pmb{\alpha}}_{i}, \hat{\pmb{\epsilon}}_{it}) E(y^{\ast}_{it} |  \mathcal{X}_{it}, \hat{\pmb{\alpha}}_{i}, \hat{\pmb{\epsilon}}_{it}, y^{\ast}_{it} > 0). 
	\end{align}
	Due to lack of space, in this application we study only $ \int \Pr(y^{\ast}_{it} > 0 | \mathcal{X}_{it}, \hat{\pmb{\alpha}}_{i}, \hat{\pmb{\epsilon}}_{it})dF(\hat{\pmb{\alpha}},\hat{\pmb{\epsilon}}) $.}.

To identify the impact of the endogenous variables on the parents' decision to make their children work, we employ the following instruments: (1) $ NREGS $, which is the total  sanctioned amount at the mandal (region) level at the beginning of financial year (in 2008-09 prices) to support an employment guarantee scheme; (2) $ CASTE $, caste (social group) of the child; and (3) a set of four indicator variables that capture the level of infrastructural development in the household's locality/settlement.

%, which \cite{afridi:2016} employ to instrument income in their paper

The National Rural Employment Guarantee Scheme (NREGS) was initiated in 2006 by the Government of India with the objective to alleviate rural poverty. NREGS legally entitles rural households to 100 days of employment in unskilled manual labour (on public work projects) at a prefixed wage. Now, it can be seen in Table \ref{table:app2} that over the years, the proportion of children with either parent working in NREGS almost doubled. This increase in participation was accompanied by a rise in the number of  days of work on NREGS projects as well. \citet{afridi:2016} claiming $NREGS$ to be a valid instrument for income, argue that since fund sanctioned at the beginning of the financial year is not be affected by current demand for work, the funds sanctioned is exogenous and more funds imply more work opportunity in NREGS, which can have a positive effect on household income. Also, the total fund allocation to NREGS increased during the period 2007-2010. However, this increase was not uniform across the 15 \textit{mandal}s\footnote{Data on the sanctioned funds at the \textit{mandal} level was obtained from the Andhra Pradesh Government's website on NREGS (http://nrega.ap.gov.in/).}.

Our second instrument is the caste, a system of social stratification, to which the child belongs. India is beleaguered with a caste system. Within this caste system, historically, the Scheduled Castes and Scheduled Tribes (SC/ST's) have been economically backward and
concentrated in low-skill (mostly agricultural) occupations in rural areas. Moreover, they were also subject to centuries of systematic caste based discrimination, both economically and socially. The historical tradition of social division through the caste system created a
social stratification along education, occupation, income, and wealth lines that has continued into modern India\footnote{In fact, this stratification was so endemic that the constitution of India aggregated these castes into a schedule of the constitution and provided them with affirmative action cover in both education and public sector employment. This constitutional initiative was viewed as a key component of attaining the goal of raising the social and economic status of the SC/STs to the levels of the non-SC/ST's.}. Fairing better than SC/ST's are those belonging to the ``Other Backward Classes " (OBC)\footnote{The Government of India classifies,  a classification based on social and economic conditions, some of its citizen as Other Backward Classes  (OBC). The OBC list is dynamic (castes and communities can be added or removed) and is supposed to change from time to time depending on social, educational and economic factors.  In the constitution, OBC's are described as ``socially and educationally backward classes", and government is enjoined to ensure their social and educational development.}. Hence, given the fact that income and wealth, both land and productive assets, vary with caste, we choose $ CASTE $ as our second instrument, which is a discrete variable that takes three values: 1 if the child belongs to SC/ST household, 2 if the child belongs to OBC, and 3 if the child belongs to group labelled as ``Others" (OT). 

%The variable $ CASTE $, thus defined, is likely to be a good predictor of household income and wealth, where the average SC/ST household is likely to be poor, followed by the OBC's, and those in the OT group being the wealthiest.

\begin{table}[h!]
	\caption{\textbf{Descriptive Statistics of some Variables by Caste}}
	\centering
	\resizebox{\columnwidth}{!}{
		\begin{tabular}{l|c |c|c |c}\hline\hline
			&&Scheduled Castes/Tribes & Other Backward Classes  & Others  \\\hline
			Year: 2007 & Household Income			&	31.22	&	31.64	&	43.21	\\
			&	(in Thousand \rupee).	&	(33.94)	&	(34.29)	&	(48.59)	\\\cline{2-5}
			&Land	Owned 		&	1.58	&	2.32	&	3.08	\\
			&	in acre	&	(2.12)	&	(3.51)	&	(4.53)	\\ \cline{2-5}
			&Index of Productive 			&	-0.22	&	-0.14	&	0.04	\\
			& Farm Asset	&	(0.71)	&	(1.02)	&	(1.17)	\\ \cline{2-5}
			&School Dummy			&	0.90	&	0.89	&	0.96	\\
			& $DSCHOOL=1$		&	(0.29)	&	(0.32)	&	(0.19)	\\ \cline{2-5}
			&Work Dummy			&	0.33	&	0.33	&	0.29	\\
			& $DWORK=1$		&	(0.47)	&	(0.47)	&	(0.45)	\\\hline 
			Year: 2010 & Household Income			&	45.99	&	50.22	&	64.76	\\
			&	in Thousand Rs.	&	(45.51)	&	(66.35)	&	(70.26)	\\ \cline{2-5}
			&Land Owned 			&	2.10	&	2.79	&	10.90	\\
			&	in acre &	(1.95)	&	(15.82)	&	(108.71)	\\ \cline{2-5}
			&Index of productive 			&	0.12	&	0.29	&	0.54	\\
			&	Farm Asset	&	(1.16)	&	(1.56)	&	(1.89)	\\ \cline{2-5}
			&School Dummy			&	0.89	&	0.87	&	0.94	\\
			&	$DCHOOL=1$	&	(0.31)	&	(0.33)	&	(0.23)	\\ \cline{2-5}
			&Work Dummy			&	0.52	&	0.57	&	0.48	\\
			&	$DWORK=1$	&	(0.50)	&	(0.49)	&	(0.50)	\\ \hline
			\multicolumn{2}{l|}{Number of Children/observations}		&	&		&	 \\
			\multicolumn{2}{l|}{in each period:}		&	906	&	1269	& 283 \\\hline \hline
			\multicolumn{4}{l}{\footnotesize Standard errors in parentheses.}\\
	\end{tabular} }
	\label{table:app4}
\end{table}

We claim that $ CASTE $ is a valid instrument for landholding because, though average wealth and income are evidently distributed along caste lines, we do not find a significant variation in child labor or school enrolment across caste or social group to which the child belongs (Table \ref{table:app4}). In other words, no social group is inherently disposed to make their children work or send them to school. This could be because rising awareness, overtime, about returns from education persuades families of all castes to send their children to
school. We find support for the assertion in the literature too. \citet{hnatkovska:2012} find significant convergence in the education attainment levels and occupation choice of SC/ST's and non-SC/ST's between 1983 and 2004-2005; moreover, the convergence in education level has been highest for the youngest cohort.  Secondly, time-invariant ethnicity variable such as caste cannot be correlated with unobserved time-invariant heterogeneity such as parents' or children's abilities and land quality.

%\citet{hnatkovska:2012} find significant convergence in the education attainment levels, occupation choice, wages and consumption of SC/ST's and non-SC/ST's between 1983 and 2004-2005. Moreover,the convergence in education level has been highest for the youngest cohort, and that the overall consumption and wage convergence between the groups has been driven significantly by convergence in the educational choices of the two groups.

Our assertion that the preferences of parents regarding child labor and schooling does not differ systemically across social groups is supported by the data. In the first wave of the data, the following question was asked: ``Imagine that a family in the village has a 12 year old son/daughter who is attending school full-time. The family badly needs to increase the household income. One option is to send the son/daughter to work but the son/daughter wants to stay in school. What should the family do?" There was little difference in the response across caste groups -- 90\% of SC/ST's, 87\% of OBC's, and 93\% of OT's wanted that sons of such distressed families be kept at school. For daughters, the
corresponding figures are: 87\% of SC/ST's, 87\% of OBC's, and 91\% of OT's. Also, 96\% of SC/ST households expected their children to complete a minimum of high school. The corresponding figure for OBC's and OT's are 95\% and 98\% respectively.

Our third set of instruments is a set of four dummy variables, which indicate (1) if drinkable water is provided in the locality, (2) if the services of a national bank are provided in the locality, (3) if private hospitals exist in the locality, and (4) if access to the locality is via an engineered road. As in \citetalias{bhalotra:2003}, these variables, which indicate the level
of infrastructure development, are employed to instrument the index of productive farm assets.

\subsection{Discussion of Results}

The results of the first stage reduced form equations in Table \ref{table:app5} suggest that our instruments are good predictors of the endogenous variables. First, we find that an increase in the amount sanctioned for  NREGS  projects  increases the household income. Secondly, $ CASTE $ does, on an average, correctly predict the economic (income, land holding, and assets) status of households. Finally, the dummy variables indicating the level of infrastructure development
are positively correlated with the index of productive farm assets.
\begin{table}[h!]
	\caption{\textbf{First Stage Reduced Form Estimates: Joint Estimation of Income, Land, and Farm Assets Equations}}
	\centering
	\resizebox{\columnwidth}{!}{
		\begin{tabular}{l| r@{}l|r@{}l |r@{}l} \hline\hline
			&	\multicolumn{2}{|c|}{Income}			&	\multicolumn{2}{|l}{Landholding}		&	\multicolumn{2}{|l}{Farm Asset}		\\	\hline
			Total NREGS amount sanctioned (in Million \rupee )	&	0.047	&	\onepc	&	-0.008	&		&	-0.0003	&		\\	
			&	(0.009)	&		&	(0.007)	&		&	(0.0002)	&		\\	\hline
			Caste (SC/ST = 1, OBC = 2, OT = 3)	&	9.220	&	\onepc	&	2.278	&	\onepc	&	0.171	&	\onepc	\\	
			&	(1.217)	&		&	(0.726)	&		&	(0.0300)	&		\\	\hline
			Drinkable Water in the Locality (Yes=1 \& No=0)	&	5.417	&		&	-1.879	&		&	0.341	&	\fivepc	\\	
			&	(5.703)	&		&	(4.260)	&		&	(0.150)	&		\\	\hline
			National Bank in the Locality (Yes=1 \& No=0)	&	-2.785	&		&	4.684	&	\fivepc	&	0.046	&		\\	
			&	(3.099)	&		&	(2.315)	&		&	(0.082)	&		\\	\hline
			Engineered Road to the Locality (Yes=1 \& No=0)	&	0.159	&		&	2.413	&		&	0.182	&	\onepc	\\	
			&	(2.130)	&		&	(1.591)	&		&	(0.0561)	&		\\	\hline
			Hospital in the Locality (Yes=1 \& No=0)	&	-0.689	&		&	-4.143	&	\onepc	&	0.056	&	\tenpc	\\	
			&	(1.248)	&		&	(0.932)	&		&	(0.033)	&		\\	\hline
			Other Exogenous Variables of the Structural  	&		\multicolumn{2}{|l|}{Yes}		&		\multicolumn{2}{|l}{Yes} &		\multicolumn{2}{|l}{Yes}	\\
			Equation: Age and Sex of the Child	&		\multicolumn{2}{|c|}{}		&		\multicolumn{2}{|c|}{}		\\\hline		\hline
			\multicolumn{7}{p{1\textwidth}}{\citeauthor{biorn:2004}'s stepwise MLE was employed to obtain these estimates.}\\\hline
			\legend & \multicolumn{4}{l}{\footnotesize Standard errors (SE) in parentheses }  
	\end{tabular} }
	\label{table:app5}
\end{table}

Before we begin to discuss the result of the second-stage estimation in Table \ref{table:app6}, we state a few points regarding the estimation. (a) The only exogenous explanatory variables in our parsimonious\footnote{Though we do not report here, we did not find that nonlinear terms of income, land, and productive
	assets to be significant. We had also included four education related dummy variables, two for the father and two for the mother. The dummy variables for the mother, for example, indicated (1) if the mother had some schooling and (2) if the mother had attended secondary or post secondary school. The education dummies, though substantially affecting household income, did not seem to affect child labor propensity. This suggests that parents' education level has had no independent impact except through income.} specification are the age and the sex of the children.  
(b) The specification includes district dummies, a time dummy, and the interaction of the two to account for the fact that the districts to
which children belong may have different economic growth trajectories as well as trends related to work and education. The time dummy allows us to control for changes in demand and supply of work over time. (c) Since the support assumption for point identification of the APEs is not met, we estimate the bounds on the APEs and the 95\% confidence interval ($ \text{CI}_{95\%} $) for the partially identified APEs.  (d) For the continuous variables, the bounds on the APE of a variable were computed by increasing the variable by one standard deviation from its mean, where the mean and the SD of the variable are from the 2010 data. For age, the bounds on APE were computed by increasing the mean age in 2010 by 1 year. (e) The standard errors of the coefficients were estimated using the analytical expression of the covariance matrix derived in Appendix D.     

\begin{table}[h!]
	\caption{\textbf{Household Income and Wealth Effect on Incidence of Child Labor}}
	\centering
	\resizebox{\columnwidth}{!}{		
		\begin{tabular}{l| r@{}l|r@{}l|r@{}l| c c |r@{}l } \hline\hline
			&\multicolumn{2}{c|}{\textbf{CRE Probit}} &\multicolumn{8}{c}{\textbf{Control Function (CF) Method}}\\\hline 
			&\multicolumn{2}{c|}{}&\multicolumn{2}{c|}{}& \multicolumn{2}{c}{APE Bounds  }& \multicolumn{2}{|c|}{}&\multicolumn{2}{c}{}\\
			&\multicolumn{2}{c|}{Coeff.}&\multicolumn{2}{c|}{Coeff.}& \multicolumn{2}{c}{ $ \text{CI}_{95\%} $}& \multicolumn{2}{|c|}{CFs}&\multicolumn{2}{c}{Coeff.}\\\hline
			Income	&	0.003  	&	\onepc	&	-0.0234	&	\onepc	&	\multicolumn{2}{c|}{ [-0.00532,   -0.00451] } & \multicolumn{2}{c|}{$\hat{\alpha}_{INCOME}$} &	0.005	& \\	
			&	(0.0008)	&	    	&	(0.0028)	&		&	\multicolumn{2}{c|}{ [-0.00533,    -0.0045] }	& & & (0.003) &  	\\	\hline
			Landholding	&	0.002	&		&	0.031	&	\onepc	&	\multicolumn{2}{c|}{ [0.00679,    0.00749 ] } 	&	\multicolumn{2}{c|}{$\hat{\alpha}_{LAND}$}			&	-0.015	&	\fivepc	\\	
			&	(0.002)	&		&	(0.007)	&		&	\multicolumn{2}{c|}{ [0.00676,    0.00752] }	& & & (0.0065) & 	\\	\hline
			Farm Asset Index	&	-0.011	&		&	-0.976	&	\onepc	&	\multicolumn{2}{c|}{[-0.22077,   -0.18734] } &	\multicolumn{2}{c|}{$\hat{\alpha}_{ASSET}$}			&	1.512	&	\onepc	\\	
			&	(0.0279)	&		&	(0.169)	&		&	\multicolumn{2}{c|}{ [-0.22121,    -0.1869] }	& & & (0.129) & 	\\	\hline
			Age	&	2.019	&	\onepc	&	0.402	&	\onepc	&	\multicolumn{2}{c|}{ [0.0757,    0.12533] } &	\multicolumn{2}{c|}{$\hat{\epsilon}_{INCOME}$}			&	0.0275	&	\onepc	\\	
			&	(0.072)	&		&	(0.057)	&		&	\multicolumn{2}{c|}{ [0.0755,    0.12552] }	& & & (0.003) & 	\\	\hline
			Sex	&	0.644	&	\onepc	&	0.394	&	\onepc	&	\multicolumn{2}{c|}{ [0.07355,    0.12318] }&	\multicolumn{2}{c|}{$\hat{\epsilon}_{LAND}$}			&	-0.031	&	\onepc	\\	
			&	(0.042)	&		&	(0.0473)	&		&	\multicolumn{2}{c|}{ [0.07322,    0.12351] }	& & & (0.0075) & 		\\	\hline
			\multicolumn{3}{c|}{ }					& 	\multicolumn{4}{c|}{ }					&	\multicolumn{2}{c|}{$\hat{\epsilon}_{ASSET}$}			&	0.882	&	\onepc	\\	
			\multicolumn{3}{c|}{ }					& 	\multicolumn{4}{c|}{ }					&	\multicolumn{2}{c|}{}			&	(0.185)	&		\\	\hline \hline
	\end{tabular}}
	\centering
	\resizebox{\columnwidth}{!}{
		\begin{tabular}{ c c c c c c c c  }
			\multicolumn{7}{l}{Total number of children: 2458. Total number of observations with positive outcome: 2128 }\\\hline
			\legend, & \multicolumn{3}{l}{\footnotesize Standard errors (SE) in parentheses} 
	\end{tabular} }
	
	\label{table:app6}
	%}
\end{table}

We begin by comparing the results from \citeauthor{chamberlain:1984}'s CRE probit model with the estimates obtained from applying the method developed in this paper. The significance of estimated coefficients of the control functions suggests that income, land size, and productive farm assets are endogenously determined along with household's labor supply, including that of the child's, decisions.  When income and  wealth are not instrumented, as in the CRE probit, considering the discussion in the paper, we get an incorrect sign for the coefficient on income. Moreover, the result of CRE probit suggests that ownership of land and farm assets do not affect child labor, which, given the many recent evidences, is unlikely in a developing country. The results, thus, make clear the importance of accounting for endogeneity of income, landholding, and farm asset.

The estimates from the control function method suggest that children of households that have a higher landholding are more likely to engage in work. This is in conformity with the findings in \citetalias{basu:2010},  \citetalias{bhalotra:2003} and \citetalias{Cockburn:2007}, where, due to presence of land, labor, and credit market imperfections, ownership of large amount land provides incentives for children to work. As far as income is concerned, we find that higher household income reduces the chances of child labor, which again confirms poverty to be a cause of child labor.

Since the upper and lower bounds of the APE of productive farm assets are high and since the $ \text{CI}_{95\%} $ is only marginally bigger than the bound, it seems that ownership of farm assets leads to a significantly high reduction in children's participation in work. \citeauthor{dumas:2007}, \citetalias{bhalotra:2003} and \citetalias{Cockburn:2007} argue that an increase in asset holding that increases the marginal productivity of labor induces two opposite effects on labor. While the income effect of increased wealth tends to reduce the labor time, the substitution effect, due to the absence of labor market, provides incentives for work, and tends to increase children's labor time. Our results suggest that the wealth effect of farm assets, which are not likely to be operated by children, dominate to reduce children's labor time. Secondly, since the prevalence of farm assets is high in those regions where there has been infrastructure development, it seems that lack of infrastructure development that impedes access to, or does not provide incentives to acquire, productive farm assets may be an important factor determining child labor\footnote{In a separate set of regressions that included only the exogenous variables, we tried to assess if the infrastructure variables had independent impacts on work and schooling decisions of children. These variables turned out to be insignificant, suggesting that the demand for child labor or opportunities for schooling were not affected by infrastructure development or its lack in rural AP. In other words, infrastructure had
	its effect on work and schooling outcomes only through its impact on the economic conditions of certain households, which validates using infrastructure variables as instruments for farming assets.}. Finally, we find that older children and boys are more likely to work.

\section{Concluding Remarks}
The objective of the paper has been to develop a method to estimate structural measures of interest such as the average partial effects for panel data binary response model in a triangular system while accounting for multiple unobserved heterogeneities. The unobserved heterogeneity terms constitute of time invariant random effects/coefficients and idiosyncratic errors. We propose that the expected values -- conditional on the histories of the endogenous variables, $X_{i}\equiv(\pmb{x}^{\prime}_{i1},\ldots, \pmb{x}^{\prime}_{iT})^{\prime}$, and the exogenous variables, $Z_{i}\equiv(\pmb{z}^{\prime}_{i1},\ldots, \pmb{z}^{\prime}_{iT})^{\prime}$ -- of the heterogeneity terms be used as control functions (CF).

The proposed method makes a number of interesting contribution to the literature. First, among the class of triangular system with imposed structures similar to ours, the proposed CF method requires weaker restrictions than the traditional control function methods. Secondly, when instruments have a small support, the CFs, which exploit panel data, help in point-identifying structural measures such as the APEs when the endogenous variables have a large support. Bounds on the structural measures are provided when the support assumption is not satisfied. Thirdly, the method allows for multiple endogenous variables, all of which are determined simultaneously. Finally, in an equivalence result we showed that for linear panel data models, when the structural equation is augmented with the proposed control functions, the resulting estimates are equivalent to the ones that are obtained when the structural model is estimated by a certain two-stage least squares. Also, Monte Carlo experiments show that compared to alternative panel data binary choice models similar to ours, our method performs better.  

% The proposed method makes a number of interesting contribution to the literature. Apart from achieving identification of measures such as the average partial effects in a triangular system with multi-dimensional heterogeneity, among the class of triangular system with imposed structures similar to ours, the proposed control function method requires weaker restrictions than the traditional control function methods. Finally, the method allows for instruments with small support, which was possible due to panel data and time invariance of certain heterogeneity terms. Also, Monte Carlo experiments show that compared to alternative panel data binary choice models similar to ours, our method performs better.

The estimator was applied to estimate the causal effects of income, land size, and farm assets on the incidence of child labor. We found that household income and ownership of farming assets significantly lower the incidence of child labor, suggesting a strong income effect of farm assets. Secondly, large landholding increases the likelihood of child labor, suggesting a substitution effect of land ownership. Thirdly, a test of exogeneity revealed that land size is determined endogenously along with household labor supply decisions, contrary to what most empirical studies on child labor in developing countries assume.

Finally, we would like to note that (i) extension of the methodology for estimating dynamic binary choice models and (ii) identification and estimation the proposed control functions without making distributional assumptions about the heterogeneity terms of the reduced form equations would be important contributions to the literature.

%\vspace{.25cm}\noindent\textbf{Acknowledgement}\vspace{.25cm}I would like to thank anonymous referees for helpful comments. Thanks are due to seminar participants at the The Bank of Estonia, the $ 10^{th} $ Nordic Econometric Meeting (Stockholm), the Institute of Mathematics and Statistics (University of Tartu), and the Inaugural Baltic Economic Conference (Vilnius) for the same. I would especially like to thank Soham Sahoo for helping me with the data. All remaining errors are mine.

\bibliography{Citation}{}

\begin{thebibliography}{58}
\providecommand{\natexlab}[1]{#1}

\bibitem[{Afridi \textit{et~al.}(2016)Afridi, Mukhopadhyay and
  Sahoo}]{afridi:2016}
\textsc{Afridi, F.}, \textsc{Mukhopadhyay, A.} and \textsc{Sahoo, S.} (2016).
  {Female Labor Force Participation and Child Education in India: Evidence from
  the National Rural Employment Guarantee Scheme}. \textit{IZA Journal of Labor
  \& Development}, \textbf{5:7}, doi:10.1186/s40175--016--0053--y.

\bibitem[{Altonji and Matzkin(2005)}]{altonji:2005}
\textsc{Altonji, J.~G.} and \textsc{Matzkin, R.~L.} (2005). {Cross Section and
  Panel Data Estimators for Nonseparable Models with Endogenous Regressors}.
  \textit{Econometrica}, \textbf{73}, 1053--1102.

\bibitem[{Arellano and Bonhomme(2011)}]{arellano:2011}
\textsc{Arellano, M.} and \textsc{Bonhomme, S.} (2011). {Nonlinear Panel Data
  Analysis}. \textit{Annual Review of Economics}, \textbf{3}, 395--424.

\bibitem[{Arellano and Carrasco(2003)}]{arellano:2003}
\textsc{---} and \textsc{Carrasco, R.} (2003). {Binary Choice Panel Data Models
  With Predetermined Variables}. \textit{Journal of Econometrics},
  \textbf{115}~(1), 125--157.

\bibitem[{Baland and Robinson(2000)}]{baland:2000}
\textsc{Baland, J.~M.} and \textsc{Robinson, J.~A.} (2000). {Is Child Labor
  Inefficient?} \textit{Journal of Political Economy}, \textbf{108}, 663--679.

\bibitem[{Baltagi(1981)}]{baltagi:1981}
\textsc{Baltagi, B.} (1981). Simultaneous equations with error components.
  \textit{Journal of Econometrics}, \textbf{17}~(2), 189--200.

\bibitem[{Baltagi \textit{et~al.}(2010)Baltagi, Song and Jung}]{baltagi:2010}
\textsc{Baltagi, B.~H.}, \textsc{Song, S.~H.} and \textsc{Jung, B.~C.} (2010).
  {Testing for Heteroskedasticity and Serial Correlation in a Random Effects
  Panel Data Model}. \textit{Journal of Econometrics}, \textbf{154}, 122--124.

\bibitem[{Basu \textit{et~al.}(2010)Basu, Das and Dutta}]{basu:2010}
\textsc{Basu, K.}, \textsc{Das, S.} and \textsc{Dutta, B.} (2010). {Child Labor
  and Household Wealth: Theory and Empirical Evidence of an Inverted-U}.
  \textit{Journal of Development Economics}, \textbf{91}, 8--14.

\bibitem[{Basu and Van(1998)}]{basu:1998}
\textsc{---} and \textsc{Van, P.~H.} (1998). {The Economics of Child Labor}.
  \textit{American Economic Review}, \textbf{88}, 412--427.

\bibitem[{Bester and Hansen(2009)}]{bester1:2009}
\textsc{Bester, C.~A.} and \textsc{Hansen, C.} (2009). {Identification of
  Marginal Effects in a Nonparametric Correlated Random Effects Model}.
  \textit{Journal of Business and Economic Statistics}, \textbf{27}, 235--250.

\bibitem[{Bhalotra and Heady(2003)}]{bhalotra:2003}
\textsc{Bhalotra, S.} and \textsc{Heady, C.} (2003). {Child Farm Labor: The
  Wealth Paradox}. \textit{World Bank Economic Review}, \textbf{17}, 197--227.

\bibitem[{Bi$\o$rn(2004)}]{biorn:2004}
\textsc{Bi$\o$rn, E.} (2004). {Regression Systems for Unbalanced Panel Data: A
  Stepwise Maximum Likelihood Procedure }. \textit{Journal of Econometrics},
  \textbf{122}, 281--291.

\bibitem[{Blundell \textit{et~al.}(2007)Blundell, MaCurdy and
  Meghir}]{blundell:2007}
\textsc{Blundell, R.}, \textsc{MaCurdy, T.} and \textsc{Meghir, C.} (2007).
  {Chapter 69 Labor Supply Models: Unobserved Heterogeneity, Nonparticipation
  and Dynamics}. \textit{Handbook of Econometrics}, vol.~6, Elsevier, pp. 4667
  -- 4775.

\bibitem[{Blundell and Powell(2003)}]{blundell:2003}
\textsc{---} and \textsc{Powell, J.} (2003). {Endogeneity in Nonparametric and
  Semiparametric Regression Models}. In M.~Dewatripont, L.~Hansen and
  S.~Turnovsky (eds.), \textit{Advances in Economics and Econonometrics: Theory
  and Applications, Eighth World Congress}, vol.~2, Cambridge: Cambridge
  University Press.

\bibitem[{Blundell and Powell(2004)}]{blundell:2004}
\textsc{---} and \textsc{---} (2004). {Endogeneity in Semiparametric Binary
  Response Models}. \textit{Review of Economic Studies}, \textbf{71}, 655--679.

\bibitem[{Blundell and Smith(1994)}]{blundell:1994}
\textsc{---} and \textsc{Smith, R.~J.} (1994). {Coherency and Estimation in
  Simultaneous Models with Censored or Qualitative Dependent Variables}.
  \textit{Journal of Econometrics}, \textbf{64}~(1-2), 355--373.

\bibitem[{Cadre \textit{et~al.}(2013)Cadre, Pelletier and Pudlo}]{cadre:2013}
\textsc{Cadre, B.}, \textsc{Pelletier, B.} and \textsc{Pudlo, P.} (2013).
  {Estimation of Density Level Sets with a given Probability Content}.
  \textit{Journal of Nonparametric Statistics}, \textbf{25}~(1), 261--272.

\bibitem[{Chamberlain(1984)}]{chamberlain:1984}
\textsc{Chamberlain, G.} (1984). {Panel Data}. In Z.~Griliches and M.~D.
  Intriligator (eds.), \textit{Handbook of Econometrics}, vol.~2, Elsevier.

\bibitem[{Chamberlain(2010)}]{chamberlain:2010}
\textsc{---} (2010). {Binary Response Models for Panel Data: Identification and
  Information}. \textit{Econometrica}, \textbf{78}, 159--168.

\bibitem[{Cockburn and Dostie(2007)}]{Cockburn:2007}
\textsc{Cockburn, J.} and \textsc{Dostie, B.} (2007). {Child Work and
  Schooling: The Role of Household Asset Profiles and Poverty in Rural
  Ethiopia}. \textit{Journal of African Economies}, \textbf{16}, 519--563.

\bibitem[{Constantinou and Dawid(2017)}]{constantinou:2017}
\textsc{Constantinou, P.} and \textsc{Dawid, A.~P.} (2017). Extended
  conditional independence and applications in causal inference. \textit{Annals
  of Statistics}, \textbf{45}~(6), 2618--2653.

\bibitem[{D'Haultf\oe{}uille and F\'{e}vrier(2015)}]{dhault:2015}
\textsc{D'Haultf\oe{}uille, X.} and \textsc{F\'{e}vrier, P.} (2015).
  {Identification of Nonseparable Triangular Models with Discrete Instruments}.
  \textit{Econometrica}, \textbf{83}~(3), 1199--1210.

\bibitem[{Dumas(2007)}]{dumas:2007}
\textsc{Dumas, C.} (2007). {Why do Parents make their Children Work? A Test of
  the Poverty Hypothesis in Rural Areas of Burkina Faso}. \textit{Oxford
  Economic Papers}, \textbf{59}, 301--329.

\bibitem[{Efron(2010)}]{efron:2010}
\textsc{Efron, B.} (2010). \textit{{Large-Scale Inference: Empirical Bayes
  Methods for Estimation, Testing, and Prediction}}. Institute of Mathematical
  Statistics Monographs, Cambridge University Press.

\bibitem[{Fern\'{a}ndez-Val and Vella(2011)}]{val:2011}
\textsc{Fern\'{a}ndez-Val, I.} and \textsc{Vella, F.} (2011). Bias corrections
  for two-step fixed effects panel data estimators. \textit{Journal of
  Econometrics}, \textbf{163}~(2), 144 -- 162.

\bibitem[{Florens \textit{et~al.}(2008)Florens, Heckman, Meghir and
  Vytlacil}]{Florens:2008}
\textsc{Florens, J.}, \textsc{Heckman, J.~J.}, \textsc{Meghir, C.} and
  \textsc{Vytlacil, E.} (2008). {Identification of Treatment Effects Using
  Control Functions in Models With Continuous, Endogenous Treatment and
  Heterogeneous Effects}. \textit{Econometrica}, \textbf{76}, 1191--1206.

\bibitem[{Greene(2004)}]{Greene:2004}
\textsc{Greene, W.} (2004). {Convenient Estimators for the Panel Probit Model:
  Further Results}. \textit{Empirical Economics}, \textbf{29}, 21--47.

\bibitem[{Gu and Koenker(2017)}]{gu:2017}
\textsc{Gu, J.} and \textsc{Koenker, R.} (2017). {Unobserved Heterogeneity in
  Income Dynamics: An Empirical Bayes Perspective}. \textit{Journal of Business
  \& Economic Statistics}, \textbf{35}~(1), 1--16.

\bibitem[{Guarino \textit{et~al.}(2015)Guarino, Maxfield, Reckase, Thompson and
  Wooldridge}]{guarino:2015}
\textsc{Guarino, C.~M.}, \textsc{Maxfield, M.}, \textsc{Reckase, M.~D.},
  \textsc{Thompson, P.~N.} and \textsc{Wooldridge, J.~M.} (2015). {An
  Evaluation of Empirical Bayes's Estimation of Value-Added Teacher Performance
  Measures}. \textit{Journal of Educational and Behavioral Statistics},
  \textbf{40}~(2), 190--222.

\bibitem[{Hall \textit{et~al.}(2004)Hall, Racine and Li}]{hall:2004}
\textsc{Hall, P.}, \textsc{Racine, J.} and \textsc{Li, Q.} (2004).
  {Cross-Validation and the Estimation of Conditional Probability Densities}.
  \textit{Journal of the American Statistical Association}, \textbf{99},
  1015--1026.

\bibitem[{Hayfield and Racine(2008)}]{hayfield:2008}
\textsc{Hayfield, T.} and \textsc{Racine, J.} (2008). {Nonparametric
  Econometrics: The np Package}. \textit{Journal of Statistical Software,
  Articles}, \textbf{27}~(5), 1--32.

\bibitem[{Heiss and Winschel(2008)}]{heiss:2008}
\textsc{Heiss, F.} and \textsc{Winschel, V.} (2008). {Likelihood Approximation
  by Numerical Integration on Sparse Grids}. \textit{Journal of Econometrics},
  \textbf{144}, 62--80.

\bibitem[{Hnatkovska \textit{et~al.}(2012)Hnatkovska, Lahiri and
  Paul}]{hnatkovska:2012}
\textsc{Hnatkovska, V.}, \textsc{Lahiri, A.} and \textsc{Paul, S.} (2012).
  {Castes and Labor Mobility}. \textit{American Economic Journal: Applied
  Economics}, \textbf{4}, 274--307.

\bibitem[{Hoderlein and Sherman(2015)}]{hoderlein:2015}
\textsc{Hoderlein, S.} and \textsc{Sherman, R.} (2015). {Identification and
  Estimation in a Correlated Random Coefficients Binary Response Model}.
  \textit{Journal of Econometrics}, \textbf{188}, 135--149.

\bibitem[{Hoderlein and White(2012)}]{hoderlein:2012}
\textsc{---} and \textsc{White, H.} (2012). {Nonparametric Identification in
  Nonseparable Panel Data Models with Generalized Fixed Effects}.
  \textit{Journal of Econometrics}, \textbf{168}, 300--314.

\bibitem[{Honor\'{e} and Kyriazidou(2000)}]{honore:2000}
\textsc{Honor\'{e}, B.} and \textsc{Kyriazidou, E.} (2000). {Panel Data
  Discrete Choice Models With Lagged Dependent Variables}.
  \textit{Econometrica}, \textbf{68}, 839--874.

\bibitem[{Horowitz(2009)}]{horowitz:2009}
\textsc{Horowitz, J.~L.} (2009). \textit{{Semiparametric and Nonparametric
  Methods in Econometrics}}. Springer, 2nd edn.

\bibitem[{Imbens and Manski(2004)}]{imbens:2004}
\textsc{Imbens, G.~W.} and \textsc{Manski, C.~F.} (2004). {Confidence Intervals
  for Partially Identified Parameters}. \textit{Econometrica}, \textbf{72}~(6),
  1845--1857.

\bibitem[{Imbens and Newey(2009)}]{imbens:2009}
\textsc{---} and \textsc{Newey, W.~K.} (2009). {Identification and Estimation
  of Triangular Simultaneous Equations Models without Additivity}.
  \textit{Econometrica}, \textbf{77}, 1481--1512.

\bibitem[{James and Stein(1961)}]{james:1961}
\textsc{James, W.} and \textsc{Stein, C.} (1961). Estimation with quadratic
  loss. In \textit{Proceedings of the Fourth Berkeley Symposium on Mathematical
  Statistics and Probability, Volume 1: Contributions to the Theory of
  Statistics}, Berkeley, Calif.: University of California Press, pp. 361--379.

\bibitem[{Kasy(2011)}]{kasy:2011}
\textsc{Kasy, M.} (2011). {Identification in Triangular Systems Using Control
  Functions}. \textit{Econometric Theory}, \textbf{27}~(3).

\bibitem[{Khan \textit{et~al.}(2020)Khan, Ponomareva and Tamer}]{khan:2020}
\textsc{Khan, S.}, \textsc{Ponomareva, M.} and \textsc{Tamer, E.} (2020).
  {Identification of Dynamic Binary Response Models}. {Harvard University,
  Working Paper}.

\bibitem[{Kim and Petrin(2017)}]{kim:2017}
\textsc{Kim, K.~i.} and \textsc{Petrin, A.} (2017). {A Generalized
  Non-parametric Instrumental Variable-Control Function Approach to Estimation
  in Non-linear Settings}, \textsc{W}orking Paper.

\bibitem[{Kitazawa(2021)}]{kitazawa:2021}
\textsc{Kitazawa, Y.} (2021). {Transformations and Moment Conditions for
  Dynamic Fixed Effects Logit Models}. \textit{Journal of Econometrics},
  \textbf{Forthcoming}.

\bibitem[{Liang and Zeger(1986)}]{liang:1986}
\textsc{Liang, K.~Y.} and \textsc{Zeger, S.~L.} (1986). {Longitudinal Data
  Analysis using Generalized Linear Models}. \textit{Biometrika}, \textbf{73},
  13--22.

\bibitem[{Manski(1988)}]{manski:1988}
\textsc{Manski, C.~F.} (1988). {Identification of Binary Response Models}.
  \textit{Journal of the American Statistical Association}, \textbf{83},
  729--738.

\bibitem[{Miller(1981)}]{miller:1981}
\textsc{Miller, K.~S.} (1981). {On the Inverse of the Sum of Matrices}.
  \textit{Mathematics Magazine}, \textbf{54}~(2), 67--72.

\bibitem[{Mundlak(1978)}]{mundlak:1978}
\textsc{Mundlak, Y.} (1978). {On the Pooling of Time Series and Cross Section
  Data}. \textit{Econometrica}, \textbf{46}, 69--85.

\bibitem[{Newey(1984)}]{newey:1984}
\textsc{Newey, W.~K.} (1984). {A Method of Moment Interpretation of Sequential
  Estimators}. \textit{Economics Letters}, \textbf{14}, 201--206.

\bibitem[{Papke and Wooldridge(2008)}]{papke:2008}
\textsc{Papke, L.~E.} and \textsc{Wooldridge, J.~M.} (2008). {Panel Data
  Methods for Fractional Response Variables with an application to Test Pass
  Rates}. \textit{Journal of Econometrics}, \textbf{145}, 121--133.

\bibitem[{Rivers and Vuong(1988)}]{rivers:1988}
\textsc{Rivers, D.} and \textsc{Vuong, Q.~H.} (1988). {Limited information
  estimators and exogeneity tests for simultaneous probit models}.
  \textit{Journal of Econometrics}, \textbf{39}~(3), 347--366.

\bibitem[{Rothe(2009)}]{rothe:2009}
\textsc{Rothe, C.} (2009). {Semiparametric Estimation of Binary Response Models
  with Endogenous Regressors}. \textit{Journal of Econometrics}, \textbf{153},
  51--64.

\bibitem[{Semykina and Wooldridge(2018)}]{semykina:2018}
\textsc{Semykina, A.} and \textsc{Wooldridge, J.~M.} (2018). {Binary Response
  Panel Data Models with Sample Selection and Self-Selection}. \textit{Journal
  of Applied Econometrics}, \textbf{33}~(2), 179--197.

\bibitem[{Torgovitsky(2015)}]{torgovitsky:2015}
\textsc{Torgovitsky, A.} (2015). {Identification of Nonseparable Models Using
  Instruments With Small Support}. \textit{Econometrica}, \textbf{83},
  1185--1197.

\bibitem[{Wooldridge(2010{\natexlab{a}})}]{wooldridge:2010}
\textsc{Wooldridge, J.~M.} (2010{\natexlab{a}}). \textit{Econometric Analysis
  of Cross Section and Panel Data}. Mit Press.

\bibitem[{Wooldridge(2010{\natexlab{b}})}]{wooldridge:2010a}
\textsc{---} (2010{\natexlab{b}}). \textit{{Solutions Manual and Supplementary
  Materials for Econometric Analysis of Cross Section and Panel Data}}.
  Cambridge, MA: MIT Press.

\bibitem[{Wooldridge(2015)}]{wooldridge:2015}
\textsc{---} (2015). {Control Function Methods in Applied Econometrics}.
  \textit{Journal of Human Resources}, \textbf{50}~(2), 420--445.

\bibitem[{Wooldridge(2019)}]{wooldridge:2019}
\textsc{---} (2019). {Correlated Random Effects Models with Unbalanced Panels}.
  \textit{Journal of Econometrics}, \textbf{211}~(1), 137--150.

\end{thebibliography}
\bibliographystyle{ecca}

\clearpage	

\appendix

% add "Appendix" to the section heading
\newcommand{\appsection}[1]{\let\oldthesection\thesection
	\renewcommand{\thesection}{Section \oldthesection:}
	\section{#1}\let\thesection\oldthesection}

% example (otherwise, use just \section)
\renewcommand{\theequation}{A-\arabic{equation}}
% redefine the command that creates the equation no.
\setcounter{equation}{0}  % reset counter
%\section*{APPENDIX}  % use *-form to suppress numbering

\section{Proofs}

\begin{customlem}{2}\label{alm:1} 
	
	(a) Let $X\equiv (\pmb{x}^{\prime}_{1},\ldots,\pmb{x}^{\prime}_{T})^{\prime}$ and $Z\equiv(\pmb{z}_{1}\ldots \pmb{z}_{T} )$. If $\pmb{x}_{t}$ is specified as
	\begin{align}\label{aeq:1}
	\pmb{x}_{t} = \pi\pmb{z}_{t} + \underbrace{\bar{\pi}\bar{\pmb{z}}+
		\pmb{a}}_{\pmb{\alpha}} + \pmb{\epsilon}_{t}, t \in \{1,\ldots, T\} ,
	\end{align}
	where $ \bar{\pmb{z}} = \frac{1}{T}\sum_{t=1}^{T}\pmb{z}_{t} $, $ Z \independent (\pmb{a},\pmb{\epsilon}_{t}) $, $ \pmb{a}\independent \pmb{\epsilon}_{t} $, $ \pmb{\epsilon}_{t} $ is i.i.d., $ \pmb{a} \sim \text{N}(0, \Lambda_{\alpha\alpha})$ and $ \pmb{\epsilon}_{t} \sim \text{N}(0, \Sigma_{\epsilon\epsilon})$, then
	\begin{align}\nonumber
	\E( \pmb{\alpha}|X, Z) \equiv \hat{\pmb{\alpha}}(X, Z, \Theta_{1} )=\bar{\pi}\bar{\pmb{z}}+ \E( \pmb{a}|X, Z) = \bar{\pi}\bar{\pmb{z}}+ \Omega\Sigma_{\epsilon\epsilon}^{-1}\sum_{t=1}^{T}(\pmb{x}_{t}-\pi\pmb{z}_{t}  -\bar{\pi}\bar{\pmb{z}}),
	\end{align}
	where $ \Omega=[T\Sigma_{\epsilon\epsilon}^{-1}+\Lambda_{\alpha\alpha}^{-1}]^{-1} $ is the conditional variance of $ \pmb{a} $ given $ X $ and $ Z $.

	(b) Suppose we have a single endogenous variable, $ x_{t} $. Let $X\equiv(x_{1},\ldots,x_{T})$ and define $Z\equiv(\pmb{z}_{1}\ldots \pmb{z}_{T} )$. Suppose $ x_{t} $ is given by
	\begin{align}\label{aeq:2}
	x_{t} = \pi\pmb{z}_{t} + \bar{\pi}\bar{\pmb{z}}+ a + \epsilon_{t}, t = 1, \ldots, T, 
	\end{align}
	where $ \bar{\pmb{z}} = \frac{1}{T}\sum_{t=1}^{T}\pmb{z}_{t} $. If the errors, $\pmb{\epsilon} \equiv (\epsilon_{1}, \ldots, \epsilon_{T})^{\prime} $, are normally distributed with mean 0 and are non-spherical such that  $ \E(\pmb{\epsilon}\pmb{\epsilon}^{\prime}) = \Omega_{\epsilon\epsilon}$,  $ a \sim N(0,\sigma^{2}_{\alpha})$, $ a\independent \pmb{\epsilon} $, and $ Z \independent (\pmb{a},\pmb{\epsilon}) $, then   
	\begin{align}\nonumber
	\E( a|X, Z) =  \hat{a}(X, Z )=(x_{1}-\pi\pmb{z}_{1} -\bar{\pi}\bar{\pmb{z}})\omega_{1}  + \ldots + (x_{T}-\pi\pmb{z}_{T}-\bar{\pi}\bar{\pmb{z}})\omega_{T}, 
	\end{align}
	where $ (\omega_{1}, \ldots, \omega_{T})^{\prime} = \frac{\Omega_{\epsilon\epsilon}^{-1}e_{T}}{(e^{\prime}_{T}\Omega_{\epsilon\epsilon}^{-1}e_{T} + \sigma^{-2}_{\alpha})} $ and $ e_{T} $ is a vector of ones of dimension $ T $.

\end{customlem}

\begin{proof}
	
	(a)  To obtain $ \hat{\pmb{a}}(X, Z) = \E( \pmb{a}|X, Z)$, we first derive $ f(\pmb{a}|X,
	Z) $, the conditional density function of $ \pmb{a} $ given $ X $ and $ Z $. By Bayes' rule we have
	\begin{align}\label{aeq:4}
	f(\pmb{a}|X,
	Z)=\frac{f(X, Z|\pmb{a})f(\pmb{a}) }{f(X, Z)} = \frac{f(X| Z, \pmb{a})f(Z |\pmb{a})f(\pmb{a}) }{f(X|Z)f(Z)}= \frac{f(X| Z, \pmb{a})f(\pmb{a}) }{f(X| Z)},
	\end{align}
	where the last equality is obtained because $Z$ is independent of the residual individual effects, $\pmb{a}$; that is, $f(Z|\pmb{a}) = f(Z)$. 
	
	Since $ \pmb{a}\independent \pmb{\epsilon}_{t} $, $ \pmb{\epsilon}_{t} $ is i.i.d., $ \pmb{a} \sim \text{N}(0, \Lambda_{\alpha\alpha})$ and $ \pmb{\epsilon}_{t} \sim \text{N}(0, \Sigma_{\epsilon\epsilon})$ then, given (\ref{aeq:1}), it implies that $X$, given $Z$, is normally distributed with mean $ ( (\pi\pmb{z}_{1} + \bar{\pi}\bar{\pmb{z}})^{\prime}, \ldots, (\pi\pmb{z}_{T} + \bar{\pi}\bar{\pmb{z}}))^{\prime} $, and variance $ \Sigma= I_{T} \otimes\Sigma_{\epsilon\epsilon}  + E_{T} \otimes\Lambda_{\alpha\alpha} $, where $ I_{T} $ is an identity matrix of dimension $ T $ and $ E_{T} $ is a $ T\times T $ matrix of ones. That is, $ f(X| Z) $ in (\ref{aeq:4}) is given by 
	\begin{align}\nonumber
	f(X| Z) = \frac{1}{\sqrt{(2\pi)^{mT}|\Sigma|} }\exp\biggr(-\frac{1}{2}R^{\prime}\Sigma^{-1}R\biggr), \text{  where }  R = X -  \begin{bmatrix} \pi\pmb{z}_{1} + \bar{\pi}\bar{\pmb{z}} \\ \vdots \\ \pi\pmb{z}_{T} + \bar{\pi}\bar{\pmb{z}}\end{bmatrix} = \begin{bmatrix} \pmb{r}_{1}\\  \vdots \\ \pmb{r}_{T} \end{bmatrix} 
	\end{align}
	and $m = d_{x}$ is the dimension of $ \pmb{x}_{t} $. Since $ \rank(E_{T}) = 1 $, we can use example 5 in \cite{miller:1981}, which is on the inverse of sum of two Kronecker products, to obtain 
	\begin{align}\nonumber
	&\Sigma^{-1} = I_{T} \otimes\Sigma_{\epsilon\epsilon}^{-1} - E_{T} \otimes[\Sigma_{\epsilon\epsilon} + \text{tr}(E_{T})\Lambda_{\alpha\alpha}]^{-1}\Lambda_{\alpha\alpha}\Sigma_{\epsilon\epsilon}^{-1} \text{ and } \\\nonumber
	&|\Sigma| = |\Sigma_{\epsilon\epsilon}|^{(T-1)}|\Sigma_{\epsilon\epsilon} + \text{tr}(E_{T})\Lambda_{\alpha\alpha}|   \text{ where $ \text{tr}(E_{T}) = T $}, 
	\end{align}
	which allows us to write $ f(X| Z) $ as
	\begin{align}\nonumber
	f(X| Z) &= \frac{1}{\sqrt{(2\pi)^{mT}|\Sigma_{\epsilon\epsilon}|^{(T-1)}|\Sigma_{\epsilon\epsilon} + T\Lambda_{\alpha\alpha}|} }\times\\\label{aeq:6}&\exp\biggr(-\frac{1}{2}\biggr[R^{\prime}[I_{T} \otimes\Sigma_{\epsilon\epsilon}^{-1}]R-\sum_{t=1}^{T}\pmb{r}_{t}^{\prime}[\Sigma_{\epsilon\epsilon} + T\Lambda_{\alpha\alpha}]^{-1}\Lambda_{\alpha\alpha}\Sigma_{\epsilon\epsilon}^{-1}\sum_{t=1}^{T}\pmb{r}_{t} \biggr]\biggr).
	\end{align}

	Since $ f(X| Z, \pmb{a} ) = f((\pmb{\epsilon}^{\prime}_{1}, \ldots \pmb{\epsilon}_{T}^{\prime} )^{\prime}) $, $ \pmb{\epsilon}_{t} $'s are i.i.d., $ \pmb{\epsilon}_{t} \sim \text{N}(0, \Sigma_{\epsilon\epsilon})$, and $ \pmb{a} \sim \text{N}(0, \Lambda_{\alpha\alpha})$, $ f(X| Z, \pmb{a} )f(\pmb{a}) $  in (\ref{aeq:4}) is
	\begin{align}\nonumber
	f(X| Z, \pmb{a})f(\pmb{a}) &= \frac{1}{\sqrt{(2\pi)^{mT+m}|\Sigma_{\epsilon\epsilon}|^{T}|\Lambda_{\alpha\alpha}|} }\times\\\label{aeq:7}&\exp\biggr(-\frac{1}{2}\biggr[(R-e_{T}\otimes\pmb{a})^{\prime}[I_{T} \otimes\Sigma_{\epsilon\epsilon}]^{-1}(R-e_{T}\otimes\pmb{a})+ \pmb{a}^{\prime}\Lambda_{\alpha\alpha}^{-1}\pmb{a}\biggr]\biggr),
	\end{align}	 
	where $ R-e_{T}\otimes\pmb{a} = (\pmb{\epsilon}^{\prime}_{1}, \ldots \pmb{\epsilon}_{T}^{\prime} )^{\prime} $, $ e_{T} $ being vector of ones of dimension $ T $. 
	
	The following matrix results, 
	\begin{enumerate}
		
		\item $ (A_{m\times m} \otimes B_{n\times n})^{-1}  = A_{m\times m}^{-1} \otimes B^{-1}_{n\times n} $,
		
		\item $ (A_{p\times q} \otimes B_{r\times s})(C_{q \times k} \otimes D_{s\times l}) = A_{p\times q}C_{q\times k} \otimes B_{r\times s}D_{s\times l} $ and
		
		\item $ e^{\prime}_{T}I_{T}e_{T} = T $,
	\end{enumerate}
	allow us to write the expression in the square parenthesis in (\ref{aeq:7}) as 
	\begin{align}\label{aeq:8}
	&R^{\prime}[I_{T} \otimes\Sigma_{\epsilon\epsilon}^{-1}]R - \pmb{a}^{\prime}\Sigma_{\epsilon\epsilon}^{-1}\sum_{t=1}^{T}\pmb{r}_{t}-  \sum_{t=1}^{T}\pmb{r}_{t}^{\prime}\Sigma_{\epsilon\epsilon}^{-1}\pmb{a}+ \pmb{a}^{\prime}[\Lambda_{\alpha\alpha}^{-1} +T\Sigma_{\epsilon\epsilon}^{-1} ]\pmb{a}.
	\end{align}
	
	Using the results in  (\ref{aeq:6}),  (\ref{aeq:7}) and (\ref{aeq:8}) and the result that $ |A^{-1}| = |A|^{-1} $, if $ A $ is nonsingular, we get
	\begin{align}\nonumber
	f(\pmb{a}|X, Z) &= \frac{f(X| Z, \pmb{a})f(\pmb{a})}{f(X| Z)}\\\nonumber
	&= \frac{1}{\sqrt{(2\pi)^{m}|\Sigma_{\epsilon\epsilon}[\Sigma_{\epsilon\epsilon} + T\Lambda_{\alpha\alpha}]^{-1}\Lambda_{\alpha\alpha}|} }\times\\\nonumber&\exp\biggr(-\frac{1}{2}\biggr[\sum_{t=1}^{T}\pmb{r}_{t}^{\prime}[\Sigma_{\epsilon\epsilon} + T\Lambda_{\alpha\alpha}]^{-1}\Lambda_{\alpha\alpha}\Sigma_{\epsilon\epsilon}^{-1}\sum_{t=1}^{T}\pmb{r}_{t} - \pmb{a}^{\prime}\Sigma_{\epsilon\epsilon}^{-1}\sum_{t=1}^{T}\pmb{r}_{t}-\\ \label{aeq:9} &\hspace{2cm}\sum_{t=1}^{T}\pmb{r}_{t}^{\prime}\Sigma_{\epsilon\epsilon}^{-1}\pmb{a}+ \pmb{a}^{\prime}[\Lambda_{\alpha\alpha}^{-1} +T\Sigma_{\epsilon\epsilon}^{-1} ]\pmb{a}\biggr]\biggr).
	\end{align}

	Let $ [\Lambda_{\alpha\alpha}^{-1} +T\Sigma_{\epsilon\epsilon}^{-1} ] = \Omega^{-1} $, then $ \sum_{t=1}^{T}\pmb{r}_{t}^{\prime}[\Sigma_{\epsilon\epsilon} + T\Lambda_{\alpha\alpha}]^{-1}\Lambda_{\alpha\alpha}\Sigma_{\epsilon\epsilon}^{-1}\sum_{t=1}^{T}\pmb{r}_{t} $ and  $ \sum_{t=1}^{T}\pmb{r}^{\prime}_{t}\Sigma_{\epsilon\epsilon}^{-1} $, 
	in (\ref{aeq:9}), after a few matrix manipulations, can be written as
	\begin{align}\label{aeq:10}
	\sum_{t=1}^{T}\pmb{r}_{t}^{\prime}[\Sigma_{\epsilon\epsilon} + T\Lambda_{\alpha\alpha}]^{-1}\Lambda_{\alpha\alpha}\Sigma_{\epsilon\epsilon}^{-1}\sum_{t=1}^{T}\pmb{r}_{t} = \sum_{t=1}^{T}\pmb{r}_{t}^{\prime}\Sigma_{\epsilon\epsilon}^{-1}\Omega\Omega^{-1}\Omega\Sigma_{\epsilon\epsilon}^{-1}\sum_{t=1}^{T}\pmb{r}_{t} \text{  and }
	\end{align}  
	\begin{align}\label{aeq:11}
	\sum_{t=1}^{T}\pmb{r}^{\prime}_{t}\Sigma_{\epsilon\epsilon}^{-1}=\sum_{t=1}^{T}\pmb{r}^{\prime}_{t}\Sigma_{\epsilon\epsilon}^{-1}\Omega\Omega^{-1} \text{  respectively. }
	\end{align}

	Given (\ref{aeq:10}) and (\ref{aeq:11}), we can write $ f(\pmb{a}|X, Z) $ in (\ref{aeq:9}) as 
	\begin{align}\nonumber
	f(\pmb{a}|X, Z) &=\frac{1}{\sqrt{(2\pi)^{m}|\Omega|} }\exp\biggr(-\frac{1}{2}\biggr[ \biggr[\pmb{a}-\Omega\Sigma_{\epsilon\epsilon}^{-1}\sum_{t=1}^{T}\pmb{r}_{t}\biggr]^{\prime}\Omega^{-1}\biggr[\pmb{a}-\Omega\Sigma_{\epsilon\epsilon}^{-1}\sum_{t=1}^{T}\pmb{r}_{t}\biggr]  \biggr]\biggr).
	\end{align} 
	In other words,  $ \pmb{a} $, given $ X $ and $ Z $, is normally distributed with conditional mean 
	\begin{align}\nonumber
	\E( \pmb{a}|X, Z) = \hat{\pmb{a}}(X, Z)=\Omega\Sigma_{\epsilon\epsilon}^{-1}\sum_{t=1}^{T}(\pmb{x}_{t}-\pi\pmb{z}_{t}-\bar{\pi}\bar{\pmb{z}}_{t}) 
	\end{align} and conditional variance $\Omega =\Sigma_{\epsilon\epsilon} [\Sigma_{\epsilon\epsilon} + T\Lambda_{\alpha\alpha} ]^{-1}\Lambda_{\alpha\alpha}$.

	(b) While discussing the restrictions imposed on the reduced form equation, we had stated that when $ d_{x} = 1 $, the assumption that $ a $ and $ \epsilon_{t} $ are completely independent of $ Z $ can be weakened to allow for non-spherical error components. Suppose that $  \epsilon_{t}, t = 1, \ldots, T $ are serially dependent such that $\pmb{\epsilon} \equiv (\epsilon_{1}, \ldots, \epsilon_{T})^{\prime} $ normally distributed with  $\E(\pmb{\epsilon}\pmb{\epsilon}^{\prime}) = \Omega_{\epsilon\epsilon}$  and $ a $ is normally distributed and  is heteroscedastic as in \cite{baltagi:2010}.
	
	To obtain $ \hat{a}(X, Z) = \E( a|X, Z)$, as in part (a), we first derive $ f(a|X,
	Z) $.  Using the fact that $ Z \independent a $, by an application of Bayes' rule, as in part (a), equation (\ref{aeq:4}), we have $ f(a|X, Z)= \frac{f( X| Z, a)f(a) }{f(X|Z)}, $ where $ f(a) $ is the normal density function of $a$.

	Now, since in (\ref{aeq:2}), $ a \independent \pmb{\epsilon}  $, $ a \sim N(0,\sigma^{2}_{\alpha}) $, and $ \pmb{\epsilon}\sim N(0,\Omega_{\epsilon\epsilon})$, it implies that $X$, given $Z$, is normally distributed with mean $ ( (\pi\pmb{z}_{1} + \bar{\pi}\bar{\pmb{z}}), \ldots, (\pi\pmb{z}_{T} + \bar{\pi}\bar{\pmb{z}}))^{\prime} $, and variance $ \Sigma= \Omega_{\epsilon\epsilon}  + \sigma^{2}_{\alpha}e_{T}e_{T}^{\prime} $, where  $ e_{T} $ is a vector of ones of dimension $ T $. That is, 
	\begin{align}\label{aeq:12}
	f(X|Z) = \frac{1}{\sqrt{(2\pi)^{T}|\Sigma|}}\exp(-\frac{1}{2}R^{\prime}\Sigma^{-1}R), \text{  where }  R = X -  \begin{bmatrix} \pi\pmb{z}_{1} + \bar{\pi}\bar{\pmb{z}} \\ \vdots \\ \pi\pmb{z}_{T} + \bar{\pi}\bar{\pmb{z}}\end{bmatrix},
	\end{align}
	and where by Sherman-Morrison formula,  
	$\Sigma^{-1}= \Omega_{\epsilon\epsilon}^{-1}  - \frac{\sigma^{2}_{\alpha}\Omega_{\epsilon\epsilon}^{-1}e_{T}e_{T}^{\prime}\Omega_{\epsilon\epsilon}^{-1}}{1+e_{T}^{\prime} \Omega_{\epsilon\epsilon}^{-1} e_{T}} $, and $ |\Sigma| =  |\Omega_{\epsilon\epsilon}|(1+\sigma^{2}_{\alpha}e_{T}^{\prime} \Omega_{\epsilon\epsilon}^{-1} e_{T}) $. %Hence 
	%	\begin{align}	f(X|Z) = \frac{1}{\sqrt{(2\pi)^{T}|\Omega_{\epsilon\epsilon}|(1+\sigma^{2}_{\alpha}e_{T}^{\prime} \Omega_{\epsilon\epsilon}^{-1} e_{T})}}\exp\biggr(-\frac{1}{2}R^{\prime}\biggr[\Omega_{\epsilon\epsilon}^{-1}  - \frac{\sigma^{2}_{\alpha}\Omega_{\epsilon\epsilon}^{-1}e_{T}e_{T}^{\prime}\Omega_{\epsilon\epsilon}^{-1}}{1+e_{T}^{\prime} \Omega_{\epsilon\epsilon}^{-1} e_{T}}\biggr]R\biggr).	\end{align}
	
	Since $ X$ given $ (Z, a) $ has the same distribution as $ \pmb{\epsilon} = R - ae_{T}  $, we have  
	\begin{align}\label{aeq:13}
	f( X| Z, a)f(a) = \frac{1}{\sqrt{(2\pi)^{T+1}|\Omega_{\epsilon\epsilon}|\sigma^{2}_{\alpha}}}  \exp\biggr(-\frac{1}{2}[(R - ae_{T})^{\prime}\Omega_{\epsilon\epsilon}^{-1}(R - ae_{T})+ \frac{a^{2}}{\sigma^{2}_{\alpha}}]\biggr).
	\end{align}
	
	Finally, because $ f(a|X,Z)=\frac{f(X| Z, a)f(a)}{f(X| Z)} $, using (\ref{aeq:12}) and (\ref{aeq:13}), it can be shown that $ a $ given $ X $ and $ Z $ is normally distributed with conditional mean 
	\begin{align}\nonumber
	\E(a|X, Z) =\hat{a}(X, Z )=(x_{1}-\pi\pmb{z}_{1}-\bar{\pi}\bar{\pmb{z}})\omega_{1}  + \ldots + (x_{T}-\pi\pmb{z}_{T}-\bar{\pi}\bar{\pmb{z}})\omega_{T}, 
	\end{align}
	where $ (\omega_{1}, \ldots, \omega_{T})^{\prime} = \frac{\Omega_{\epsilon\epsilon}^{-1}e_{T}}{(e_{T}^{\prime}\Omega_{\epsilon\epsilon}^{-1}e_{T} + \sigma^{-2}_{\alpha})} $, and conditional variance, $ \sigma^{2}_{\alpha}(\sigma^{2}_{\alpha}e_{T}^{\prime}\Omega_{\epsilon\epsilon}^{-1}e_{T} +1)^{-1}  $.
	
\end{proof}

\begin{aprop}
	Let  $ Z \indep (\theta, \pmb{\alpha}, \zeta_{t}, \pmb{\epsilon}_{t})$. When $ \theta$ and $ \pmb{\alpha}$  are correlated and so are $ \zeta_{t}$ and $ \pmb{\epsilon}_{t} $, then $ \theta + \zeta_{t}\indep X | V$ whereas $\theta + \zeta_{t}\nindep X | \pmb{\upsilon}_{t}$, where  
	$ \pmb{\upsilon}_{t} =  \pmb{\alpha}+ \pmb{\epsilon}_{t} = \pmb{x}_{t}+ \pi\pmb{z}_{t} $ and $ V \equiv (\pmb{\upsilon}_{1}, \ldots, \pmb{\upsilon}_{T})$. 
\end{aprop}

\begin{proof}
	
	Now, to show that $ r_{t} =\theta + \zeta_{t}\indep X | V$, we can show that 
	\begin{align}\label{eq:8}
	\E(f(r_{t})|X, V)=\E(f(r_{t})|V),
	\end{align}
	where $ f $ is real, bounded and measurable function \citep[see Proposition 2.3 in ][]{constantinou:2017}. 
	
	Since $ X = \pi Z + V $, there is one-to-one mapping between $(X, V) $ and $(\pi Z, V) $, and therefore the conditioning $ \sigma $-algebra, $\sigma(X, V) $, is same as the $ \sigma $-algebra, $\sigma(\pi Z, V) $. Hence, 
	\begin{align}\label{eq:9}
	\E(f(r_{t})|X, V )=\E(f(r_{t})|\pi Z, V). 
	\end{align}
	Since $ r_{t} $ and $ V $ are independent of $ \pi Z $, we get $ \E(f(r_{t})|\pi Z, V)=\E(f(r_{t})| V) $, and therefore we have $ \E(f(r_{t})|X, V)=\E(f(r_{t})| V) $, which is what we wanted to show. 
	
	When the control function is $ \pmb{\upsilon}_{t} $, we have 
	\begin{align}\nonumber
	\E(f(r_{t})|X, \pmb{\upsilon}_{t} )&=\E(f(r_{t})|X_{-t}, \pmb{x}_{t}, \pmb{\upsilon}_{t} )=\E(f(r_{t})|X_{-t}, \pi \pmb{z}_{t} + \pmb{\upsilon}_{t} , \pmb{\upsilon}_{t}) = \E(f(r_{t})|X_{-t}, \pi \pmb{z}_{t}, \pmb{\upsilon}_{t}),
	\end{align}
	where the second equality follows because $\pmb{x}_{t}= \pi\pmb{z}_{t} + \pmb{\upsilon}_{t}$ and third by the same logic by which we get (\ref{eq:9}). Because $r_{t} = \theta + \zeta_{t} $  and $ \upsilon_{s} =  \pmb{\alpha} +  \pmb{\epsilon}_{s} $, $ s\neq t$, are correlated even after conditioning on $ \pmb{\upsilon}_{t} $, so, conditional on $ \pmb{\upsilon}_{t} $, $ r_{t} $ is correlated with $ \pmb{x}_{s} =\pi \pmb{z}_{s} +  \pmb{\upsilon}_{s} $, $ s\neq t$; that is, $ \E(f(r_{t})|X_{-t}, \pi\pmb{z}_{t}, \pmb{\upsilon}_{t})\neq \E(f(r_{t})| \pi\pmb{z}_{t}, \pmb{\upsilon}_{t}) = \E(f(r_{t})| \pmb{\upsilon}_{t})$.

\end{proof}

\begin{atheo}\label{alm:2}
	If  (i) $ \rank(\E(\pmb{x}_{t}\pmb{x}_{t}^{\prime}))=d_{x}$; (ii) $ \rank(\Pi)=d_{x} $, where $ \Pi = \begin{pmatrix}
	\pi & \bar{\pi}
	\end{pmatrix} $; (iii) $ \rank(\E((\pmb{z}^{\prime}_{t},\bar{\pmb{z}}^{\prime} )^{\prime}(\pmb{z}^{\prime}_{t},\bar{\pmb{z}}^{\prime} )))=k$  where $ k = \dim((\pmb{z}^{\prime}_{t},\bar{\pmb{z}}^{\prime} )^{\prime}) $; and (iv) if assumption AS \ref{as:3} holds so that the covariance matrices of $\pmb{\epsilon}_{t}$ and $\pmb{\alpha}$ are of full rank, then $\rank(\E(\mathbb{X}_{t}\mathbb{X}_{t}^{\prime})) = 3d_{x} $. 
\end{atheo}

\begin{proof}

	Now, condition (i) of the lemma is the ``rank condition" for the standard probit model when $ \pmb{x}_{it} $ is exogenous and the object of interest is $ \pmb{\varphi} $ or marginal effects. This condition is assumed to hold true. Similarly, condition (iii) is the rank condition for the identification of the reduced form coefficients, $  \Pi = \begin{pmatrix}
	\pi & \bar{\pi}\end{pmatrix}  $, which is also assumed to hold.  
	
	To begin with, without loss of generality assume that $ \pmb{z}_{t} $ is uncorrelated with the individual effects $\pmb{\alpha}$ so that $ \bar{\pi}\bar{\pmb{z}} = 0 $. This implies that we can ignore $ \bar{\pmb{z}} $ in the reduced form equation (\ref{aeq:1}) and consider only the dimension of $ \pmb{z}_{t} $, which is $ d_{z} $, in condition (iii) of the lemma, and that  
	\begin{align}\nonumber
	\Pi=\pi, \textrm{  } \hat{\pmb{\alpha}} = \hat{\pmb{a}}, \textrm{  } \hat{\pmb{\epsilon}}_{t} = \pmb{x}_{t} - \pi\pmb{z}_{t} - \hat{\pmb{a}}, \textrm{  } k = d_{z}, \textrm{ and }  \pi \textrm{ a }  d_{x}\times d_{z} \text{  matrix}.
	\end{align}
	
	Since, $ \hat{\pmb{\epsilon}}_{t}= \pmb{\upsilon}_{t} - \hat{\pmb{\alpha}}$, then if $\E(\mathbb{X}_{t}\mathbb{X}_{t}^{\prime})$, where $ \mathbb{X}_{t} = (\pmb{x}^{\prime}_{t},  
	\hat{\pmb{\epsilon}}_{t}^{\prime}, 
	\hat{\pmb{\alpha}}^{\prime})^{\prime} $, were to be invertible (or equivalently have a rank of $ 3d_{x} $) so, too, would  
	\begin{align}\nonumber
	\E\begin{bmatrix} \begin{bmatrix}
	\pmb{x}_{t} \\ \pmb{\upsilon}_{t}\\ \hat{\pmb{\alpha}} \end{bmatrix}
	\begin{bmatrix} \pmb{x}^{\prime}_{t} & \pmb{\upsilon}_{t}^{\prime} &\hat{\pmb{\alpha}}^{\prime}
	\end{bmatrix} \end{bmatrix} \text{ be.}
	\end{align}  
	This is equivalent to stating that the columns of $ (\pmb{x}^{\prime}_{t},  \pmb{\upsilon}_{t}^{\prime}, \hat{\pmb{\alpha}}^{\prime}) $ are linearly independent. Now, if the columns of $ (\pmb{x}^{\prime}_{t},  \pmb{\upsilon}_{t}^{\prime}, \hat{\pmb{\alpha}}^{\prime}) $ are linearly independent, then every subset of its columns, too, is linearly independent. Thus, to show the statement of the theorem to be true, we can show that
	\begin{align}\label{aeq:28}
	&\rank(\E[(\pmb{x}^{\prime}_{t},  
	\pmb{\upsilon}_{t}^{\prime}, \hat{\pmb{\alpha}}^{\prime})^{\prime} 
	\pmb{x}_{t}^{\prime}]) = d_{x},\\ \label{aeq:21}
	&\rank(\E[(\pmb{x}^{\prime}_{t},  
	\pmb{\upsilon}_{t}^{\prime}, \hat{\pmb{\alpha}}^{\prime})^{\prime} 
	\pmb{\upsilon}_{t}^{\prime}]) = d_{x} \text{ and } \\ \label{aeq:27}
	&\rank(\E[(\pmb{x}^{\prime}_{t}, \pmb{\upsilon}_{t}^{\prime}, 
	\hat{\pmb{\alpha}}^{\prime})^{\prime}\hat{\pmb{\alpha}}^{\prime}])= d_{x}, 
	\end{align}
	as both $ \pmb{\upsilon}_{t}$ and $\hat{\pmb{\alpha}}^{\prime} $ are vectors of dimension $ d_{x} $.
	
	Since $ \E[(\pmb{x}^{\prime}_{t},  \pmb{\upsilon}_{t}^{\prime}, \hat{\pmb{\alpha}}^{\prime})^{\prime} \pmb{x}_{t}^{\prime}]$ in eq. (\ref{aeq:28}) has $3d_{x}$ rows and $ d_{x} $ columns, $ \rank(\E[(\pmb{x}^{\prime}_{t}, \pmb{\upsilon}_{t}^{\prime}, \hat{\pmb{\alpha}}^{\prime})^{\prime} 
	\pmb{x}_{t}^{\prime}]) \leq d_{x} $. Consider $  
	\pmb{x}_{t}\pmb{x}_{t}^{\prime} $ in eq. (\ref{aeq:28}). Now, by condition (i) of the lemma, $ \rank(\E[ \pmb{x}_{t}\pmb{x}_{t}^{\prime}])= d_{x} $, which implies that $ d_{x} $ columns of $ \E[ \pmb{x}_{t}\pmb{x}_{t}^{\prime}] $ are linearly independent. This then implies that the $  d_{x}  $ columns of $ \E[(\pmb{x}^{\prime}_{t},\pmb{\upsilon}_{t}^{\prime}, \hat{\pmb{\alpha}}^{\prime})^{\prime} \pmb{x}_{t}^{\prime}] $ are also linearly independent; that is, it implies that $ \rank(\E[(\pmb{x}^{\prime}_{t},  \pmb{\upsilon}_{t}^{\prime},  \hat{\pmb{\alpha}}^{\prime})^{\prime} \pmb{x}_{t}^{\prime}]) \geq d_{x}$. Thus, we can conclude that $\rank(\E[(\pmb{x}^{\prime}_{t},  \pmb{\upsilon}_{t}^{\prime},  \hat{\pmb{\alpha}}^{\prime})^{\prime} \pmb{x}_{t}^{\prime}]) = d_{x}$.
	
	Again, given that $ \E[(\pmb{x}^{\prime}_{t}, \pmb{\upsilon}_{t}^{\prime}, \hat{\pmb{\alpha}}^{\prime})^{\prime} \pmb{\upsilon}_{t}^{\prime}]$ in eq. (\ref{aeq:21}) has $3d_{x}$ rows and $ d_{x} $ columns, $ \rank(\E[(\pmb{x}^{\prime}_{t}, \pmb{\upsilon}_{t}^{\prime}, \hat{\pmb{\alpha}}^{\prime})^{\prime} 
	\pmb{\upsilon}_{t}^{\prime}]) \leq d_{x} $. If we can show the $ \rank(\E[ \pmb{\upsilon}_{t}\pmb{\upsilon}_{t}^{\prime}])$ in eq. (\ref{aeq:21}) is $ d_{x} $, then it would imply that  the $  d_{x}  $ columns of $ \E[(\pmb{x}^{\prime}_{t},\pmb{\upsilon}_{t}^{\prime}, \hat{\pmb{\alpha}}^{\prime})^{\prime} 
	\pmb{\upsilon}_{t}^{\prime}] $ are also linearly independent, implying that $ \rank(\E[(\pmb{x}^{\prime}_{t},  \pmb{\upsilon}_{t}^{\prime},  \hat{\pmb{\alpha}}^{\prime})^{\prime} \pmb{\upsilon}_{t}^{\prime}]) \geq d_{x}$. Thus, we would be able to show that $\rank(\E[(\pmb{x}^{\prime}_{t},  \pmb{\upsilon}_{t}^{\prime},  \hat{\pmb{\alpha}}^{\prime})^{\prime} \pmb{\upsilon}_{t}^{\prime}]) = d_{x} $, as desired.
	
	To show that $\E[\pmb{\upsilon}_{t}\pmb{\upsilon}_{t}^{\prime}]$  has a full column rank of $ d_{x} $, is equivalent to showing that $\pmb{\upsilon}^{\prime}\pmb{c} = (\pmb{x}^{\prime} - \pmb{z}_{t}^{\prime}\pi^{\prime} )\pmb{c}\neq 0 $ almost surely (\textit{a.s.}) whenever $\pmb{c}\neq 0$, $ \pmb{c} \in \mathbb{R}^{d_{x}}$, and $ \pmb{x} \neq \pi\pmb{z}_{t} $ \textit{a.s.}.

	Now, by condition (i) of the lemma 
	\begin{align}\label{aeq:22}
	\pmb{x}^{\prime}_{t}\pmb{c} \neq 0  \text{  \textit{a.s.}}.
	\end{align}
	By condition (ii), according to which the rank of $\pi^{\prime}$ is $ d_{x} $,
	\begin{align}\label{aeq:23}
	\pi^{\prime}\pmb{c} = \pmb{c}_{\pi}\neq 0, \text{ where $ \pmb{c}_{\pi} $ is of dimension, $ d_{z}\times 1 $}.
	\end{align}
	By condition (iii) of the lemma and (\ref{aeq:23}), we have  
	\begin{align}\label{aeq:24}
	\pmb{z}^{\prime}_{t}\pi^{\prime}\pmb{c} = \pmb{z}^{\prime}_{t}\pmb{c}_{\pi} \neq 0  \text{  \textit{a.s.}}.
	\end{align}
	From (\ref{aeq:22}) and (\ref{aeq:24}) we can conclude that $\E[\pmb{\upsilon}_{t}\pmb{\upsilon}_{t}^{\prime}]$ has a full column rank of $ d_{x} $; thus we are able to establish (\ref{aeq:21}).  
	
	Similarly, to show (\ref{aeq:27}) to be true, we can show that $\hat{\pmb{\alpha}}^{\prime}\pmb{c}\neq 0$ a.s. whenever $\pmb{c}\neq 0$, $ \pmb{c} \in \mathbb{R}^{d_{x}}$. Now, 
	\begin{align}\nonumber
	\hat{\pmb{\alpha}}^{\prime}\pmb{c} = \biggr([T\Sigma_{\epsilon\epsilon}^{-1}+\Lambda_{\alpha\alpha}^{-1}]^{-1}\Sigma_{\epsilon\epsilon}^{-1}\biggr(\sum_{t=1}^{T}\pmb{\upsilon}_{t}\biggr)\biggr)^{\prime}\pmb{c}&=\biggr(\sum_{t=1}^{T}\pmb{\upsilon}^{\prime}_{t}\biggr)\Sigma_{\epsilon\epsilon}^{-1}[T\Sigma_{\epsilon\epsilon}^{-1}+\Lambda_{\alpha\alpha}^{-1}]^{-1}\pmb{c}\\\nonumber &=\biggr(\sum_{t=1}^{T}\pmb{\upsilon}^{\prime}_{t}\biggr)\pmb{c}_{\alpha}.
	\end{align} 
	Because  $\Sigma_{\epsilon\epsilon}$ and $\Lambda_{\alpha\alpha}$, the covariance matrices of $\pmb{\epsilon}_{t}$ and $\pmb{\alpha}$ respectively, are symmetric positive definite matrices, $\Sigma_{\epsilon\epsilon}^{-1}[T\Sigma_{\epsilon\epsilon}^{-1}+\Lambda_{\alpha\alpha}^{-1}]^{-1} $ is nonsingular. This implies that  $ \Sigma_{\epsilon\epsilon}^{-1}[T\Sigma_{\epsilon\epsilon}^{-1}+\Lambda_{\alpha\alpha}^{-1}]^{-1}\pmb{c} =\pmb{c}_{\alpha} \neq 0 $. 
	
	Since by (\ref{aeq:22})  and (\ref{aeq:24}),  $ (\pmb{x}^{\prime}_{t} - \pmb{z}^{\prime}_{t}\pi^{\prime})\bar{\pmb{c}}_{\alpha} = \pmb{\upsilon}^{\prime}_{t}\pmb{c}_{\alpha}\neq  0$, when $ \pmb{x}_{t} \neq \pi \pmb{z}_{t} $ \textit{a.s.}, we therefore have
	\begin{align}\label{aeq:25}
	\hat{\pmb{\alpha}}^{\prime}\pmb{c} &= \biggr(\sum_{t=1}^{T}\pmb{\upsilon}^{\prime}_{t}\biggr)\pmb{c}_{\alpha} \neq 0 \textrm{ \textit{a.s.}, thus establishing (\ref{aeq:27}).} 
	\end{align} 
	
	Having established (\ref{aeq:28}), (\ref{aeq:21}) and (\ref{aeq:27}), we can conclude that $\rank(\E(\mathbb{X}_{t}\mathbb{X}_{t}^{\prime})) = 3d_{x} $. 
	
\end{proof}

\begin{atheo}
	Let the linear structural model be $ y_{t}  =  \pmb{x}_{t}^{\prime}\pmb{\varphi}  + \theta + \zeta_{t}, t \in \{1,\ldots, T\} $. Under assumption AS \ref{as:3} and condition (ii) of lemma \ref{lm:0} we can write $ \theta $ as $\theta =  \E(\theta|Z) + \tau = \pmb{\rho}_{\alpha}\bar{\pi}\bar{\pmb{z}} + \tau $, which allows us to write the linear structural equation as 
	\begin{align}\label{aeq:34}
	y_{t}  =  \pmb{x}_{t}^{\prime}\pmb{\varphi} + \pmb{\rho}_{\alpha}\bar{\pi}\bar{\pmb{z}} + \tau + \zeta_{t}. 
	\end{align}
	Under assumptions ACF \ref{ascf:1} and AS \ref{as:4} in the main text, we can write $ \theta + \zeta_{t} $ in the structural equation as $ \theta + \zeta_{t} = \E(\theta+ \zeta_{t}|X,Z) +\eta_{t} = \hat{\pmb{\alpha}}^{\prime}\pmb{\varphi}_{\alpha}+\hat{\pmb{\epsilon}}^{\prime}_{t}\pmb{\varphi}_{\epsilon} + \eta_{t} $, which allows us to write the linear structural equation as 
	\begin{align}\label{aeq:35}
	y_{t} =  \pmb{x}_{t}^{\prime}\pmb{\varphi}_{CF}  + \hat{\pmb{\alpha}}^{\prime}\pmb{\varphi}_{\alpha}+\hat{\pmb{\epsilon}}^{\prime}_{t}\pmb{\varphi}_{\epsilon} + \eta_{t}
	\end{align}
	where $ \tau + \zeta_{t} $ is mean independent of $ Z $ and  $\eta_{t} = \theta+\zeta_{t} - \E(\theta+\zeta_{t}|X, Z)$ is mean independent of $ X $ and $ Z $. 
	
	Given that $ X $ and $ \tau + \zeta_{t} $ are dependent, let $ \hat{\pmb{\varphi}}_{IV} $ denote the estimated coefficient of $ \pmb{x}_{t} $  when (\ref{aeq:34}) is estimated by two stage least squares (2SLS) with instruments, $ \bar{\pmb{z}} = \frac{1}{T}\sum_{t=1}^{T}\pmb{z}_{t} $, $\ddot{\pmb{z}}_{t} = \pmb{z}_{t} - \bar{\pmb{z}}$, and $ \bar{\pmb{z}}^{\prime}\bar{\pi} $. And let $ \hat{\pmb{\varphi}}_{CF} $ denote the estimate of $ \pmb{\varphi} $ obtained from estimating (\ref{aeq:35}) by pooling the data. In the following, we show that $ \hat{\pmb{\varphi}}_{CF} = \hat{\pmb{\varphi}}_{IV} $.  
\end{atheo}

\begin{proof} 
	Without any loss of generality assume that there is only endogenous regressor, $ x_{t} $; the proof generalizes to multiple regressors. Given assumption in AS \ref{as:3}, the reduced form is given by 
	\begin{align}\label{aeq:36}
	x_{t} = \pmb{z}^{\prime}_{t}\pi + \bar{\pmb{z}}^{\prime}\bar{\pi}+
	a + \epsilon_{t} ,
	\end{align}
	where $ \bar{\pmb{z}}^{\prime}\bar{\pi}+
	a= \alpha  $ and $ a \sim \N(0, \sigma^{2}_{\alpha} ) $ and $ \epsilon_{t} \sim \N(0, \sigma^{2}_{\epsilon} ) $. Here, $ \bar{y} $,  $ \bar{x} $ and  $ \bar{\pmb{z}} $ denote group means of $ y_{t} $, $ x_{t} $, and $ \pmb{z}_{t} $ respectively.

	Now, by \cite{mundlak:1978} we know that $ \pi $  is equal to the within estimator, $ \pi_{w} $, of the reduced form equation (\ref{aeq:36}) and $ \bar{\pi} = \pi_{b} -\pi_{w}  $, where $ \pi_{b} $ is the between estimator of $ \pi $ in (\ref{aeq:36}). Thus by lemma \ref{lm:1}, the control functions, $\hat{\alpha}$ and $\hat{\epsilon}_{t}$, in (\ref{aeq:35}) are respectively given by    
	\begin{align}\nonumber
	&\hat{\alpha} = \bar{\pmb{z}}^{\prime}(\pi_{b}-\pi_{w}) + (1-\lambda)(\bar{x} - \bar{\pmb{z}}^{\prime}\pi_{b})\\\nonumber
	&\hat{\epsilon}_{t} = x_{t} - \pmb{z}^{\prime}_{t}\pi_{w} - \bar{\pmb{z}}^{\prime}(\pi_{b}-\pi_{w})-(1-\lambda)(\bar{x} - \bar{\pmb{z}}^{\prime}\pi_{b}), \text{    where $  \lambda = \frac{\sigma^{2}_{\epsilon}}{\sigma^{2}_{\epsilon}+T\sigma^{2}_{\alpha}}$.}
	\end{align}
	
	It is convenient to write the panel model in (\ref{aeq:35}) in vector form as  
	\begin{align}\label{aeq:37}
	\textbf{Y} = \textbf{X}\varphi + \hat{\textbf{A}}\varphi_{\alpha} + \hat{\textbf{E}}\varphi_{\epsilon} + \pmb{\eta}, 
	\end{align}
	where \textbf{Y} is a $ NT \times 1 $ data matrix which has the observations on $ y_{it} $ stacked in $ NT $ rows; similarly, \textbf{X} has stacked observations on $ x_{it} $, $ \hat{\textbf{A}} $  on $ \hat{\alpha}_{i} $, and  $ \hat{\textbf{E}} $ on $ \hat{\epsilon}_{it} $. Now, given $ \hat{\alpha} $ and $ \hat{\epsilon}_{t} $, we can write $\hat{\textbf{A}} $  and  $ \hat{\textbf{E}} $, respectively, as 
	\begin{align}\nonumber
	&\hat{\textbf{A}} = P\textbf{Z}\bar{\pi} + (1-\lambda)P\textbf{R} \text{ and  } \hat{\textbf{E}} = Q\textbf{R} + \lambda P\textbf{R},
	\end{align}
	where \textbf{R} is a $ NT \times 1 $  matrix, which has the residuals of the reduced form equation (\ref{aeq:36}), $ r_{it} =  x_{t} - \pmb{z}^{\prime}_{t}\pi_{w} - \bar{\pmb{z}}^{\prime}(\pi_{b}-\pi_{w})$, stacked in $ NT $ rows, \textbf{Z} has stacked observations on $ \pmb{z}_{it} $,  $ P $ is a matrix that averages the observation across time for each individual, i.e., $P=I_{N}\otimes\tilde{J}_{T}$ where $J_{T}$ is a matrix of ones of dimension $T$ and $\tilde{J}_{T}=J_{T}/T$,  and $ Q=I_{NT}-P $ is a matrix that obtains the deviations from individual means. Given the above, we can write (\ref{aeq:37}) as 
	\begin{align} \label{aeq:38}
	\textbf{Y} = \textbf{X}\varphi + (P\textbf{Z}\bar{\pi} + (1-\lambda)P\textbf{R})\varphi_{\alpha} + (Q\textbf{R} + \lambda P\textbf{R})\varphi_{\epsilon} + \pmb{\eta}.  
	\end{align}
	
	Because $ P $ and $ Q $ are idempotent and $ P $ is orthogonal to $ Q $, premultiplying (\ref{aeq:38}) through out by $ Q $, we get
	\begin{align} \label{aeq:39}
	\ddot{\textbf{Y}} = \ddot{\textbf{X}}\varphi  + \ddot{\textbf{R}}\varphi_{\epsilon} + \ddot{\pmb{\eta}}, 
	\end{align}
	where $ \ddot{\textbf{Y}} $ is a $ NT \times 1 $  matrix that has $\ddot{y}_{it} = y_{it} -\bar{y}_{i} $, the deviations of $ y_{it} $ from the individual means, $ \bar{y}_{i} $, (within transformation) stacked in $ NT $ rows, and $ \ddot{\textbf{R}} = \ddot{\textbf{X}}-\ddot{\textbf{Z}}\pi_{w} $ is the matrix obtained after within transforming the residuals, $ \textbf{R} $. $ \ddot{\textbf{R}} $, incidentally, are the residuals obtained after applying within transformation to the reduced form equation (\ref{aeq:36}) and estimating it as a fixed effect model. 
	
	Now, premultiplying (\ref{aeq:38}) through out by $ P $, we get
	\begin{align} \label{aeq:40}
	\bar{\textbf{Y}} = \bar{\textbf{X}}\varphi +(\bar{\textbf{Z}}\bar{\pi} + (1-\lambda)\bar{\textbf{R}})\varphi_{\alpha} + \lambda\bar{\textbf{R}}\varphi_{\epsilon} + \bar{\pmb{\eta}}, 
	\end{align}
	where $ \bar{\textbf{Y}} $ is a $ NT \times 1 $  matrix that has the individual means, $ \bar{y}_{i} $, (between transformation) stacked in $ NT $ rows, and $ \bar{\textbf{R}} = \bar{\textbf{X}}-\bar{\textbf{Z}}\pi_{b} $ is the matrix obtained after between transforming the residuals, $ \textbf{R} $. $ \bar{\textbf{R}} $, incidentally, are the residuals obtained after applying between transformation to the reduced form equation (\ref{aeq:36}) and estimating it by OLS.
	
	Since the column space of the LHS variables, $  [\bar{\textbf{X}}, \bar{\textbf{Z}}\bar{\pi} + (1-\lambda)\bar{\textbf{R}},    \lambda\bar{\textbf{R}}] $,  in (\ref{aeq:40})  is same as that of $  [\bar{\textbf{X}}, \bar{\textbf{Z}}\bar{\pi} ,  \bar{\textbf{R}}] $, the projections of the two matrices will be the same. Therefore, the  estimates of $ \varphi $, $ \varphi_{\alpha} $, and  $ \varphi_{\epsilon} $ in (\ref{aeq:40}) will be the same from estimating the following equation:     
	\begin{align} \label{aeq:41}
	\bar{\textbf{Y}} = \bar{\textbf{X}}\varphi +\bar{\textbf{Z}}\bar{\pi}\varphi_{\alpha} + \bar{\textbf{R}}\varphi_{\epsilon} + \bar{\pmb{\eta}}, 
	\end{align}
	
	From (\ref{aeq:39}) and (\ref{aeq:41}), it is thus evident that estimating the following equation 
	\begin{align} \label{aeq:42}
	\textbf{Y} = \textbf{X}\varphi +\bar{\textbf{Z}}\bar{\pi}\varphi_{\alpha} +[ Q\textbf{R} + P\textbf{R}]\varphi_{\epsilon} + \pmb{\eta}, 
	\end{align}
	by OLS would yield the same result as estimating (\ref{aeq:38}) by OLS. 
	
	By Frisch-Waugh-Lovell Theorem, then $ \varphi $ and $ \varphi_{\alpha} $ in (\ref{aeq:42}) is estimated as  
	\begin{align}\nonumber
	\begin{bmatrix}
	\hat{\varphi}_{CF}\\
	\hat{\varphi}_{\alpha}
	\end{bmatrix} =   \begin{bmatrix}\begin{bmatrix}\textbf{X}^{\prime}\\\hat{\bar{\pi}}^{\prime}\bar{\textbf{Z}}^{\prime}\end{bmatrix} \begin{bmatrix} QM_{\ddot{\textbf{R}}} +PM_{\bar{\textbf{R}}}\end{bmatrix}\begin{bmatrix}\textbf{X} &\bar{\textbf{Z}}\hat{\bar{\pi}}\end{bmatrix}\end{bmatrix}^{-1} \begin{bmatrix}\textbf{X}^{\prime}\\\hat{\bar{\pi}}^{\prime}\bar{\textbf{Z}}^{\prime}\end{bmatrix}\begin{bmatrix} QM_{\ddot{\textbf{R}}} +PM_{\bar{\textbf{R}}}\end{bmatrix} \textbf{Y}, 
	\end{align}
	where $ \hat{\bar{\pi}} $ is the estimate of $ \bar{\pi} $, $ M_{\ddot{\textbf{R}}} = I_{NT} - \ddot{\textbf{R}}[\ddot{\textbf{R}}^{\prime}\ddot{\textbf{R}}]^{-1}\ddot{\textbf{R}}^{\prime}$  and $ M_{\bar{\textbf{R}}} = I_{NT} - \bar{\textbf{R}}[\bar{\textbf{R}}^{\prime}\bar{\textbf{R}}]^{-1}\bar{\textbf{R}}^{\prime}$. Since \textbf{X} and \textbf{Y} can be decomposed as $ \textbf{X} = \ddot{\textbf{X}} + \bar{\textbf{X}} $ and $ \textbf{Y} = \ddot{\textbf{Y}} + \bar{\textbf{Y}} $ respectively, we can write the above equation as 
	\begin{align}\nonumber
	\begin{bmatrix}
	\hat{\varphi}_{CF}\\
	\hat{\varphi}_{\alpha}
	\end{bmatrix} =   \begin{bmatrix}\begin{bmatrix}\ddot{\textbf{X}}^{\prime} + \bar{\textbf{X}}^{\prime}\\\hat{\bar{\pi}}^{\prime}\bar{\textbf{Z}}^{\prime}\end{bmatrix} \begin{bmatrix} QM_{\ddot{\textbf{R}}} +PM_{\bar{\textbf{R}}}\end{bmatrix}\begin{bmatrix}\ddot{\textbf{X}}+\bar{\textbf{X}} &\bar{\textbf{Z}}\hat{\bar{\pi}}\end{bmatrix}\end{bmatrix}^{-1} \begin{bmatrix}\ddot{\textbf{X}}^{\prime} + \bar{\textbf{X}}^{\prime}\\\hat{\bar{\pi}}^{\prime}\bar{\textbf{Z}}^{\prime}\end{bmatrix}\begin{bmatrix} QM_{\ddot{\textbf{R}}} +PM_{\bar{\textbf{R}}}\end{bmatrix} \begin{bmatrix}\ddot{\textbf{Y}} + \bar{\textbf{Y}}\end{bmatrix}.
	\end{align}
	Taking into account the orthogonality conditions in the above, the above simplifies to  
	\begin{align}\label{aeq:43}
	\begin{bmatrix}
	\hat{\varphi}_{CF}\\
	\hat{\varphi}_{\alpha}
	\end{bmatrix} =   \begin{bmatrix} \ddot{\textbf{X}}^{\prime}M_{\ddot{\textbf{R}}}\ddot{\textbf{X}} + \begin{bmatrix}\bar{\textbf{X}}^{\prime}\\\hat{\bar{\pi}}^{\prime}\bar{\textbf{Z}}^{\prime}\end{bmatrix} M_{\bar{\textbf{R}}}\begin{bmatrix}\bar{\textbf{X}} &\bar{\textbf{Z}}\hat{\bar{\pi}}\end{bmatrix} \end{bmatrix}^{-1} \begin{bmatrix} \ddot{\textbf{X}}^{\prime}M_{\ddot{\textbf{R}}}\ddot{\textbf{Y}} + \begin{bmatrix}\bar{\textbf{X}}^{\prime}\\\hat{\bar{\pi}}^{\prime}\bar{\textbf{Z}}^{\prime}\end{bmatrix} M_{\bar{\textbf{R}}}\bar{\textbf{Y}} \end{bmatrix}.
	\end{align}
	Because $\ddot{\textbf{R}} = M_{\ddot{\textbf{Z}}}\ddot{\textbf{X}} $, where $ M_{\ddot{\textbf{Z}}} = I_{NT} - \ddot{\textbf{Z}}[\ddot{\textbf{Z}}^{\prime}\ddot{\textbf{Z}}]^{-1}\ddot{\textbf{Z}}^{\prime} $, it can be verified that $ \ddot{\textbf{X}}^{\prime}M_{\ddot{\textbf{R}}}= \ddot{\textbf{X}}^{\prime}P_{\ddot{\textbf{Z}}} $,
	where $ P_{\ddot{\textbf{Z}}} = \ddot{\textbf{Z}}[\ddot{\textbf{Z}}^{\prime}\ddot{\textbf{Z}}]^{-1}\ddot{\textbf{Z}}^{\prime} $ is the projection matrix. And because $\bar{\textbf{R}} = M_{\bar{\textbf{Z}}}\bar{\textbf{X}} $, where $ M_{\bar{\textbf{Z}}} = I_{NT} - \bar{\textbf{Z}}[\bar{\textbf{Z}}^{\prime}\bar{\textbf{Z}}]^{-1}\bar{\textbf{Z}}^{\prime} $, it can be verified that $ \bar{\textbf{X}}^{\prime}M_{\bar{\textbf{R}}}= \bar{\textbf{X}}^{\prime}P_{\bar{\textbf{Z}}} $,
	where $ P_{\bar{\textbf{Z}}} = \bar{\textbf{Z}}[\bar{\textbf{Z}}^{\prime}\bar{\textbf{Z}}]^{-1}\bar{\textbf{Z}}^{\prime} $, and that $ \hat{\bar{\pi}}^{\prime}\bar{\textbf{Z}}^{\prime}M_{\bar{\textbf{R}}}= \hat{\bar{\pi}}^{\prime}\bar{\textbf{Z}}^{\prime}$. Thus, we can write $ \hat{\varphi}_{CF}$ and  $
	\hat{\varphi}_{\alpha} $ in (\ref{aeq:43}) as
	\begin{align}\label{aeq:44}
	\begin{bmatrix}
	\hat{\varphi}_{CF}\\
	\hat{\varphi}_{\alpha}
	\end{bmatrix} =   \begin{bmatrix} \ddot{\textbf{X}}^{\prime}P_{\ddot{\textbf{Z}}}\ddot{\textbf{X}} + \begin{bmatrix}\bar{\textbf{X}}^{\prime}\\\hat{\bar{\pi}}^{\prime}\bar{\textbf{Z}}^{\prime}\end{bmatrix} P_{\bar{\textbf{Z}}}\begin{bmatrix}\bar{\textbf{X}} &\bar{\textbf{Z}}\hat{\bar{\pi}}\end{bmatrix} \end{bmatrix}^{-1} \begin{bmatrix} \ddot{\textbf{X}}^{\prime}P_{\ddot{\textbf{Z}}}\ddot{\textbf{Y}} + \begin{bmatrix}\bar{\textbf{X}}^{\prime}\\\hat{\bar{\pi}}^{\prime}\bar{\textbf{Z}}^{\prime}\end{bmatrix} P_{\bar{\textbf{Z}}}\bar{\textbf{Y}} \end{bmatrix},
	\end{align}
	which is the same as the estimates $\hat{\varphi}_{IV}$ and $
	\hat{\rho}_{\alpha}$ when (\ref{aeq:34}) is estimated by the 2SLS using  $ \bar{\pmb{z}} $, $\ddot{\pmb{z}}_{t} = \pmb{z}_{t} - \bar{\pmb{z}}$ and $ \bar{\pmb{z}}^{\prime}\hat{\bar{\pi}} $ as instruments. 
	
	If only the within transform, equation (\ref{aeq:39}), is estimated, then by a similar process as above, beginning with the Frisch-Waugh-Lovell Theorem,  one obtains 
	\begin{align} 
	\hat{\varphi}^{w}_{CF} =   [\ddot{\textbf{X}}^{\prime}P_{\ddot{\textbf{Z}}}\ddot{\textbf{X}}]^{-1} \ddot{\textbf{X}}^{\prime}P_{\ddot{\textbf{Z}}} \ddot{\textbf{Y}}, 
	\end{align}
	which is same as the within estimate of $ \varphi^{w}_{IV} $ when (\ref{aeq:34})  is estimated by fixed effect two-stage least squares (FE2SLS) that utilizes $\ddot{\pmb{z}}_{t} = \pmb{z}_{t} - \bar{\pmb{z}}$ as instruments.
	
	If only the between transform, equation ( \ref{aeq:41}), is estimated, then, again, a similar process as the one used to derive (\ref{aeq:44}), one finally obtains 
	\begin{align}\nonumber
	\begin{bmatrix}
	\hat{\varphi}_{CF}^{b}\\
	\hat{\varphi}_{\alpha}^{b}
	\end{bmatrix} =   \begin{bmatrix}\begin{bmatrix}\bar{\textbf{X}}^{\prime}P_{\bar{\textbf{Z}}}\\\hat{\bar{\pi}}^{\prime}\bar{\textbf{Z}}^{\prime}\end{bmatrix}\begin{bmatrix}P_{\bar{\textbf{Z}}}\bar{\textbf{X}} &\bar{\textbf{Z}}\hat{\bar{\pi}}\end{bmatrix}\end{bmatrix}^{-1} \begin{bmatrix}\bar{\textbf{X}}^{\prime}P_{\bar{\textbf{Z}}}\\\hat{\bar{\pi}}^{\prime}\bar{\textbf{Z}}^{\prime}\end{bmatrix} \bar{\textbf{Y}}, 
	\end{align}
	which is the estimate of $ \varphi^{b}_{IV} $ and $\rho^{b}_{\alpha}$ when (\ref{aeq:34}) is between transformed and estimated by two-stage least squares (2SLS) with $ \bar{\pmb{z}}$  and $ \bar{\pmb{z}}^{\prime}\hat{\bar{\pi}} $ as instruments. After some algebraic manipulations, we get   
	\begin{align}\nonumber
	\hat{\varphi}^{b}_{CF} = 
	\begin{bmatrix}\bar{\textbf{X}}^{\prime}P_{\bar{\textbf{Z}}}M_{\bar{\textbf{Z}}\hat{\bar{\pi}}}P_{\bar{\textbf{Z}}}\bar{\textbf{X}}\end{bmatrix}^{-1}\bar{\textbf{X}}^{\prime}P_{\bar{\textbf{Z}}}M_{\bar{\textbf{Z}}\hat{\bar{\pi}}} \bar{\textbf{Y}}, 
	\end{align}
	where $ M_{\bar{\textbf{Z}}\hat{\bar{\pi}}} = I -\bar{\textbf{Z}}\hat{\bar{\pi}}[\hat{\bar{\pi}}^{\prime}\bar{\textbf{Z}}^{\prime}\bar{\textbf{Z}}\hat{\bar{\pi}}]^{-1}\hat{\bar{\pi}}^{\prime}\bar{\textbf{Z}}^{\prime}$. 
	
\end{proof}

\begin{customlem}{3}\label{alm:4}
	If the endogenous variables, $ \pmb{x} $, have large a support, then under AS 3, the support of the conditional distribution of $\hat{\pmb{\alpha}}(X, Z, \Theta_{1})$ and $\hat{\pmb{\epsilon}}_{t}(X, Z, \Theta_{1})$, conditional on $\pmb{x}_{t}=\bar{\pmb{x}}$, is same as the support of their marginal distribution.
\end{customlem}
\begin{proof} 
	
	(a) We have shown that the expected value of $ \pmb{\alpha} = \bar{\pi}\bar{\pmb{z}} + \pmb{a} $ and $ \pmb{\epsilon}_{t} $ given  $ Z $ and $ X  $, where  $ \pmb{a} $ and $ \pmb{\epsilon}_{t} $ are normally distributed with variances $ \Lambda_{\alpha\alpha} $ and $\Sigma_{\epsilon\epsilon}$ respectively, are given 
	\begin{align}\nonumber
	&\E(\pmb{\alpha}|X, Z) =\hat{\pmb{\alpha}}=\bar{\pi}\bar{\pmb{z}}+\hat{\pmb{a}} =\bar{\pi}\bar{\pmb{z}}+\sum_{t=1}^{T}\varOmega(\pmb{x}_{t}-\Pi Z_{t}) \text{ and }\\\nonumber
	&\E(\pmb{\epsilon}_{t}|X, Z) =\hat{\pmb{\epsilon}}_{t}=\pmb{x}_{t}-\Pi Z_{t}-\sum_{t=1}^{T}\varOmega(\pmb{x}_{t}-\Pi Z_{t}) \text{  respectively, }
	\end{align}
	where $ \Pi = (\bar{\pi}, \bar{\pi}) $, $Z_{t} = (\pmb{z}^{\prime}_{t},  \bar{\pmb{z}}^{\prime})^{\prime} $ and  $\varOmega =[T\Sigma_{\epsilon\epsilon}^{-1}+\Lambda_{\alpha\alpha}^{-1}]^{-1}\Sigma_{\epsilon\epsilon}^{-1}$.

	Because the support of $ \pmb{x}_{t} $ is $ \mathbb{R}^{d_{x}} $ and $\varOmega $ is a $ d_{x}\times d_{x} $ nonsingular matrix, 
	\begin{align}\nonumber
	\textrm{Supp}(\hat{\pmb{\alpha}}) =\textrm{Supp}(\hat{\pmb{\epsilon}}_{t})= \mathbb{R}^{d_{x}} \text{ whether or not $ \pmb{z}_{t} $ has a large support}. 
	\end{align} 
	
	Now fix $ \pmb{x}_{t}=\bar{\pmb{x}} $. Then, because the  $ \pmb{x}_{s}$'s, $ s\neq t $, are not restricted, we have
	\begin{align}\nonumber
	&\textrm{Supp}(\hat{\pmb{\alpha}}| \pmb{x}_{t}=\bar{\pmb{x}}) = \textrm{Supp}\biggr(\bar{\pi}\bar{\pmb{z}}+ \varOmega(\bar{\pmb{x}}-\Pi Z_{t}) + \sum_{s\neq t}\varOmega(\pmb{x}_{s}-\Pi Z_{s})\biggr) = \mathbb{R}^{d_{x}}\text{ and }\\\nonumber
	&\textrm{Supp}(\hat{\pmb{\epsilon}}_{t}| \pmb{x}_{t}=\bar{\pmb{x}}) = \textrm{Supp}\biggr( [I_{m}-\varOmega](\bar{\pmb{x}}-\Pi Z_{t}) - \sum_{s\neq t}\varOmega(\pmb{x}_{s}-\Pi Z_{s})\biggr) = \mathbb{R}^{d_{x}}.
	\end{align} 
	
	(b) When we have a single endogenous variable, $ x_{t} $, given by
	\begin{align}\nonumber
	x_{t} = \Pi Z_{t} + a + \epsilon_{t}, t = 1, \ldots, T, 
	\end{align}
	where the errors, $\pmb{\epsilon} \equiv (\epsilon_{1}, \ldots, \epsilon_{T})^{\prime} $, are non-spherical such that  $ \E(\pmb{\epsilon}\pmb{\epsilon}^{\prime}) = \Omega_{\epsilon\epsilon}$, a invertible $ T\times T $ matrix, and $ a $ is normally distributed with variance $ \sigma^{2}_{\alpha} $, then we showed that  
	\begin{align}\nonumber
	\hat{a}(X, Z, \Theta_{1} )=(x_{1}-\Pi Z_{1})\omega_{1}  + \ldots + (x_{T}-\Pi Z_{T})\omega_{T}, 
	\end{align}
	where $ (\omega_{1}, \ldots, \omega_{T})^{\prime} = \frac{\Omega_{\epsilon\epsilon}^{-1}e_{T}}{(e^{\prime}\Omega_{\epsilon\epsilon}^{-1}e_{T} + \sigma^{-2}_{\alpha})} $ and $ e_{T} $ is a vector of ones of dimension $ T $.
	
	Given that $x_{t}$'s have large supports, using a similar argument as in part (a), we get 
	\begin{align}\nonumber
	\textrm{Supp}(\hat{a}| x_{t}=\bar{x}) =  \mathbb{R} \textrm{  and   } \textrm{Supp}(\hat{\epsilon}_{t}| x_{t}=\bar{x}) =  \mathbb{R}.
	\end{align}
	
\end{proof}

% example (otherwise, use just \section)
\renewcommand{\theequation}{B-\arabic{equation}}
% redefine the command that creates the equation no.
\setcounter{equation}{0}  % reset counter
%\section*{APPENDIX}  % use *-form to suppress numbering

\section{Estimation Of Probit Conditional Mean Function}
If $\eta_{t}$ in equation (2.9) in main text is assumed to follow a normal distribution, then  
\begin{align}
\E(y_{t}|X, Z)=\Phi((\mathbb{X}_{t}^{\prime}\Theta_{2})/\sigma),
\end{align}
where $\mathbb{X}_{t} = (\pmb{x}_{t}^{\prime}, \hat{\pmb{\alpha}}^{\prime}(X, Z), \hat{\pmb{\epsilon}}_{t}^{\prime}(X, Z))$, $\Theta_{2}=(\pmb{\varphi}^{\prime}, \pmb{\rho}^{\prime}_{\alpha}, \pmb{\rho}^{\prime}_{\epsilon})^{\prime}$, and $\sigma^{2}$ is the variance of $\eta_{t}$. Since in probit models the coefficients can only be identified up to a scale, in this section with a slight abuse of notation we denote the scaled parameters,  $\frac{1}{\sigma} \Theta_{2} $, by $\Theta_{2} $.  To estimate $\Theta_{2}$, one can employ nonlinear least squares by pooling the data. However, as \citetalias{papke:2008} discuss, since $\var(y_{t}| (X, Z))$ will most likely be heteroscedastic and since there will be serial correlation across time in the joint distribution, $F(y_{1},\ldots,y_{T}|X, Z)$, the estimates, though consistent, will be estimated inefficiently resulting in biased standard errors. \citetalias{papke:2008} argue that modelling $F(y_{1},\ldots,y_{T}|X, Z)$ when $ y_{t} $'s are fractional response and applying MLE methods, while possible, is not trivial. Moreover, if the model for $F(y_{1},\ldots,y_{T}|X, Z)$  is misspecified but $\E(y_{t}|X, Z)$ is correctly specified, the MLE will be inconsistent for $\Theta_{2}$ and the resulting APEs. This is likely to be true when response outcomes are binary. 

To account for heteroscedasticity\footnote{Another possibility would be to assume the form of heteroscedasticity such as multiplicative heteroscedasticity as is common in heteroscedastic probit model. However, since it is likely that there will be serial dependence across time, we favour the GEE method proposed in this section, which can potentially account for both heteroscedasticity and serial dependence.} and serial dependence when all covariates are exogenous, \citetalias{papke:2008} employ the method of multivariate weighted nonlinear least squares (MWNLS) to obtain efficient estimates of $\Theta_{2}$. To get the correct estimates of the standard errors, the method requires is a parametric model of $\var(\textit{y}_{i}|X_{i}, Z_{i})$, where $\textit{y}_{i}$ is the $T \times 1$ vector of responses. Similar to \citetalias{papke:2008}'s, we specify $\var(y_{t}|X, Z)$ as
\begin{align}\label{eq:13}
\var(y_{t}|X, Z) = \tau\textrm{\textbf{m}}(\mathbb{X}_{t},\Theta_{2})(1-\textrm{\textbf{m}}(\mathbb{X}_{t},\Theta_{2})),
\end{align}
where $\textrm{\textbf{m}}(\mathbb{X}_{t},\Theta_{2})= \Phi(\mathbb{X}_{t}^{\prime}\Theta_{2})$ and $0<\tau\leq1$. For covariances, $\cov(y_{t} , y_{r} |X, Z)$, a ``working" version, which can be misspecified for $\var(\textit{y}|X, Z)$, is assumed. This, in the context of panel data, is what underlies the method of generalized estimating equation (GEE), as described in \citet{liang:1986}. The main advantage of GEE lies in the consistent and unbiased estimation of parameters' standard errors even when the correlation structure is misspecified. Also, GEE and MWNLS are asymptotically equivalent whenever they use the same estimates of the $T \times T$ positive definite matrix, $\var(\textit{y}|X, Z)$.

Generally, the conditional correlations, $\cov(y_{t} , y_{s} |X, Z)$, are a function of $X$ and $Z$. In the GEE literature, the ``working correlation matrix" is that which assumes the dependency structure to be invariant over all observations; that is, the correlations are not a function of $X$ and $Z$. Here we will focus on a particular correlation matrix that is suited for panel data applications with small $T$ . In the GEE literature it is called an ``exchangeable" correlation pattern. Exchangeable correlation assumes constant time dependency, so that all the off-diagonal elements of
the correlation matrix are equal. Though other correlation patterns such as ``autoregressive", which assumes the correlations to be an exponential function of the time lag, or  ``stationary $M$", which assumes constant correlations within equal time intervals could also be assumed.

GEE method suggests that parameter, $\rho$, that characterize $\var(\textit{y}|X, Z) = \textbf{V}(X, Z, \Theta_{2},\tau, \rho)$ can be estimated using simple functions of residuals, $u_{t}$,
\begin{align}\nonumber
u_{t}=y_{t}-\E(y_{t}|X, Z)=y_{t}-\textrm{\textbf{m}}(\mathbb{X}_{t},\Theta_{2}),
\end{align}
where the mean function, $\E(y_{t}|X, Z)$, is correctly specified. With the variance defined in (\ref{eq:13}), we can define standardized errors as
\begin{align}\nonumber
e_{t}=\frac{u_{t}}{\sqrt{\textrm{\textbf{m}}(\mathbb{X}_{t},\Theta_{2})(1-\textrm{\textbf{m}}(\mathbb{X}_{t},\Theta_{2}))}}.
\end{align}
Then we have $\var(e_{t}|X, Z)=\tau$. The exchangeability assumption is that the pairwise correlations between pairs of standardized errors are constant, say $\rho$. This, to reiterate, is a ``working" assumption that leads to an estimated variance matrix to be used in MWNLS. Neither consistency of the estimator of $\rho$, nor valid inference, will rest on exchangeability being true. 

To estimate a common correlation parameter, let $\tilde{\Theta}_{2}$ be a preliminary, consistent estimator of $\Theta_{2}$.   $\tilde{\Theta}_{2}$  could be the pooled ML estimate of the heteroscedastic probit model. Define the residuals, $\tilde{u}_{t},$ as $\tilde{u}_{t}= y_{t}-\textrm{\textbf{m}}(\mathbb{X}_{t},\tilde{\Theta}_{2})$ and the standardized residuals as
\begin{align}\nonumber
\tilde{e}_{t}=\frac{\tilde{u}_{t}}{\sqrt{\textrm{\textbf{m}}(\mathbb{X}_{t},\tilde{\Theta}_{2})(1-\textrm{\textbf{m}}(\mathbb{X}_{t},\tilde{\Theta}_{2}))}}.
\end{align}
Then a natural estimator of a common correlation coefficient is
\begin{align}
\tilde{\rho}=\frac{1}{NT(T - 1)}\sum^{ N}_{i=1}
\sum^{ T}_{t=1}\sum _{s\neq t}\tilde{e}_{it}\tilde{e}_{is}.
\end{align}
Under standard regularity conditions, without much restrictions on $\textrm{Corr}(e_{t} , e_{s}|X, Z)$, the plim of $\tilde{\rho}$ is
\begin{align}
\textrm{plim}( \tilde{\rho}) = \frac{1}{[T (T - 1)]}\sum ^{T}_{t=1}
\sum_{s\neq t}\E(e_{it}e_{is})\equiv\rho^{*}.
\end{align}

If $\textrm{Corr}(e_{t} , e_{s}|X, Z)$ happens to be the same for all $t \neq s$, then $\tilde{\rho}$ consistently estimates this constant correlation. Generally, it consistently estimates the average of these correlations across all $(t, s)$ pairs, which is defined as $\textbf{C}(\tilde{\rho})$.
Given the estimated $T \times T$ working correlation matrix, $\textbf{C}(\tilde{\rho})$, which has unity down its diagonal and $\tilde{\rho}$ everywhere else, we can
construct the estimated working variance matrix:
\begin{align}\nonumber
\textbf{V}(X, Z, \tilde{\Theta}_{2}, \tilde{\rho})=\textbf{D}(X, Z, \tilde{\Theta}_{2} )^{1/2}\textbf{C}( \tilde{\rho} )\textbf{D}(X, Z, \tilde{\Theta}_{2})^{1/2} = \textbf{V}(X, Z, \tilde{\Upsilon})
\end{align}
where $\textbf{D}(X, Z, \Theta_{2})$ is the $T \times T$ diagonal matrix with $\textrm{\textbf{m}}(\mathbb{X}_{t}, \Theta_{2})(1-\textrm{\textbf{m}}(\mathbb{X}_{t},\Theta_{2}))$ down its diagonal. (Note that dropping the variance scale
factor, $\tau$ , has no effect on estimation or inference.) 

Estimation by MWNLS then involves solving for $\hat{\Theta}_{2}$ by minimizing the following with respect to $\Theta_{2}$:
\begin{align}\label{eq:14}
\min_{\Theta_{2}}\sum^{N}_{i=1}[\mathbf{y}_{i} - \mathbf{m}_{i}(X_{i}, Z_{i}, \Theta_{2})]^{\prime}[\textbf{V}(X_{i}, Z_{i},\tilde{\Upsilon})]^{-1}[\mathbf{y}_{i} - \mathbf{m}_{i}(X_{i}, Z_{i}, \Theta_{2})],
\end{align}
where $\mathbf{m}_{i}(X_{i}, Z_{i}, \Theta_{2})$ is the $T \time 1$ vector with $t^{th}$ element being $\textrm{\textbf{m}}(\mathbb{X}_{it},\Theta_{2})$.

The requirement of GEE is that the mean model, $\E(y_{t}|X, Z)$,  be correctly specified, else the GEE approach to estimation can give inconsistent results. We have, given our identifying assumptions, shown that $\E(y_{t}|X, Z)=\Phi(\mathbb{X}_{t}^{\prime}\Theta_{2})$, and therefore we can employ GEE to account for serial correlation across time. Once the control functions have been estimated, one can then use the STATA command,``xtgee," which fits generalized linear models and allows one to specify the within-group correlation structure for the panels, to estimate  $ \Theta_{2} $.

% example (otherwise, use just \section)
\renewcommand{\theequation}{C-\arabic{equation}}
% redefine the command that creates the equation no.
\setcounter{equation}{0}  % reset counter
%\section*{APPENDIX}  % use *-form to suppress numbering

\section{Panel Probit Model with Random Coefficients in a Triangular System}
In this section we extend the model with random effect studied in section 2 to allow for random coefficients, and discuss identification of certain structural measures of interest. Consider the following binary choice random coefficient model 
\begin{align}\label{rc:1}
y_{it}=1\{y^{*}_{it} = \mathcal{X}_{it}^{\prime}\pmb{\varphi}_{i}  +\zeta_{it}>0\},
\end{align}
where $\mathcal{X}_{it}=(x_{it}, \pmb{w}^{\prime}_{it})^{\prime}$ and $\zeta_{it}$ are the idiosyncratic errors.  Here we consider a single continuous endogenous variable, $x_{it}$, with a large support. Let $ d_{w} $ be the dimension of the exogenous variables, $ \pmb{w}_{it} $. In (\ref{rc:1}), the random coefficients are $ \pmb{\varphi}_{i} = \pmb{\varphi}+\pmb{\theta}_{i} $, where $ \pmb{\varphi} $ is a $ (1+d_{w})\times 1 $ vector of constant means and $ \E(\pmb{\theta}_{i}) = 0 $. Thus, $ \pmb{\varphi} $ is the average
slope that we might be interested in.

The reduced form in the triangular system is given by: 
\begin{align}\label{rc:2}
x_{it} = \pmb{z}^{\prime}_{it}\pmb{\alpha}_{i} + \epsilon_{it},
\end{align}
where $ \pmb{z}_{it} = (\pmb{w}^{\prime}_{it}, \tilde{\pmb{z}}^{\prime}_{it})^{\prime} $. The dimension of the vector of instruments, $ \tilde{\pmb{z}}_{it} $, $ d_{z} $, is greater than or equal to 1.   $ \pmb{\alpha}_{i} = \pmb{\alpha} + \pmb{a}_{i} $ is the $ (d_{w}+d_{z}) \times 1 $ vector of random coefficients, where $ \pmb{\alpha} $  is a  vector of constants and $ \pmb{a}_{i} $ a  vector of stationary random
variables with zero means and constant variance-covariances. And finally, $ \epsilon_{it} $ is a scalar idiosyncratic term.

%Given that  $ \pmb{\varphi}_{i} = \pmb{\varphi}+\pmb{\theta}_{i} $ and $ \pmb{\alpha}_{i} = \pmb{\alpha} + \pmb{a}_{i} $, we can write (\ref{rc:1}) and (\ref{rc:2}) as \begin{align}\nonumber&y_{it}=1\{y^{*}_{it} = \mathcal{X}_{it}^{\prime}\pmb{\varphi} +\mathcal{X}_{it}^{\prime}\pmb{\theta}_{i} +\zeta_{it}>0\} \textrm{  and  }\\\label{rc:3}	&x_{it} = \pmb{z}^{\prime}_{it}\pmb{\alpha} + \pmb{z}^{\prime}_{it}\pmb{a}_{i} +\epsilon_{it}	\end{align}respectively.      

The identifying distributional restrictions are summarized as follows:
\begin{assrc}\label{arc:1}
	(a)	$(\pmb{\theta}_{i},  \pmb{\zeta}_{i}) , (\pmb{a}_{i}, \pmb{\epsilon}_{i}) \independent Z_{i} $ and (b) $ \pmb{\theta}_{i}, \pmb{a}_{i} \independent   \pmb{\zeta}_{i}, \pmb{\epsilon}_{i}$,   where   $Z_{i}\equiv (
	\pmb{z}_{i1}, \ldots , \pmb{z}_{iT})
	$ is a $ T\times(d_{w}+d_{z}) $ matrix, $\pmb{\zeta}_{i}\equiv (\zeta_{i1},\ldots, \zeta_{iT})^{\prime}$ , and $\pmb{\epsilon}_{i}\equiv (\epsilon_{i1},\ldots, \epsilon_{iT})^{\prime}$.
\end{assrc}
In the above assumption, $ \pmb{z}_{it} $ is independent of the random coefficients, $ (\pmb{\varphi}_{i}, \pmb{\alpha}_{i}) $, and the idiosyncratic errors, $ (\zeta_{it}, \epsilon_{it}) $. Also, as in the random effects model, we assume that the random coefficients and the idiosyncratic errors are independent of each other.

\begin{assrc}\label{arc:2} 
	\begin{align}\nonumber
	\pmb{\theta}_{i}, \zeta_{it}| X_{i}, Z_{i}, \pmb{a}_{i} &\sim \pmb{\theta}_{i}, \zeta_{it}| X_{i}-\E(X_{i}|Z_{i}, \pmb{a}_{i}), Z_{i}, \pmb{a}_{i}\\\nonumber &\sim  \pmb{\theta}_{i}, \zeta_{it}|\pmb{\epsilon}_{i}, Z_{i}, \pmb{a}_{i} \\\nonumber&\sim \pmb{\theta}_{i}, \zeta_{it}| \pmb{\epsilon}_{i}, \pmb{a}_{i},
	\end{align}
	where $X_{i}\equiv (x_{i1},\ldots, x_{iT})^{\prime}$ and $ \pmb{\epsilon}_{i} =X_{i}-\E(X_{i}|Z_{i}, \pmb{a}_{i}) =X_{i} - Z_{i}(\pmb{\alpha}+\pmb{a}_{i}) $.

\end{assrc}In RC \ref{arc:2} the assumption is that the dependence of the structural error terms $\pmb{\theta}_{i}$ and $\zeta_{it}$ on $X_{i}$, $Z_{i}$, and $\pmb{a}_{i}$ is completely characterized by the reduced form error components, $\pmb{\epsilon}_{i}$ and $\pmb{a}_{i}$. If given $(\epsilon_{it}, \pmb{a}_{i}) $ only contemporaneous correlations matter, then $ \pmb{\theta}_{i}, \zeta_{it}\independent \pmb{\epsilon}_{i,-t}|(\epsilon_{it}, \pmb{a}_{i}) $. 

As in the model for random effects, we specify the marginal distributions of $ \pmb{a}_{i} $ and $ \epsilon_{it} $. We assume that  
\begin{assrc}\label{arc:3}
	$ \pmb{a}_{i} \sim \text{N}(0, \Sigma_{a})$ and that $ \epsilon_{it} \sim \text{N}(0, \sigma^{2}_{\epsilon}) $.
\end{assrc}
Let $ \Theta_{1} \equiv \{\pmb{\alpha}, \Sigma_{a},  \sigma^{2}_{\epsilon} \}$ denote the set of parameters of random coefficient model in (\ref{rc:2}), the reduced form equation. The random coefficient model is a standard one, and most statistical packages have routines to estimate $ \Theta_{1} $.

%Since there is one-to-one correspondence between  $ (X_{i}, Z_{i}, \pmb{a}_{i}) $ and $ (Z_{i}, \pmb{\epsilon}_{i} , \pmb{a}_{i}) $, where $ \pmb{\epsilon}_{i} = X_{i}-\E(X_{i}|Z_{i}, \pmb{a}_{i}) = X_{i}-Z_{i}(\pmb{\alpha} +\pmb{a}_{i}) $, we have  $ \zeta_{it}, \pmb{\theta}_{i} | Z_{i}, X_{i} , \pmb{a}_{i}  \sim \zeta_{it},\pmb{\theta}_{i} | Z_{i}, \pmb{\epsilon}_{i} , \pmb{a}_{i}$. Then \textsc{Assumption} RC \ref{arc:1} implies that

Given assumption RC \ref{arc:2}, we have 
\begin{align}\nonumber
&\E(\pmb{\theta}_{i}|X_{i},Z_{i},\pmb{a}_{i})=\E(\pmb{\theta}_{i}|\pmb{a}_{i},\pmb{\epsilon}_{i})=\E(\pmb{\theta}_{i}|\pmb{a}_{i}) =  \pmb{\rho}_{\theta a}\pmb{a}_{i} \text{ and }\\\label{rc:4}
&\E(\zeta_{it}|X_{i},Z_{i},\pmb{a}_{i}) = \E(\zeta_{it}|\pmb{a}_{i},\pmb{\epsilon}_{i})= \E(\zeta_{it}| \pmb{\epsilon}_{i})= \pmb{\rho}_{\zeta\epsilon}\pmb{\epsilon}_{i},
\end{align}
where the second equality in each of the above follows from part (b) of assumption RC \ref{arc:1}. In the above, $\pmb{\rho}_{\theta a}$ is the $ (d_{w}+1) \times (d_{w}+d_{z}) $ matrix of population regression coefficients of $\pmb{\theta}_{i}$ on $\pmb{a}_{i}$, and $\pmb{\rho}_{\zeta\epsilon}$ is the population regression coefficient of $\zeta_{it}$ on $\pmb{\epsilon}_{i}$.

Our assumptions and (\ref{rc:4}) then imply that the conditional expectation of $y^{*}_{it}$ given $X_{i}$, $Z_{i}$, and $\pmb{a}_{i}$ is given by
\begin{align}\nonumber
\E(\textit{y}^{*}_{it}|X_{i}, Z_{i}, \pmb{a}_{i}) = \mathcal{X}_{it}^{\prime}\pmb{\varphi}+\mathcal{X}_{it}^{\prime}\pmb{\rho}_{\theta a}\pmb{a}_{i}  + \pmb{\rho}_{\zeta\epsilon}\pmb{\epsilon}_{i}.
\end{align}

Because the stochastic part, $\pmb{a}_{i}$, of the random coefficients in the reduced form equation are unobserved, the conditioning variable, $\pmb{\epsilon}_{i}  = X_{i} - Z_{i}(\pmb{\alpha} + \pmb{a}_{i})$, too, is not identified. To estimate the structural parameters, as in the model with random effects, we first integrate out $\pmb{a}_{i}$ from $\E(y^{*}_{it}|X_{i}, Z_{i}, \pmb{a}_{i})$ with respect to its conditional distribution, $f(\pmb{a}_{i}|X_{i}, Z_{i})$, to obtain 
\begin{align}\nonumber
\E(y^{*}_{t}|X_{i},
Z_{i})&= \int\E(y^{*}_{it}|X_{i},
Z_{i}, \pmb{a}_{i})f(\pmb{a}_{i}|X_{i},
Z_{i})d\pmb{a}_{i}\\\label{rc:5}
&=\mathcal{X}_{it}^{\prime}\pmb{\varphi}+\mathcal{X}_{it}^{\prime}\pmb{\rho}_{\theta a}\hat{\pmb{a}}_{i}  + \pmb{\rho}_{\zeta\epsilon}\hat{\pmb{\epsilon}}_{i},
\end{align}
where $\hat{\pmb{a}}_{i}=\E(\pmb{a}_{i}|X_{i}, Z_{i})$ and $\hat{\pmb{\epsilon}}_{i} =  X_{i} - Z_{i}(\pmb{\alpha} + \hat{\pmb{a}}_{i})$. In lemma 1 we show that

\vspace{0.5cm}

\textsc{Lemma} C1  \hspace{0.2cm}	\textit{If $ x_{it} = \pmb{z}^{\prime}_{it}\pmb{\alpha} + \pmb{z}^{\prime}_{it}\pmb{a}_{i} +\epsilon_{it} $, and if RC \ref{arc:1} and RC \ref{arc:3} hold, then
	\begin{align}\nonumber
	\E(\pmb{a}_{i}|X_{i}, Z_{i})=\hat{\pmb{a}_{i}}(X_{i}, Z_{i}, \Theta_{1} )=[\sum_{t=1}^{T}\pmb{z}_{it}\pmb{z}^{\prime}_{it}+\sigma^{2}_{\epsilon}\Sigma_{a}^{-1}]^{-1}\biggr(\sum_{t=1}^{T}\pmb{z}_{it}(x_{it} - \pmb{z}^{\prime}_{it}\pmb{\alpha})\biggr). 
	\end{align}}

\textsc{Proof of Lemma C1 } \hspace{0.2cm} Given in section C.1 of this appendix.  

\vspace{0.5cm}

From (\ref{rc:4}) and (\ref{rc:5}), it therefore follows that 
\begin{align}\label{rc:6}
\E(\pmb{\varphi}_{i}|X_{i},Z_{i}) = \E(\pmb{\varphi}+\pmb{\theta}_{i}|X_{i},Z_{i}) = \pmb{\varphi} +\pmb{\rho}_{\theta a}\hat{\pmb{a}}_{i} \text{  and  } \E(\zeta_{it}|X_{i},Z_{i}) = \pmb{\rho}_{\zeta\epsilon}\hat{\pmb{\epsilon}}_{i}.
\end{align} Writing $\pmb{\varphi}_{i}$ and $ \zeta_{it} $ in error form as $ \pmb{\varphi}_{i} = \pmb{\varphi} +\pmb{\rho}_{\theta a}\hat{\pmb{a}}_{i}  + \tilde{\pmb{\theta}}_{i} $ and $ \zeta_{it} = \pmb{\rho}_{\zeta\epsilon}\hat{\pmb{\epsilon}}_{i} +\tilde{\zeta}_{it} $ respectively, we can write the structural equation (\ref{rc:1}) as
\begin{align}
y_{it}=1\{y^{*}_{it} = \mathcal{X}_{it}^{\prime}\pmb{\varphi} +\mathcal{X}_{it}^{\prime}\pmb{\rho}_{\theta a}\hat{\pmb{a}}_{i}  + \pmb{\rho}_{\zeta\epsilon}\hat{\pmb{\epsilon}}_{i} + \mathcal{X}_{it}^{\prime}\tilde{\pmb{\theta}}_{i} +\tilde{\zeta}_{it}>0\}.
\end{align}

While $ \tilde{\pmb{\theta}}_{i}  $ and $\tilde{\zeta}_{it} $ are mean independent of $ X_{i} $ and $ Z_{i}$, if, as in \citet{chamberlain:1984}, we make a stronger assumption of complete independence and assume that $ \tilde{\pmb{\theta}}_{i}  $ and $\tilde{\zeta}_{it} $ are  distributed normally with mean zero and variances $ \Sigma_{\theta} $ and $ 1 $ respectively, then the  parameters, $ \Theta_{2} \equiv \{\pmb{\varphi}, \pmb{\rho}_{\theta a}, \pmb{\rho}_{\zeta\epsilon}, \Sigma_{\theta}\}$, of the above model can be estimated by integrated maximum likelihood method, where one can integrate out $ \tilde{\pmb{\theta}}_{i} $ using numerical multidimensional integration \citep[see][]{heiss:2008}. Alternatively,  maximum simulated likelihood or Markov Chain Monte Carlo (MCMC) methods as discussed in \citet{Greene:2004}, too, can be used to obtain $ \Theta_{2} $.

Once $ \Theta_{2} $ is estimated, the following measures of interest can be obtained. (A) The expected value, $ \E(\pmb{\varphi}_{i}|X_{i},Z_{i})= \pmb{\varphi} +\pmb{\rho}_{\theta a}\hat{\pmb{a}}_{i} $. (B)  The Average Partial Effect (APE) of changing a variable, say $w$, in time period $t$ from $w_{it}$ to $w_{it} + \Delta_{w}$ can be obtained as
\begin{align}\nonumber
\frac{\Delta G(\mathcal{X}_{it})}{\Delta_{w}}= \frac{G(\mathcal{X}_{it_{-w}},(w_{t}+\Delta_{w})) - G(\mathcal{X}_{it})}{\Delta_{w}}, \text{   where   }
\end{align}
\begin{align}\nonumber
G(\mathcal{X}_{it})=	\int\Phi\biggr(\frac{\mathcal{X}_{it}^{\prime}\pmb{\varphi} +\mathcal{X}_{it}^{\prime}\pmb{\rho}_{\theta a}\hat{\pmb{a}}_{i}  + \pmb{\rho}_{\zeta\epsilon}\hat{\pmb{\epsilon}}_{i}}{1+\mathcal{X}_{it}^{\prime}\Sigma_{\theta}\mathcal{X}_{it}} \biggr)dF(\hat{\pmb{a}},\hat{\epsilon}).
\end{align}
The above integral/partial mean rests on the assumption that analogous support requirement as in lemma 3 of the main text is satisfied. 

Since we were able to identify the average slopes coefficients and the APEs when we augmented the structural equation with $ \hat{\pmb{a}}^{\prime}_{i}\otimes\mathcal{X}^{\prime}_{it} $ and $\hat{\pmb{\epsilon}}_{i}$, as in the random effects case, we propose  $ (\hat{\pmb{\epsilon}}_{i}, \hat{\pmb{a}}_{i}) $ to be used as control function. 

%let  $ (\hat{\pmb{\epsilon}}_{i} ,\hat{\pmb{a}}_{i}) $ fully characterize the dependence of $ X_{i} $, $ Z_{i} $ on $ \zeta_{it} $ and $ \pmb{\theta}_{i} $ in the structural equation of the triangular system in (\ref{rc:3}):

\begin{asscf}\label{ascf:2}	
	\begin{align}\nonumber
	\zeta_{it},\pmb{\theta}_{i} | X_{i}, Z_{i}, \hat{\pmb{a}}_{i} &\sim \zeta_{it},\pmb{\theta}_{i} | \hat{\pmb{\epsilon}}_{i}, Z_{i} , \hat{\pmb{a}}_{i}\\\nonumber
	&\sim \zeta_{it},\pmb{\theta}_{i} | \hat{\pmb{\epsilon}}_{i}, \hat{\pmb{a}}_{i},
	\end{align}
	where $ \hat{\pmb{\epsilon}}_{i} \equiv (\hat{\epsilon}_{i1}, \ldots, \hat{\epsilon}_{iT})^{\prime} =X_{i} - Z_{i}(\pmb{\alpha}+\hat{\pmb{a}}_{i}) $ and  $ \hat{\pmb{a}}_{i} = \E(\pmb{a}_{i}| X_{i},Z_{i})  $.
\end{asscf}
In the above, $ \hat{\pmb{\epsilon}}_{i}( X_{i},Z_{i}) $ and $\hat{\pmb{a}}_{i}( X_{i},Z_{i}) $ are assumed to fully characterize the dependence of $ X_{i} $ and $ Z_{i} $ on the structural errors, $ \zeta_{it} $ and $ \pmb{\theta}_{i} $, in (\ref{rc:1}). With $ \hat{\pmb{\epsilon}}_{i} $ and $\hat{\pmb{a}}_{i}  $ as control functions the semiparametric method in \citet{hoderlein:2015} can be employed to estimate the mean of $\pmb{\varphi}_{i}$. 

\citet{kasy:2011} considers non-separable triangular systems for cross-sectional data to characterizes systems for which control functions -- control functions such as $ C(x,z) = x-\E(x|z) $ or  $ C(x,z)=F(x|z) $, where $ F $ is the conditional cumulative distribution function of $x$ given $z$ -- exist. \citeauthor{kasy:2011} shows that when unobserved heterogeneity in first-stage reduced form equations is multi-dimensional, such as the reduced form equations with random coefficients, the errors in the structural equation are not independent of the endogenous covariates, $ x $, or the instruments, $z$, given $ C(x,z) $. 

We consider panel data, where the random coefficient are time invariant, and our control functions, $ \hat{\pmb{\epsilon}}_{i} $ and $\hat{\pmb{a}}_{i} $, are different from those considered in \citeauthor{kasy:2011}. Since $\hat{\pmb{a}}_{i} $, a function of $ X_{i}$ and $ Z_{i} $, summarizes certain individual specific information, as argued in section 2.1, the assumption in ACF \ref{ascf:2} is that the dependence of $ (\pmb{\theta}_{i}, \zeta_{it}) $ on $(X_{i}, Z_{i})$ can be reduced to dependence of $ (\pmb{\theta}_{i}, \zeta_{it}) $ on  $(\hat{\pmb{a}}_{i},\hat{\pmb{\epsilon}}_{i}) $, which is akin to dependence assumption in papers such as by \citet{altonji:2005} and \citet{bester1:2009}. The assumption is motivated by the result that under the restrictions in RC \ref{arc:1}, RC \ref{arc:2}, and (\ref{rc:4}), the expectations of $ \zeta_{it}$ and $\pmb{\theta}_{i} $ given $ (X_{i}, Z_{i}) $ depend on $ (X_{i}, Z_{i} )$ only through $ \hat{\pmb{\epsilon}}_{i} $ and $\hat{\pmb{a}}_{i} $ respectively.       	

\subsection{Proof of Lemma C1 }

As in the models with random effects, to obtain $\hat{\pmb{a}}(X, Z )=\E(\pmb{a}|X, Z)$  we first derive $ f(\pmb{a}|X,
Z) $.   Again, using the fact that $ Z \independent \pmb{a} $, by an application of Bayes' rule, we have  $ f(\pmb{a}|X, Z)= \frac{f( X| Z, \pmb{a})f(\pmb{a}) }{f(X|Z)} $.

Since $ \pmb{a} \independent \epsilon_{t}  $, $ \pmb{a} \sim N(0,\Sigma_{a}) $, and $ \epsilon\sim N(0,\sigma^{2}_{\epsilon})$, it implies that $X$, given $Z$, is normally distributed with mean $ Z^{\prime}\pmb{\alpha}$, and variance $ \Sigma= \sigma^{2}_{\epsilon}I_{T} + Z^{\prime}\Sigma_{a}Z $, where $ I_{T} $ is an identity matrix of dimension $ T $. That is, 
\begin{align}\label{aeq:14}
f(X|Z) = \frac{1}{\sqrt{(2\pi)^{T}|\Sigma|}}\exp(-\frac{1}{2}R^{\prime}\Sigma^{-1}R), \text{  where $ R = X-Z^{\prime}\pmb{\alpha}$,} 
\end{align} 
and where by Woodbury matrix identity,
$ \Sigma^{-1} = \frac{1}{\sigma^{2}_{\epsilon}}I_{T} - \frac{1}{\sigma^{2}_{\epsilon}}Z^{\prime} [\sigma^{2}_{\epsilon}\Sigma_{a}^{-1}+ZZ^{\prime}]^{-1} Z $, and by Matrix determinant lemma, $ |\Sigma| =|\sigma^{2}_{\epsilon}\Sigma_{a}^{-1}+ZZ^{\prime}||\Sigma_{a}|$.

Since $ X$ given $ (Z, \pmb{a}) $ has the same distribution as $ \pmb{\epsilon} = R - Z^{\prime}\pmb{a}  $,  we have 	
\begin{align}\label{aeq:15}
f( X| Z, \pmb{a})f(\pmb{a}) = \frac{1}{\sqrt{(2\pi)^{T+k}\sigma^{2}_{\epsilon}|\Sigma_{a}|}}\exp\biggr(-\frac{1}{2\sigma^{2}_{\epsilon}}\biggr[[R-Z^{\prime}\pmb{a}]^{\prime}[R-Z^{\prime}\pmb{a}]+ \sigma^{2}_{\epsilon}\pmb{a}^{\prime}\Sigma_{a}^{-1}\pmb{a}\biggr]\biggr),   
\end{align}
where $ k $ is the dimension of $ \pmb{a} $. 

Since $ f(\pmb{a}|X, Z)= \frac{f( X| Z, \pmb{a})f(\pmb{a}) }{f(X|Z)} $, as shown earlier, using (\ref{aeq:14}) and (\ref{aeq:15}) it can be shown the $ \pmb{a} $ given $ X$ and  $Z $ is normally distributed with mean
\begin{align}\nonumber
\E( \pmb{a}|X, Z) = \hat{\pmb{a}}(X, Z )=[\sigma^{2}_{\epsilon}\Sigma_{a}^{-1}+ZZ^{\prime}]^{-1}Z(X - Z^{\prime}\pmb{\alpha}),
\end{align}	
and conditional variance $ \sigma^{2}_{\epsilon}[\sigma^{2}_{\epsilon}\Sigma_{a}^{-1}+ZZ^{\prime}]^{-1} $.

% example (otherwise, use just \section)
\renewcommand{\theequation}{D-\arabic{equation}}
% redefine the command that creates the equation no.
\setcounter{equation}{0}  % reset counter
%\section*{APPENDIX}  % use *-form to suppress numbering

\section{Asymptotic Covariance Matrix for Structural Parameters}
Though obtaining the parameters of the second stage, given the first stage consistent estimates $\hat{\Theta}_{1}$, is asymptotically equivalent to estimating the subsequent stage parameters had the true value of $\Theta_{1}^{*}$ been known, to obtain correct inference about the structural parameters, one has to account for the fact that instead of true values of first stage reduced form parameters, we use their estimated value. Here we are assuming that the first stage estimation involves the estimation of system of regression using \citeauthor{biorn:2004}'s method and that in the second stage a probit model, using the method of multivariate weighted nonlinear least squares (MWNLS), is estimated. 

\citet{newey:1984} has shown that sequential estimators can be interpreted as members of a class of Method of Moments (MM) estimators and that this
interpretation facilitates derivation of asymptotic covariance matrices for multi-step estimators. Let  $\Theta = (\Theta_{1}^{\prime}, \Theta_{2}^{\prime})^{\prime}$, where $\Theta_{1}$ and $\Theta_{2}$  are respectively the parameters to be estimated in the first and second step estimation of the sequential
estimator.  Following \citet{newey:1984} we write the  first and second step estimation as an MM estimation based on the
following population moment conditions:
\begin{align}\nonumber
\E(\mathcal{L}_{i\Theta_{1}})=E\frac{\partial \ln
	L_{i}(\Theta_{1})}{\partial\Theta_{1}}=0
\end{align}
\begin{align}\nonumber
\E(H_{i\Theta_{2}}(\Theta_{1},\Theta_{2} ))=0
\end{align}
and
where $L_{i}(\Theta_{1})$ is the likelihood function for individual $i$ for the first step system of reduced form equations and $\E(H_{i\Theta_{2}}(\Theta_{1},\Theta_{2} ))$ is the population moment condition for estimating $\Theta_{2}$ given $\Theta_{1}$.

%If in the second stage we are to use likelihood technique to estimate the second stage parameters $\Theta_{2}$ then  $\E(H_{i\Theta_{2}}(\Theta_{1},\Theta_{2} %))=\E(\mathcal{L}_{i2\Theta_{2}}(\Theta_{1},\Theta_{2}))=E\frac{\partial \ln L_{i2}(\Theta_{1},\Theta_{2} )}{\partial\Theta_{2}}=0
%$ where $L_{i2}(\Theta_{1},\Theta_{2})$ is the likelihood function for an individual for the second step estimation.

The estimates for $\Theta_{1}$ and $\Theta_{2}$ are obtained by solving the sample analog of the above population moment conditions. The sample analog of moment conditions for the first step estimation is given by
\begin{align}\nonumber
\frac{1}{N} \mathcal{L}_{\Theta_{1}}(\hat{\Theta}_{1}) =  \frac{1}{N}\sum_{i=1}^{N}\frac{\partial
	\mathcal{L}_{i}(\hat{\Theta}_{1})}{\partial\Theta_{1}} =\frac{1}{N}\sum_{i=1}^{N}\frac{\partial
	\ln L_{i}(\hat{\Theta}_{1})}{\partial\Theta_{1}}
\end{align}
where $\mathcal{L}_{i}(\Theta_{1})$ and the first order conditions with respect to 
$\Theta_{1}=(\pmb{\delta}^{\prime}, \textrm{vec}(\Lambda_{\alpha\alpha})^{\prime}, \textrm{vec}(\Sigma_{\epsilon\epsilon})^{\prime} )^{\prime}$\footnote{ While we have written our reduced form equation as 
	\begin{align}\nonumber
	\pmb{x}_{it} = \pi\pmb{z}_{it} + \bar{\pi}\bar{\pmb{z}}_{i}+\pmb{a}_{i} + \pmb{\epsilon}_{it}, \textrm{ \citeauthor{biorn:2004} writes it as } 	\pmb{x}_{it} =  Z^{\prime}_{it}\pmb{\delta} + \pmb{a}_{i} + \pmb{\epsilon}_{it},
	\end{align}
	where $ Z_{it} = \textrm{diag}((\pmb{z}^{\prime}_{it},\bar{\pmb{z}}_{i}^{\prime})^{\prime}, \ldots, (\pmb{z}^{\prime}_{it},\bar{\pmb{z}}_{i}^{\prime})^{\prime}  )  $ and $ \pmb{\delta} = (\textrm{vec}(\pi)^{\prime},\textrm{vec}(\bar{\pi})^{\prime})^{\prime} $. } are given in appendix E, and $N$ is the total number of individuals.

The sample analog of population moment condition for the second step estimation is given by
\begin{align}\nonumber
\frac{1}{N}H_{\Theta_{2}}(\hat{\Theta}_{1},\hat{\Theta}_{2})=\frac{1}{N}\sum_{i=1}^{N}H_{i\Theta_{2}}(\hat{\Theta}_{1},\hat{\Theta}_{2}).
\end{align}
We have shown that the structural equations augmented with the control functions $\hat{\pmb{\alpha}}_{i}(X_{i},Z_{i}, \Theta_{1})$  and $\hat{\pmb{\epsilon}}_{it}(X_{i},Z_{i}, \Theta_{1})$ leads to the identification of $\Theta_{2}$. Let $\Theta^{*}_{2}$ be the true values of $\Theta_{2}$. Under the assumptions we make, solving $\frac{1}{N}\sum^{N}_{i=1}H_{it\Theta_{2}}(\hat{\Theta}_{1}, \Theta_{2})=0$ is asymptotically equivalent to solving $\frac{1}{N}\sum^{N}_{i=1}H_{it\Theta_{2}}(\Theta^{*}_{1}, \Theta_{2})=0$, where $\hat{\Theta}_{1}$ is a consistent first step estimate of $\Theta_{1}$. Hence $\hat{\Theta}_{2}$ obtained by solving $\frac{1}{N}H_{\Theta_{2}}(\hat{\Theta}_{1},\hat{\Theta}_{2})=0$ is a consistent estimate of $\Theta_{2}$. %\citeauthor{newey:1984} has derived the asymptotic distribution of the second step estimates of a two step sequential estimator. 

To derive the asymptotic distribution of the second step estimates $\hat{\Theta}_{2}$, consider the stacked up sample moment conditions:
\begin{align}\label{eq:aeq31}
\frac{1}{N}\begin{bmatrix} \mathcal{L}_{\Theta_{1}}(\hat{\Theta}_{1}) \\
H_{\Theta_{2}}(\hat{\Theta}_{1}, \hat{\Theta}_{2}) \end{bmatrix}=0.
\end{align}
A series of Taylor's expansion of $\mathcal{L}_{\Theta_{1}}(\hat{\Theta}_{1})$, $H_{\Theta_{2}}(\hat{\Theta}_{1},\hat{\Theta}_{2})$ and around $\Theta^{*}$ gives
\begin{align}\label{eq:aeq32}
\frac{1}{N}\begin{bmatrix} \mathcal{L}_{\Theta_{1}\Theta_{1}} & 0 \\
H_{\Theta_{2} \Theta_{1}}&  H_{\Theta_{2} \Theta_{2}}   \end{bmatrix}\begin{bmatrix} \sqrt{N}(\hat{\Theta}_{1} -\Theta^{*}_{1}) \\
\sqrt{N}(\hat{\Theta}_{2} -\Theta^{*}_{2})  \end{bmatrix}= -\frac{1}{\sqrt{N}}\begin{bmatrix} \mathcal{L}_{\Theta_{1}} \\
H_{\Theta_{2}}. \end{bmatrix}
\end{align}
In matrix notation the above can be written as
\begin{align}\nonumber
B_{\Theta\Theta_{ N}}\sqrt{N}(\hat{\Theta} -\Theta) = -\frac{1}{\sqrt{N}}\Lambda_{\Theta_{ N}},
\end{align}
where $\Lambda_{\Theta_{ N}}$ is evaluated at $\Theta^{*}$ and $B_{\Theta\Theta_{ N}}$ is evaluated at points somewhere between $\hat{\Theta}$ and $\Theta^{*}$. Under the standard regularity conditions for Generalized Method of Moments (GMM) \citep[see][]{newey:1984} $B_{\Theta\Theta_{ N}}$
converges in probability to the lower block triangular matrix
$B_{*}= \lim \E (B_{\Theta\Theta_{ N}})$. $B_{*}$ is given by
\begin{gather*}
B_{*}= \begin{bmatrix} \mathbb{L}_{\Theta_{1}\Theta_{1}} & 0 \\
\mathbb{H}_{\Theta_{2} \Theta_{1}}&  \mathbb{H}_{\Theta_{2} \Theta_{2}} \end{bmatrix}
\end{gather*}
where $\mathbb{L}_{\Theta_{1}\Theta_{1}}=\E(\mathcal{L}_{i\Theta_{1}\Theta_{1}})$, $\mathbb{H}_{\Theta_{2}\Theta_{1}}=\E(H_{i\Theta_{2}\Theta_{1}})$. $\frac{1}{\sqrt{N}}\Lambda_{N}$
converges asymptotically in distribution to a normal random variable with mean zero and a covariance matrix $A_{*}= \lim \E\frac{1}{N}\Lambda_{N}\Lambda_{N}^{\prime}$, where $A_{*}$ is given by
\begin{gather*}
A_{*}=\begin{bmatrix} V_{LL} & V_{LH}  \\
V_{HL}&  V_{HH}   \end{bmatrix},
\end{gather*}
and a typical element of $A_{*}$, say $V_{LH}$, is given by $V_{LH}= \E[\mathcal{L}_{i\Theta_{1}}(\Theta_{1})H_{i\Theta_{2}}(\Theta_{1},\Theta_{2})^{\prime}]$. Under the regularity conditions $\sqrt{N}(\hat{\Theta} -\Theta^{*})$ is asymptotically normal with zero mean and covariance matrix given by $B_{*}^{-1}A_{*}B_{*}^{-1\prime}$, that is 
\begin{align}
\sqrt{N}(\hat{\Theta} -\Theta^{*})\stackrel{a}\sim \N[(0),(B_{*}^{-1}A_{*}B_{*}^{-1\prime}) ].
\end{align}

By an application of partitioned inverse formula and some matrix manipulation we get
the asymptotic covariance matrix of 
$ \sqrt{N}(\hat{\Theta}_{2} - \Theta_{2} ) $, $ V^{*}_{2}$ , where

\begin{align}\nonumber V^{\ast}_{2}=& \mathbb{H}^{-1}_{\Theta_{2}\Theta_{2}}V_{HH}\mathbb{H}^{-1}_{\Theta_{2}\Theta_{2}} + \mathbb{H}^{-1}_{\Theta_{2}\Theta_{2}}\mathbb{H}^{-1}_{\Theta_{2}\Theta_{1}}\{\mathbb{L}^{-1}_{\Theta_{1}\Theta_{1}}V_{LL}\mathbb{L}^{-1\prime}_{\Theta_{1}\Theta_{1}}\}\mathbb{H}^{-1\prime}_{\Theta_{2}\Theta_{1}}\mathbb{H}^{-1\prime}_{\Theta_{2}\Theta_{2}}\\\label{eq:aeq33}&-\mathbb{H}^{-1}_{\Theta_{2}\Theta_{2}}\{\mathbb{H}_{\Theta_{2}\Theta_{1}}\mathbb{L}^{-1}_{\Theta_{1}\Theta_{1}}V_{LH} + V_{HL}\mathbb{L}^{-1 \prime}_{\Theta_{1}\Theta_{1}}\mathbb{H}^{\prime}_{\Theta_{2}\Theta_{1}}\}\mathbb{H}^{-1 \prime}_{\Theta_{2}\Theta_{2}}\end{align}

To estimate $ V^{*}_{2} $, sample analog of the $B^{*}$, $B_{N}$ given in (\ref{eq:aeq32}), and sample analog of $A^{*}$,  $A_{N} = \frac{1}{N} \Lambda_{N}\Lambda_{N} $, have to be computed. A typical element of $ A_{N} $, say $  V_{LH_{N}} $ , is given by
$V_{LH_{N}} = \frac{1}{N} \sum^{N}_{i=1} \mathcal{L}_{i\Theta_{1}}(\Theta_{1})H_{i\Theta_{2}}(\hat{\Theta}_{1}, \hat{\Theta}_{2})^{\prime} $. The first and the second order conditions of the first-stage likelihood function for estimating  $ \Theta_{1} $, which are used to compute the sample analog of $\mathbb{L}_{\Theta_{1}\Theta_{1}}$ and to compute $A_{N}$, are provided in appendix E.

For binary response model, the score function pertaining to the minimand in equation (\ref{eq:14})  is given by
\begin{align}\nonumber
H_{i\Theta_{2}}(\Theta_{1}, \Theta_{2}) &= -\nabla_{\Theta_{2}}\mathbf{m}_{i}(X_{i}, Z_{i}, \Theta_{2})^{\prime}[\textbf{V}(X_{i}, Z_{i},\tilde{\Upsilon})]^{-1}[\mathbf{y}_{i} - \mathbf{m}_{i}(X_{i}, Z_{i}, \Theta_{2})]&\\\nonumber&= -\nabla_{\Theta_{2}}\mathbf{m}_{i}(\Theta_{1}, \Theta_{2})^{\prime}\tilde{\textbf{V}}^{-1}\mathbf{u}_{i},&
\end{align}
where $\mathbf{m}_{i}( \Theta_{1}, \Theta_{2})\equiv\mathbf{m}_{i}(X_{i}, Z_{i}, \Theta_{2})$ is a $T \time 1$ vector with $t^{th}$ element being $\textrm{\textbf{m}}(\mathbb{X}_{it},\Theta_{2})=\Phi(\pmb{x}_{it}^{\prime}\pmb{\varphi}  +\pmb{\rho}_{\alpha}\hat{\pmb{\alpha}}_{i}
+\pmb{\rho}_{\epsilon}\hat{\pmb{\epsilon}}_{it}) \equiv  \textrm{\textbf{m}}_{it}(\Theta_{1}, \Theta_{2})$, $ \mathbf{u}_{i} $ is a $T \time 1$ vector with $t^{th}$ element being $ y_{it} -\textrm{\textbf{m}}_{it}(\Theta_{1}, \Theta_{2}) $, and $\tilde{\textbf{V}} \equiv \textbf{V}(X_{i}, Z_{i},\tilde{\Upsilon})$. Now
\begin{align}\nonumber
&\nabla_{\Theta_{2}} \textrm{\textbf{m}}_{it}(\Theta_{1}, \Theta_{2}) = \phi(\mathbb{X}_{it}^{\prime}\Theta_{2})
\mathbb{X}^{\prime}_{it}
\end{align}
where $\mathbb{X}_{it}=(\pmb{x}_{it}^{\prime},\hat{\pmb{\alpha}}^{\prime}_{i}(\Theta_{1}), \hat{\pmb{\epsilon}}^{\prime}_{it}(\Theta_{1}) )^{\prime}$ and $\Theta_{2}=(\pmb{\varphi}^{\prime}, \pmb{\rho}_{\alpha}^{\prime}, \pmb{\rho}_{\epsilon}^{\prime})^{\prime}$.

\citet{wooldridge:2010} and  \citet{wooldridge:2010a}  show \citep[see Problem 12.11 in][]{wooldridge:2010a} that $\mathbb{H}_{\Theta_{2} \Theta_{2}}$ of $B^{*}$ is given by
\begin{align}\nonumber
\mathbb{H}_{\Theta_{2} \Theta_{2}} = \E[ H_{i\Theta_{2}\Theta_{2}}(\Theta_{1}, \Theta_{2})] = \E[\nabla_{\Theta_{2}}\mathbf{m}_{i}(\Theta_{1}, \Theta_{2})^{\prime}\tilde{\textbf{V}}^{-1}\nabla_{\Theta_{2}}\mathbf{m}_{i}(\Theta_{1}, \Theta_{2})],
\end{align}
which can be approximated as
\begin{align}\nonumber
\frac{1}{N}\sum_{i=1}^{N}\nabla_{\Theta_{2}}\mathbf{m}_{i}(\hat{\Theta}_{1}, \hat{\Theta}_{2})^{\prime}\hat{\textbf{V}}^{-1}\nabla_{\Theta_{2}}\mathbf{m}_{i}(\hat{\Theta}_{1}, \hat{\Theta}_{2}),&
\end{align}
where $\hat{\textbf{V}}=\textbf{V}(X_{i}, Z_{i}, \hat{\Upsilon}) =\textbf{V}(X_{i}, Z_{i}, \hat{\Theta}_{2}, \hat{\rho})$.

We now compute $ H_{\Theta_{2} \Theta_{1}}=\sum^{N}_{i=1}H_{i\Theta_{2} \Theta_{1}}=\sum^{N}_{i=1}\frac{\partial H_{i\Theta_{2}}(\Theta_{1}, \Theta_{2})}{\partial \Theta_{1}^{\prime}}$ in order to compute the sample analog of $ \mathbb{H}_{\Theta_{2} \Theta_{1}} $.  Now,
\begin{align}\nonumber
\frac{\partial H_{i\Theta_{2}}(\Theta_{1}, \Theta_{2})}{\partial \Theta_{1}^{\prime}} = -\biggr[&[\mathbf{u}_{i}^{\prime}\tilde{\textbf{V}}^{-1}\otimes I]\frac{\partial\textrm{vec}(\nabla_{\Theta_{2}}\mathbf{m}_{i}(\Theta_{1},\Theta_{2})^{\prime})}{\partial\Theta_{1}^{\prime}}&\\\nonumber
&+[\mathbf{u}_{i}\otimes\nabla_{\Theta_{2}}\mathbf{m}_{i}(\Theta_{1}, \Theta_{2})^{\prime}]\frac{\partial\textrm{vec}(\tilde{\textbf{V}}^{-1})}{\partial\Theta_{1}^{\prime}}&\\\nonumber
&-\nabla_{\Theta_{2}}\mathbf{m}_{i}(\Theta_{1}, \Theta_{2})^{\prime}\tilde{\textbf{V}}^{-1}\nabla_{\Theta_{1}}\mathbf{m}_{i}(\Theta_{1}, \Theta_{2})\biggr].&
\end{align}
Taking expectation of the above, we find that the first two terms are zero. Hence we have
\begin{align}\nonumber
\mathbb{H}_{\Theta_{2} \Theta_{1}} = \E[ H_{i\Theta_{2}\Theta_{1}}(\Theta_{1}, \Theta_{2})] = \E[\nabla_{\Theta_{2}}\mathbf{m}_{i}(\Theta_{1}, \Theta_{2})^{\prime}\tilde{\textbf{V}}^{-1}\nabla_{\Theta_{1}}\mathbf{m}_{i}(\Theta_{1}, \Theta_{2})],
\end{align}
which can be approximated by
\begin{align}\nonumber
\frac{1}{N}\sum_{i=1}^{N}\nabla_{\Theta_{2}}\mathbf{m}_{i}(\hat{\Theta}_{1}, \hat{\Theta}_{2})^{\prime}\hat{\textbf{V}}^{-1}\nabla_{\Theta_{1}}\mathbf{m}_{i}(\hat{\Theta}_{1}, \hat{\Theta}_{2}).&
\end{align}
The constituents, $\nabla_{\Theta_{1}} \textrm{\textbf{m}}_{it}(\Theta_{1}, \Theta_{2})$, of $\nabla_{\Theta_{1}} \textrm{\textbf{m}}_{i}(\Theta_{1}, \Theta_{2})$ are given by
\begin{align}\nonumber
&\nabla_{\Theta_{1}} \textrm{\textbf{m}}_{it}(\Theta_{1}, \Theta_{2}) = \phi(\mathbb{X}_{it}^{\prime}\Theta_{2})\Theta_{2}^{\prime}\frac{\partial\mathbb{X}_{it}}{\partial \Theta_{1}^{\prime}}, &
\end{align}
which is row matrix with dimension that of $\Theta_{1}$, and where
\begin{gather}\nonumber
\resizebox{.5\hsize}{!}{$\begin{Large}\frac{\partial\mathbb{X}_{it}}{\partial \Theta_{1}^{\prime}}= \begin{bmatrix}\frac{\partial\pmb{x}_{it}}{\partial \pmb{\delta}^{\prime}}& \frac{\partial\pmb{x}_{it}}{\partial \textrm{vec}(\Lambda_{\alpha\alpha})^{\prime}} &\frac{\partial\pmb{x}_{it}}{\partial \textrm{vec}(\Sigma_{\epsilon\epsilon})^{\prime}} \\
	\frac{\partial\hat{\pmb{\alpha}}_{i}}{\partial \pmb{\delta}^{\prime}}& \frac{\partial\hat{\pmb{\alpha}}_{i}}{\partial \textrm{vec}(\Lambda_{\alpha\alpha})^{\prime}} &\frac{\partial\hat{\pmb{\alpha}}_{i}}{\partial \textrm{vec}(\Sigma_{\epsilon\epsilon})^{\prime}} \\
	\frac{\partial \hat{\pmb{\epsilon}}_{it}}{\partial \pmb{\delta}^{\prime}}& \frac{\partial \hat{\pmb{\epsilon}}_{it}}{\partial \textrm{vec}(\Lambda_{\alpha\alpha})^{\prime}} &\frac{\partial \hat{\pmb{\epsilon}}_{it}}{\partial \textrm{vec}(\Sigma_{\epsilon\epsilon})^{\prime}}\end{bmatrix}\end{Large}$}.
\end{gather}

Since $\pmb{x}_{it}$ is not a function of $\Theta_{1}$,  $\frac{\partial\pmb{x}_{it}}{\partial \Theta_{1}^{\prime}}=\mathbf{0}_{\pmb{x}}$, where $\mathbf{0}_{\pmb{x}}$ is a null matrix with row dimension that of column vector $\pmb{x}_{it}$ and column dimension that of column vector $\Theta_{1}$. Using the following matrix results:
\begin{align}\nonumber
&\partial\textrm{vec}(\Omega\pmb{b}) = (\pmb{b}^{\prime}\otimes I_{m}) \partial\textrm{vec}(\Omega),\hspace{0.2cm}
\partial\textrm{vec}(\Omega^{-1}) = -(\Omega^{\prime-1} \otimes \Omega^{-1}) \partial\textrm{vec}(\Omega) \textrm{ and }\\\nonumber
& \frac{\partial\textrm{vec}(\Omega)}{\partial\textrm{vec}(\Omega)} = I_{mm},
\end{align}
where $\pmb{b}  $ is a vector of dimension $ m $, $ \Omega $ is a symmetric $ m\times m $ matrix and $I_{mm}$ is the $ mm\times mm $ identity matrix, it can be shown that
\begin{align}\nonumber
&\frac{\partial \hat{\pmb{\alpha}}_{i}}{\partial \pmb{\delta}^{\prime}}
=\frac{\partial (\textrm{diag}(\bar{\pmb{z}}_{i}^{\prime}, \ldots, \bar{\pmb{z}}_{i}^{\prime} )^{\prime}\textrm{vec}(\bar{\pi})+\hat{\pmb{a}}_{i})}{\partial \pmb{\delta}^{\prime}}
=\mathbb{O}_{Zi}^{\prime}-[T\Sigma_{\epsilon\epsilon}^{-1}+\Lambda_{\alpha\alpha}^{-1}]^{-1}\Sigma_{\epsilon\epsilon}^{-1}Z_{it}^{\prime},\\\nonumber
&\frac{\partial \hat{\pmb{\epsilon}}_{it}}{\partial \pmb{\delta}^{\prime}}=\frac{\partial (\pmb{x}_{it}-Z_{it}^{\prime}\pmb{\delta}-\hat{\pmb{a}}_{i})}{\partial \pmb{\delta}^{\prime}}
=-Z_{it}^{\prime}+[T\Sigma_{\epsilon\epsilon}^{-1}+\Lambda_{\alpha\alpha}^{-1}]^{-1}\Sigma_{\epsilon\epsilon}^{-1}Z_{it}^{\prime},
\end{align}
\begin{align}\nonumber
\frac{\partial \hat{\pmb{\alpha}}_{i}}{\partial \textrm{vec}(\Lambda_{\alpha\alpha})^{\prime}}
=-\biggr(\biggr(\sum_{t=1}^{T}\pmb{\upsilon}^{\prime}_{t}\biggr)\otimes I_{m}\biggr)\biggr[\biggr(\Sigma_{\epsilon\epsilon}^{-1}\otimes I_{m}\biggr)\biggr(\Sigma^{\prime}\otimes\Sigma\biggr)\biggr]I_{mm},&
\end{align}
\begin{align}\nonumber
\frac{\partial \hat{\pmb{\alpha}}_{i}}{\partial \textrm{vec}(\Sigma_{\epsilon\epsilon})^{\prime}}
=-\biggr(\biggr(\sum_{t=1}^{T}\pmb{\upsilon}^{\prime}_{t}\biggr)\otimes I_{m}\biggr)\biggr[\biggr( I_{m}\otimes\Sigma\biggr)\biggr(\Sigma_{\epsilon\epsilon}^{-1}\otimes \Sigma_{\epsilon\epsilon}^{-1}\biggr)+\biggr(\Sigma_{\epsilon\epsilon}^{-1}\otimes I_{m}\biggr)\biggr(\Sigma^{\prime}\otimes\Sigma\biggr)\biggr]TI_{mm},&
\end{align}
\begin{align}\nonumber
\frac{\partial \hat{\pmb{\epsilon}}_{it} }{\partial\textrm{vec}(\Lambda_{\alpha\alpha})^{\prime}}
=\frac{-\partial \hat{\pmb{\alpha}}_{i} }{\partial\textrm{vec}(\Lambda_{\alpha\alpha})^{\prime}}, \textrm{   and   } \frac{\partial \hat{\pmb{\epsilon}}_{it} }{\partial\textrm{vec}(\Sigma_{\epsilon\epsilon})^{\prime}}
=\frac{-\partial \hat{\pmb{\alpha}}_{i} }{\partial\textrm{vec}(\Sigma_{\epsilon\epsilon})^{\prime}},
\end{align}
where $\mathbb{O}_{Zi}=\textrm{diag}((0^{\prime}_{z},\bar{\pmb{z}}^{\prime}_{i})^{\prime}, \ldots,
(0^{\prime}_{z},\bar{\pmb{z}}^{\prime}_{i})^{\prime})$, $0_{z}$ denoting a vector of zeros with dimension that of $\pmb{z}_{it}$, $\pmb{\upsilon}_{t} = \pmb{x}_{t} - \pi \pmb{z}_{t}$, and $ \Sigma =[T\Sigma_{\epsilon\epsilon}^{-1}+\Lambda_{\alpha\alpha}^{-1}]^{-1} $. 

Since $H_{i\Theta_{2}\Theta_{1}}(\hat{\Theta}_{1}, \hat{\Theta}_{2})$ and  $H_{i\Theta_{2}\Theta_{2}}(\hat{\Theta}_{1}, \hat{\Theta}_{2})$ converge almost surly to $H_{i\Theta_{2}\Theta_{1}}(\Theta_{1}^{*}, \Theta_{2}^{*})$ and $H_{i\Theta_{2}\Theta_{2}}(\Theta_{1}^{*}, \Theta_{2}^{*})$ respectively, by the weak LLN $\frac{1}{N}\sum_{i=1}^{N} H_{i\Theta_{2}\Theta_{1}}(\hat{\Theta}_{1}, \hat{\Theta}_{2})$ will converge in probability to $\E(H_{i\Theta_{2}\Theta_{1}}(\Theta_{1}^{*}, \Theta_{2}^{*}))=\mathbb{H}_{\Theta_{2}\Theta_{1}}$ and $\frac{1}{N}\sum_{i=1}^{N} H_{i\Theta_{2}\Theta_{2}}(\hat{\Theta}_{1}, \hat{\Theta}_{2})$ will converge in probability to $\E(H_{i\Theta_{2}\Theta_{2}}(\Theta_{1}^{*}, \Theta_{2}^{*}))=\mathbb{H}_{\Theta_{2}\Theta_{2}}$.

\subsection{Hypothesis Testing of Average Partial Effects }

In section 2 of the paper we discussed the estimation of average structural function (ASF) and average partial effect (APE) of a variable $w$. If the support condition in lemma 3 for the point identification of ASF and the APEs is met, then if $ \eta_{t} $ in (2.9) of the main text is assumed to follow a normal distribution, the estimated APE of $w$ on the probability of $y_{it} = 1$ given $\pmb{x}_{it}=\bar{\pmb{x}}$ is given by
\begin{align}\nonumber
&\frac{\partial\widehat{\Pr}(y_{it} = 1|\bar{\pmb{x}} )}{\partial w}=\frac{1}{NT}\sum^{N}_{i=1}\sum_{t=1}^{T}\hat{\varphi}_{w}\phi\biggr(\bar{\mathbb{X}}^{\prime}_{it}\hat{\Theta}_{2}\biggr)\equiv \frac{1}{NT}\sum^{N}_{i=1}\sum_{t=1}^{T}g_{wit}(\hat{\Theta}_{2}),&
\end{align}
where $\bar{\mathbb{X}}_{it}=(\bar{\pmb{x}}^{\prime},\hat{\hat{\pmb{\alpha}}}_{i}(\hat{\Theta}_{1})^{\prime}, \hat{\hat{\pmb{\epsilon}}}_{it}(\hat{\Theta}_{1})^{\prime} )^{\prime}$ and $\hat{\Theta}_{2}=(\hat{\pmb{\varphi}}^{\prime}, \hat{\pmb{\rho}}_{\alpha}, \hat{\pmb{\rho}}_{\epsilon})^{\prime}$. 

To test various hypothesis in order to draw inferences about the APE's we need to compute the standard errors of their estimates.  Now, we know that by the linear approximation approach (delta method), the asymptotic variance of $\frac{\partial\widehat{\Pr}(y_{it} = 1|\bar{\pmb{x}} )}{\partial w}$ can be estimated by computing
\begin{align}\label{aeq:5}
\biggr[\frac{1}{NT}\sum^{N}_{i=1}\sum_{t=1}^{T}\frac{\partial g_{wit}(\hat{\Theta}_{2})}
{\partial\hat{\Theta}_{2}^{\prime}}\biggr]\hat{V}_{2}^{*}
\biggr[\frac{1}{NT}\sum^{N}_{i=1}\sum_{t=1}^{T}\frac{\partial g_{wit}(\hat{\Theta}_{2})}
{\partial\hat{\Theta}_{2}^{\prime}}\biggr]^{\prime},
\end{align}
where $\hat{V}_{2}^{*}$, the second stage error adjusted covariance matrix of $\Theta_{2}$ estimated at $\hat{\Theta}_{2}$, is given in (\ref{eq:aeq33}). $\frac{\partial g_{wit}(\hat{\Theta}_{2})}{\partial\hat{\Theta}_{2}^{\prime}}$ in (\ref{aeq:5}) turns out to be
\begin{align}\nonumber
\frac{\partial g_{wit}(\hat{\Theta}_{2})}{\partial\hat{\Theta}_{2}^{\prime}}=\phi(\bar{\mathbb{X}}^{\prime}_{it}
\hat{\Theta}_{2})[e_{w} -\hat{\varphi}_{w}(\bar{\mathbb{X}}^{\prime}_{it}
\hat{\Theta}_{2})\bar{\mathbb{X}}_{it}] 
\end{align}
where $e_{w}$ is a column vector having the dimension of $\Theta_{2}^{\prime}$ and with 1 at the position of $\varphi_{w}$ in $\Theta_{2}$ and zeros elsewhere.

If $w$ is a dummy variable then the estimated APE of $w$  when the APE of $w$ is point identified is given by
\begin{align}\nonumber
\Delta_{w}\Pr(y_{it}=1)&=\frac{1}{NT}\sum^{N}_{i=1}\sum_{t=1}^{T}\Phi(\bar{\pmb{x}}_{-w},w=1, \hat{\hat{\pmb{\alpha}}}_{i},\hat{\hat{\pmb{\epsilon}}}_{it})-\Phi(\bar{\pmb{x}}_{-w},w=0, \hat{\hat{\pmb{\alpha}}}_{i},\hat{\hat{\pmb{\epsilon}}}_{it})&\\\nonumber
&=\frac{1}{NT}\sum^{N}_{i=1}\sum_{t=1}^{T}\Phi_{it}(w=1)-\Phi_{it}(w=0)&\\\nonumber
&=\frac{1}{NT}\sum^{N}_{i=1}\sum_{t=1}^{T}\Delta_{w}\Phi_{it}().&
\end{align}
The asymptotic variance of the above can again by the application of delta method be obtained as
\begin{align}
\biggr[\frac{1}{NT}\sum^{N}_{i=1}\sum_{t=1}^{T}\frac{\partial\Delta\Phi_{it}(.)}{\partial\Theta_{2}^{\prime}}\biggr]\hat{V}_{2}^{*}
\biggr[\frac{1}{NT}\sum^{N}_{i=1}\sum_{t=1}^{T}\frac{\partial\Delta\Phi_{it}(.)}{\partial\Theta_{2}^{\prime}}\biggr]^{\prime},
\end{align}
where
\begin{align}\nonumber
\frac{\partial\Delta\Phi_{it}(.)}{\partial\Theta^{\prime}_{2}}&=\frac{\partial\Phi_{it}(w=1)}{\partial\Theta^{\prime}_{2}}-
\frac{\partial\Phi_{it}(w=0)}{\partial\Theta^{\prime}_{2}}&\\\nonumber&=\phi_{it}(w=1)\begin{bmatrix}\bar{\mathbb{X}}_{it_{-w}} \\1\end{bmatrix}^{\prime}-\phi_{it}(w=0)\begin{bmatrix}\bar{\mathbb{X}}_{it_{-w}} \\0\end{bmatrix}^{\prime}.
\end{align}

When the support condition in lemma 3 for the point identification of ASF and the APEs is not met, we compute the 95\% confidence interval ($ \text{CI}_{95\%} $) as proposed in \citet{imbens:2004} for the partially identified APEs. In section 2 of the main text, we have shown that when support condition in lemma 3 is not met, the the APE of changing  $ x_{k} $ from $ \bar{x}_{k} $ to $ \bar{x}_{k}+\Delta_{k} $, $ \Delta G(\bar{\pmb{x}})/\Delta_{k} $, lies in the interval,
\begin{align}\label{aeq:45}
\biggr[ \Psi_{l} =\frac{\tilde{G}(\bar{\pmb{x}}_{\Delta k}) - \tilde{G}(\bar{\pmb{x}}) - P(\bar{\pmb{x}})}{\Delta_{k}}, \hspace{.5cm} \Psi_{u} = \frac{\tilde{G}(\bar{\pmb{x}}_{\Delta k}) + P(\bar{\pmb{x}}_{\Delta k}) - \tilde{G}(\bar{\pmb{x}})}{\Delta_{k}}  \biggr], 
\end{align}
where  $ \bar{\pmb{x}}_{\Delta k} = (\bar{\pmb{x}}^{\prime}_{-k}, \bar{x}_{k}+\Delta_{k})^{\prime} $.

Let $ \sigma_{l} $ be the standard errors of the estimate of the lower bound of the interval and let $ \sigma_{u} $ be the standard errors of the estimate of the upper bound. To construct the confidence interval for the partially identified APEs, we first show that  $ \sigma_{u} = \sigma_{l} =\bar{\sigma} $. Let $ \widehat{\Psi}_{l} $ be the estimate of the lower bound and let $ \widehat{\Psi}_{u} $ be that of the upper bound. Since the estimates of $ P(.) $ in (\ref{aeq:45}) does not depend on $ \Theta_{2} $, 
\begin{align}\nonumber
\frac{\partial \widehat{\Psi}_{l}}{\partial\Theta^{\prime}_{2}} = \frac{\partial \widehat{\Psi}_{u}}{\partial\Theta^{\prime}_{2}} = \biggr[\frac{\partial \hat{\tilde{G}}(\bar{\pmb{x}}_{\Delta k})}{\partial\Theta^{\prime}_{2}} -  \frac{\partial \hat{\tilde{G}}(\bar{\pmb{x}}) }{\partial\Theta^{\prime}_{2}}\biggr]\frac{1}{\Delta_{k}}. 
\end{align}   
By applying by the delta method, we get 
\begin{align}\label{aeq:46}
\bar{\sigma}^{2}=\frac{1}{\Delta_{k}}\biggr[\frac{\partial \hat{\tilde{G}}(\bar{\pmb{x}}_{\Delta k})}{\partial\Theta^{\prime}_{2}} -  \frac{\partial \hat{\tilde{G}}(\bar{\pmb{x}}) }{\partial\Theta^{\prime}_{2}}\biggr]\hat{V}_{2}^{*}\frac{1}{\Delta_{k}}\biggr[\frac{\partial \hat{\tilde{G}}(\bar{\pmb{x}}_{\Delta k})}{\partial\Theta^{\prime}_{2}} -  \frac{\partial \hat{\tilde{G}}(\bar{\pmb{x}}) }{\partial\Theta^{\prime}_{2}}\biggr]^{\prime},
\end{align}
where in (\ref{aeq:46}) the derivative of $ \hat{\tilde{G}}(.) $ with respect to $\Theta_{2}$  at $ \pmb{x} $ is given by  
\begin{align}\nonumber
\frac{\partial \hat{\tilde{G}}(\pmb{x}) }{\partial\Theta^{\prime}_{2}}=\frac{1}{NT}\sum_{i,t}\phi(\pmb{x}^{\prime}\hat{\pmb{\varphi}}  +\hat{\pmb{\rho}}_{\alpha}\hat{\hat{\pmb{\alpha}}}_{i}
+\hat{\pmb{\rho}}_{\epsilon}\hat{\hat{\pmb{\epsilon}}}_{it})1[(\hat{\hat{\pmb{\alpha}}}_{i},
\hat{\hat{\pmb{\epsilon}}}_{it}) \in \hat{\mathcal{A}}(\pmb{x})](\pmb{x}^{\prime},\hat{\hat{\pmb{\alpha}}}_{i}^{\prime}, \hat{\hat{\pmb{\epsilon}}}_{it}^{\prime} ).
\end{align}

According to lemma 4.1 in \cite{imbens:2004}, the confidence interval 
\begin{align}\nonumber
\text{CI}_{95\%} = \biggr[\Psi_{l} - C_{NT}\frac{\bar{\sigma}}{\sqrt{NT}}, \Psi_{u} + C_{NT}\frac{\bar{\sigma}}{\sqrt{NT}} \biggr],
\end{align}
where $ C_{NT} $ is a solution to 
\begin{align}\nonumber
\Phi\biggr(C_{NT} + \sqrt{NT}\frac{(\Psi_{u} -\Psi_{l})}{\bar{\sigma}} \biggr)-\Phi\biggr(-C_{NT}  \biggr) = 0.95,
\end{align}
achieves a uniform coverage rate of at least 95\%.

\renewcommand{\theequation}{E-\arabic{equation}}
% redefine the command that creates the equation no.
\setcounter{equation}{0}  % reset counter

\section{Estimation of the Reduced form Equations}
In this section we briefly describe \citet{biorn:2004} step wise maximum likelihood procedure to estimate the reduced form system of equation
\begin{align}
\pmb{x}_{it} = Z^{\prime}_{it}\pmb{\delta} + \pmb{a}_{i} + \pmb{\epsilon}_{it},
\end{align}
where $ Z_{it}= \textrm{diag}((\pmb{z}^{\prime}_{it},\bar{\pmb{z}}_{i}^{\prime})^{\prime}, \ldots, (\pmb{z}^{\prime}_{it},\bar{\pmb{z}}_{i}^{\prime})^{\prime}  ) $ and $ \pmb{\delta} = (\textrm{vec}(\pi)^{\prime},\textrm{vec}(\bar{\pi})^{\prime})^{\prime} $. While \citet{biorn:2004} deals with unbalanced panel, here we assume that our panel is balanced.  Let $N$ be the total number of individuals. Let $\mathcal{N}$ be the total number of observations, i.e., $\mathcal{N}= NT$. Let $\pmb{x}_{i(T)} = (\pmb{x}_{i1}^{\prime},\ldots \pmb{x}_{iT}^{\prime})^{\prime}$, $Z_{i(T)} = (Z_{i1}^{\prime},\ldots Z_{iT}^{\prime})^{\prime}$ and $\pmb{\epsilon}_{i(T)} = (\pmb{\epsilon}_{i1}^{\prime},\ldots \pmb{\epsilon}_{iT}^{\prime})^{\prime}$ and write the model as
\begin{align}\label{eq:e1}
\pmb{x}_{i(T)} = Z^{\prime}_{i(T)}\pmb{\delta} +
(e_{T}\otimes\pmb{a}_{i}) + \pmb{\epsilon}_{i(T)}=Z^{\prime}_{i(T)}\pmb{\delta} + \pmb{u}_{i(T)},
\end{align}
Now,
\begin{align}\nonumber
\E(\pmb{u}_{i(T)}\pmb{u}_{i(T)}^{\prime}) = I_{T}\otimes\Sigma_{\epsilon\epsilon} + E_{T}\otimes\Lambda_{\alpha\alpha}  = K_{T}\otimes\Sigma_{\epsilon\epsilon} + J_{T}\otimes\Sigma_{(T)} =\Omega_{u(T)}
\end{align}
where
\begin{align}\nonumber
\Sigma_{(T)} =  \Sigma_{\epsilon\epsilon}+ T\Lambda_{\alpha\alpha},
\end{align}
where $I_{T}$ is the $T$ dimensional identity matrix, $e_{T}$ is the $(T \times 1)$ vector of ones, $E_{T}=e_{T}e_{T}^{\prime}$,
$J_{T}=(1/T)E_{T}$, and $K_{T}=I_{T}-J_{T}$. The latter two matrices are symmetric and idempotent
and have orthogonal columns, which facilitates inversion of $\Omega_{u(T)}$.

\subsection{GLS estimation}
Before addressing the maximum likelihood problem, we consider the GLS problem for $\pmb{\delta}$ when $\Lambda_{\alpha}$ and $\Sigma_{\epsilon\epsilon}$ are known. Define $Q_{i(T)} = \pmb{u}_{i(T)}^{\prime}\Omega_{u(T)}^{-1}\pmb{u}_{i(T)}$,  then GLS estimation is the problem of minimizing
$Q = \sum_{i=1}^{N}Q_{i(T)}$ with respect to $\pmb{\delta}$.
Since $\Omega_{u(T)}^{-1} = K_{T}\otimes\Sigma_{\epsilon\epsilon}^{-1} + J_{T}\otimes(\Sigma_{\epsilon\epsilon}+ T\Lambda_{\alpha\alpha})^{-1}$, we can rewrite $Q$ as
\begin{align}\nonumber
Q = \sum_{i=1}^{N}\pmb{u}_{i(T)}^{\prime}[K_{T}\otimes\Sigma_{\epsilon\epsilon}^{-1}]\pmb{u}_{i(T)} + \sum_{i=1}^{N}\pmb{u}_{i(T)}^{\prime}[J_{T}\otimes(\Sigma_{\epsilon\epsilon}+ T\Lambda_{\alpha\alpha})^{-1}]\pmb{u}_{i(T)}.
\end{align}
GLS estimator of $\pmb{\delta}$ when $\Lambda_{\alpha\alpha}$ and $\Sigma_{\epsilon\epsilon}$ are known is obtained from $\partial Q /\partial\pmb{\delta}=0$, and is given by
\begin{align}\nonumber
\hat{\pmb{\delta}}_{GLS}=\biggr[\sum_{i=1}^{N}Z_{i(T)}^{\prime}[K_{T}\otimes\Sigma_{\epsilon\epsilon}^{-1}]Z_{i(T)} + \sum_{i=1}^{N}Z_{i(T)}^{\prime}[J_{T}\otimes(\Sigma_{\epsilon\epsilon}+ T\Lambda_{\alpha\alpha})^{-1}]Z_{i(T)}\biggr]^{-1}\times\\\label{eq:e2}\biggr[\sum_{i=1}^{N}Z_{i(T)}^{\prime}[K_{T}\otimes\Sigma_{\epsilon\epsilon}^{-1}]\pmb{x}_{i(T)} + \sum_{i=1}^{N}Z_{i(T)}^{\prime}[J_{T}\otimes(\Sigma_{\epsilon\epsilon}+ T\Lambda_{\alpha\alpha})^{-1}]\pmb{x}_{i(T)}\biggr].
\end{align}

\subsection{Maximum Likelihood Estimation}
Now consider ML estimation of $\pmb{\delta}$, $\Sigma_{\epsilon\epsilon}$, and
$\Lambda_{\alpha\alpha}$. Assuming normality of
the individual effects and the disturbances, i.e., $\pmb{a}_{i}\sim \textrm{IIN}(0,\Lambda_{\alpha\alpha})$ and
$\pmb{\epsilon}_{it} \sim \textrm{IIN}(0,\Sigma_{\epsilon\epsilon})$, then $\pmb{u}_{i(T)}=(e_{T}\otimes\pmb{a}_{i}) + \pmb{\epsilon}_{i(T)}\sim \textrm{IIN}(0_{mT,1}, \Omega_{u(T)})$. The log-likelihood functions of all $\pmb{x}$'s conditional on all $\textit{\textbf{Z}}$'s for an individual and for all individuals in the data set then become, respectively,

\begin{align}
\mathcal{L}_{i}=\frac{-mT}{2}\ln(2\pi) - \frac{1}{2}\ln|\Omega_{u(T)}| - \frac{1}{2}Q_{i(T)}(\pmb{\delta},\Sigma_{\epsilon\epsilon},\Lambda_{\alpha\alpha}),
\end{align}
\begin{align}\label{eq:e3}
\mathcal{L}=\sum_{i=1}^{N}\mathcal{L}_{i}=\frac{-mNT}{2}\ln(2\pi) - \frac{1}{2}N\ln|\Omega_{u(T)}| - \frac{1}{2}\sum_{i=1}^{N}Q_{i(T)}(\pmb{\delta},\Sigma_{\epsilon\epsilon},\Lambda_{\alpha\alpha}),
\end{align}
where
\begin{align}\nonumber
Q_{i(T)}(\pmb{\delta},\Sigma_{\epsilon\epsilon},\Lambda_{\alpha\alpha})=[\pmb{x}_{i(T)}-Z_{i(T)}^{\prime}\pmb{\delta}]^{\prime}[K_{T}\otimes\Sigma_{\epsilon\epsilon}^{-1} + J_{T}\otimes(\Sigma_{\epsilon\epsilon}+ p\Lambda_{\alpha\alpha})^{-1}][\pmb{x}_{i(T)}-Z_{i(T)}^{\prime}\pmb{\delta}],
\end{align}
and $|\Omega_{u(T)}|= |\Sigma_{(T)}||\Sigma_{\epsilon\epsilon}|^{T-1}$.

\citeauthor{biorn:2004} splits the problem  of estimation into: (A) \textit{Maximization of
	$\mathcal{L}$ with respect to $\pmb{\delta}$ for given $\Sigma_{\epsilon\epsilon}$ and $\Lambda_{\alpha\alpha}$} and (B) \textit{Maximization of $\mathcal{L}$ with respect
	to $\Sigma_{\epsilon\epsilon}$ and $\Lambda_{\alpha\alpha}$ for given $\pmb{\delta}$}. \textit{Subproblem} (A) is identical with the GLS problem, since maximization
of $\mathcal{L}$ with respect to $\pmb{\delta}$ for given $\Sigma_{\epsilon\epsilon}$ and $\Lambda_{\alpha\alpha}$ is equivalent to minimization of $\sum_{i}^{N}Q_{i(T)}(\pmb{\delta},\Sigma_{\epsilon\epsilon},\Lambda_{\alpha\alpha})$, which gives (\ref{eq:e2}). To solve \textit{subproblem}(B) \citeauthor{biorn:2004} derives expressions for the derivatives of both $\mathcal{L}_{i}$  and $\mathcal{L}$ with respect to $\Sigma_{\epsilon\epsilon}$ and $\Lambda_{\alpha\alpha}$. The complete stepwise algorithm for solving jointly subproblems (A) and (B) then consists in switching between (\ref{eq:e2}) and minimizing (\ref{eq:e3})  with respect to $\Sigma_{\epsilon\epsilon}$ and $\Lambda_{\alpha\alpha}$ to obtain $\Sigma_{\epsilon\epsilon}$ and $\Lambda_{\alpha\alpha}$ and iterating until convergence.

The first order conditions for the log-likelihood function for an individual $i$ with respect to $ \pmb{\delta} $, $\textrm{vech}(\Sigma_{\epsilon\epsilon})$ and $\textrm{vech}(\Lambda_{\alpha\alpha})$ are:
\begin{align}\nonumber
\frac{\partial \mathcal{L}_{i}}{\partial\pmb{\delta} } = [\pmb{x}_{i(T)} - Z^{\prime}_{i(T)}\pmb{\delta}]^{\prime}[K_{T} \otimes \Sigma^{-1}_{\epsilon\epsilon} + J_{T} \otimes (\Sigma_{\epsilon\epsilon} + p\Sigma_{\alpha\alpha})^{-1}]Z^{\prime}_{i(T)},
\end{align}
\begin{align}\nonumber
\frac{\partial\mathcal{L}_{i}}{\partial \textrm{vech}(\Sigma_{\epsilon\epsilon})}&= -\frac{1}{2}L_{m}\textrm{vec}\biggr[\Sigma_{(T)}^{-1}+(T-1)\Sigma_{\epsilon\epsilon}^{-1} -\Sigma_{(T)}^{-1}B_{ui(T)}\Sigma_{(T)}^{-1}
-\Sigma_{\epsilon\epsilon}^{-1}W_{ui(T)}\Sigma_{\epsilon\epsilon}^{-1}\biggr],&
\end{align}
and
\begin{align}\nonumber
\frac{\partial\mathcal{L}_{i}}{\partial \textrm{vech}(\Lambda_{\alpha\alpha})}&= -\frac{1}{2}L_{m}\textrm{vec}\biggr[T\Sigma_{(T)}^{-1} -T\Sigma_{(T)}^{-1}B_{ui(T)}\Sigma_{(T)}^{-1}\biggr],&
\end{align}
where $\textrm{vech}(\Sigma_{\epsilon\epsilon})$ and $\textrm{vech}(\Lambda_{\alpha\alpha})$ are column-wise vectorization of the lower triangle of the symmetric matrix $\Sigma_{\epsilon\epsilon}$ and $\Lambda_{\alpha\alpha}$, and $L_{m}$ is an elimination matrix. $W_{ui(T)}$ and $B_{ui(T)}$ respectively are defined as follows
\begin{align}\nonumber
W_{ui(T)} =\tilde{E}_{i(T)}K_{T}\tilde{E}_{i(T)}^{\prime}  \textrm{   and   }  B_{ui(T)} =\tilde{E}_{i(T)}J
_{T}\tilde{E}_{i(T)}^{\prime},
\end{align}
where $\tilde{E}_{i(T)} = [\pmb{u}_{i1}, \ldots, \pmb{u}_{iT}]$ is a $(m \times T)$ matrix and $\pmb{u}_{i(T)} = \textrm{vec}(E_{i(T)})$, $\textrm{`vec'}$ being the vectorization
operator. That is, the disturbances defined in (\ref{eq:e1}) for an individual $i$ has been arranged in $(m \times T)$ matrix, $\tilde{E}_{i(T)}$.

The second order conditions are:
\begin{align}\nonumber
\frac{\partial^{2}\mathcal{L}_{i}}{\partial\pmb{\delta}\partial\pmb{\delta}^{\prime}} = -Z_{i(T)}[K_{T} \otimes \Sigma^{-1}_{\epsilon\epsilon} + J_{T} \otimes (\Sigma_{\epsilon\epsilon} + p\Sigma_{\alpha\alpha})^{-1}]Z_{i(T)}^{\prime}
\end{align}
\begin{align}\nonumber
\frac{\partial^{2}\mathcal{L}_{i}}{\partial\pmb{\delta}\partial\textrm{vec}(\Lambda_{\alpha\alpha})^{\prime}} = -T(\pmb{u}_{i(T)}\otimes Z_{i(T)})(I_{T}  K_{m,T} \otimes I_{m})(\textrm{vec}(J_{T} ) \otimes \Sigma^{-1}_{(T)} \otimes \Sigma^{-1}_{(T)})
\end{align}
\begin{align}\nonumber
\frac{\partial^{2}\mathcal{L}_{i}}{\partial\pmb{\delta}\partial\textrm{vec}(\Sigma_{\epsilon\epsilon})^{\prime}} = -(\pmb{u}_{i(T)} \otimes Z_{i(T)})(I_{T} \otimes K_{m,T} \otimes I_{m})(\textrm{vec}(K_{T} ) \otimes \Sigma_{\epsilon\epsilon}^{-1}
\otimes \Sigma_{\epsilon\epsilon}^{-1} + \textrm{vec}(J_{T} ) \otimes \Sigma^{-1}_{(T)} \otimes \Sigma^{-1}_{(T)})
\end{align}
\begin{align}\nonumber
\frac{\partial^{2}\mathcal{L}_{i}}{\partial\textrm{vec}(\Lambda_{\alpha\alpha})\partial\pmb{\delta}^{\prime}}  = -\frac{T}{2}
(\Sigma^{-1}_{(T)} \otimes \Sigma^{-1}_{(T)})[(\tilde{E}_{i(T)}J_{T} \otimes I_{m}) + (I_{m} \otimes \tilde{E}_{i(T)}J_{T} )K_{m,T} ]Z^{\prime}_{i(T)}
\end{align}
\begin{align}\nonumber
\frac{\partial^{2}\mathcal{L}_{i}}{\partial\textrm{vec}(\Lambda_{\alpha\alpha})\partial\textrm{vec}(\Lambda_{\alpha\alpha})^{\prime}} =
\frac{T^{2}}{2}[(\Sigma^{-1}_{(T)} \otimes \Sigma^{-1}_{(T)}) - \Sigma^{-1}_{(T)}B_{ui(T)}\Sigma^{-1}_{(T)} \otimes \Sigma^{-1}_{(T)} - \Sigma^{-1}_{(T)} \otimes \Sigma^{-1}_{(T)}B_{ui(T)}\Sigma^{-1}_{(T)}]
\end{align}

\begin{align}\nonumber
\frac{\partial^{2}\mathcal{L}_{i}}{\partial\textrm{vec}(\Lambda_{\alpha\alpha})\partial\textrm{vec}(\Sigma_{\epsilon\epsilon})^{\prime}} =
\frac{T}{2}
[(\Sigma^{-1}_{(T)} \otimes \Sigma^{-1}_{(T)}) - \Sigma^{-1}_{(T)}B_{ui(T)}\Sigma^{-1}_{(T)} \otimes \Sigma^{-1}_{(T)} - \Sigma^{-1}_{(T)} \otimes\Sigma^{-1}_{(T)} B_{ui(T)}\Sigma^{-1}_{(T)}]
\end{align}

\begin{align}\nonumber
\frac{\partial^{2}\mathcal{L}_{i}}{\partial\textrm{vec}(\Sigma_{\epsilon\epsilon})\partial\pmb{\delta}^{\prime}} =& -
\frac{1}{2}
(\Sigma^{-1}_{(T)} \otimes \Sigma^{-1}_{(T)})[(\tilde{E}_{i(T)}J_{T} \otimes I_{m}) + (I_{m} \otimes \tilde{E}_{i(T)}J_{T} )K_{m,T} ]Z^{\prime}_{i(T)}\\\nonumber
&-\frac{1}{2}(\Sigma^{-1}_{\epsilon\epsilon} \otimes \Sigma^{-1}_{\epsilon\epsilon} )[(\tilde{E}_{i(T)}K_{T} \otimes I_{m}) + (I_{m} \otimes \tilde{E}_{i(T)}K_{T} )K_{m,T} ]Z^{\prime}_{i(T)}
\end{align}

\begin{align}\nonumber
\frac{\partial^{2}\mathcal{L}_{i}}{\partial\textrm{vec}(\Sigma_{\epsilon\epsilon})\partial\textrm{vec}(\Lambda_{\alpha\alpha})^{\prime}}
=
\frac{T}{2}
[(\Sigma^{-1}_{(T)} \otimes \Sigma^{-1}_{(T)}) - \Sigma^{-1}_{(T)}B_{ui(T)}\Sigma^{-1}_{(T)} \otimes \Sigma^{-1}_{(T)} - \Sigma^{-1}_{(T)} \otimes \Sigma^{-1}_{(T)}B_{ui(T)}\Sigma^{-1}_{(T)}]
\end{align}

\begin{align}\nonumber
\frac{\partial^{2}\mathcal{L}_{i}}{\partial\textrm{vec}(\Sigma_{\epsilon\epsilon})\partial\textrm{vec}(\Sigma_{\epsilon\epsilon})^{\prime}} = &
\frac{1}{2}[\Sigma^{-1}_{(T)}\otimes \Sigma^{-1}_{(T)} + (T - 1)\Sigma_{\epsilon\epsilon}\otimes\Sigma_{\epsilon\epsilon} - \Sigma_{\epsilon\epsilon}B_{ui(T)}\Sigma_{(T)} \otimes \Sigma^{-1}_{(T)} \\\nonumber &- \Sigma^{-1}_{(T)} \otimes \Sigma^{-1}_{(T)}B_{ui(T)} - \Sigma^{-1}_{\epsilon\epsilon}W_{ui(T)}\Sigma^{-1}_{\epsilon\epsilon} \otimes \Sigma^{-1}_{\epsilon\epsilon}- \Sigma^{-1}_{\epsilon\epsilon} \otimes \Sigma^{-1}_{\epsilon\epsilon} W_{ui(T)}\Sigma^{-1}_{(T)}].
\end{align}

%%%%%%%%%%%%%%%%%%%%%%%%%%%%%%%%%%%%%%%%%%%%%
\end{document}